\documentclass[12pt,oneside,reqno]{article}
\makeatletter
\newcommand{\singlespacing}{\let\CS=\@currsize\renewcommand{\baselinestretch}{1}\tiny\CS}
 \usepackage{a4wide}
 \usepackage[a4paper, total={6.4in, 9.5in}]{geometry}
\usepackage{epsfig}
\usepackage{lineno}
\usepackage{rotating}
\usepackage{booktabs}
\usepackage{enumitem}
\newlist{steps}{enumerate}{1}
\setlist[steps, 1]{label = Step \arabic*:}
\usepackage{comment}
\usepackage{multirow}
\usepackage{siunitx}
\usepackage{lscape}
\usepackage{rotating}
\usepackage[T1]{fontenc}
\usepackage[utf8]{inputenc}
\usepackage{amsmath}
\usepackage{amsthm}
\usepackage{amsfonts,graphicx}
\usepackage{authblk}
\usepackage{verbatim}
\usepackage[mathscr]{eucal}
\usepackage{booktabs}
\usepackage{siunitx}

\usepackage{mathptmx} 
\usepackage{adjustbox}
\usepackage{xcolor}
\usepackage{xspace}
\usepackage{enumitem}
\usepackage{subfig}
\usepackage[colorlinks=true,urlcolor=black,citecolor=black,linkcolor=black,bookmarks=true]{hyperref}
\usepackage[sort & compress,comma,authoryear]{natbib}
\usepackage{float}
\usepackage{epstopdf}
\usepackage[parfill]{parskip}

\newcommand{\RN}[1]{%
	\textup{\uppercase\expandafter{\romannumeral#1}}%
}

\newcommand{\be}{\begin{equation}}
\newcommand{\ee}{\end{equation}}
\newcommand{\beanno}{\begin{eqnarray*}}

\newcommand{\eeanno}{\end{eqnarray*}}
\newcommand{\bea}{\begin{eqnarray}}
\newcommand{\eea}{\end{eqnarray}}
\newcommand{\ba}{\begin{array}}
\newcommand{\ea}{\end{array}}

\newcommand{\bc}{\begin{center}}
\newcommand{\ec}{\end{center}}

\newcommand{\ncom}{\newcommand}

\ncom{\Th}{\Theta}
\ncom{\e}{\xi}
\ncom{\sgn}{{\rm sgn\,}}
\ncom{\integ}[4]{\int_{#1}^{#2}\,{#3}\,d{#4}}
\ncom{\vspan}[1]{{{\rm\,span}\{ #1 \}}}
\ncom{\dm}[1]{ {\displaystyle{#1} } }
\ncom{\ri}[1]{{#1} \index{#1}}
\newtheorem{theorem}{\bf Theorem}[section]

\newtheorem{corollary}{Corollary}[section]

\newtheoremstyle
{remarkstyle}
{}
{11pt}
{}
{}
{\bfseries}
{:}
{     }
{\thmname{#1} \thmnumber{#2} }
\theoremstyle{remarkstyle}

\renewcommand{\baselinestretch}{1.5}
\def \x {x}

\def \a {a}
\def \b {b}
\def \th {\theta}

\def \f {f}
\def \F {F}
\def \ss {n}
\def \i {i}
\def \fx {\phi(\x;\th)}

\def \pv {\Theta}
\allowdisplaybreaks
\begin{document}
	\title{\bf Copula-Based Bivariate Kumaraswamy–Teissier Distributions: Modeling Temperature–Rainfall Dependence and Compound Extremes}
 
\author{ Kamana Mishra$^1$, Tanmay Kayal$^1$ and Sarita Azad$^1$ \thanks{Corresponding Author Email: sarita@iitmandi.ac.in}}
\affil{\it $^1$School of Mathematical and Statistical Sciences \\ \it Indian Institute of Technology Mandi, India  \\ }
	\date{}
	\maketitle
	\vspace{-1.5em}
\begin{abstract}
\noindent This study proposes two novel bivariate distributions for jointly modeling temperature and rainfall by integrating Kumaraswamy–Teissier marginals with Clayton and Gumbel copula structures. To capture a wide range of dependence patterns, including both positive and negative associations, rotated copula variants (90°, 180°, and 270°) are incorporated along with their corresponding tail dependence characteristics. Model parameters are estimated using maximum likelihood and the inference functions for margins (IFM) approach, and their finite-sample performance is assessed through a comprehensive Monte Carlo simulation study. The proposed models are applied to monthly gridded temperature and rainfall data from the Northwest Himalayas, a region characterized by complex hydro-climatic variability. Comparative analysis demonstrates that the proposed framework outperforms several existing bivariate models and effectively captures lower-tail, upper-tail, and asymmetric dependence structures across summer and winter seasons. Based on the selected best-fitting copula models, univariate, joint, and conditional return periods are derived to quantify the risk of compound extremes. The results highlight the capability of the proposed approach to provide a more realistic representation of hydro-climatic dependence and offer a robust framework for assessing the risk of extreme temperature and rainfall events in mountainous regions.

\vspace{-1em}
\end{abstract}

\section{Introduction} 
Bivariate analysis has been widely recognized as an effective framework for modeling hydrological and climatic extremes characterized by multiple interdependent variables \citep{shiau2003return, poonia2022bivariate}. Conventional approaches often rely on restrictive assumptions, such as identical marginal distributions or normality, which are rarely satisfied in real-world hydrological processes \citep{tosunouglu2017joint, zhang2007bivariate}. In contrast, copula-based methods offer a flexible alternative by allowing arbitrary marginal distributions while explicitly capturing the underlying dependence structure between variables \citep{Genest2007, zhang2006bivariate}. A key strength of copulas lies in their ability to model not only overall dependence but also tail dependence, which is crucial for accurately representing the joint occurrence of extreme events. This flexibility has led to their increasing adoption in a wide range of hydrological and climatic applications. In rainfall studies, copula-based frameworks have been extensively utilized to model the joint behaviour of characteristics such as intensity, depth, and duration \citep{nazeri2022application}. Similarly, in flood frequency analysis, copulas have proven effective in capturing the dependence among peak discharge, volume, and duration, thereby improving the estimation of joint flood risks \citep{razmkhah2022multivariate, berbesi2025flood, li2025identifying, xie2023assessment}. Beyond floods, these methods have also been successfully applied to drought and low-flow analysis, where the joint assessment of severity and occurrence provides a more comprehensive understanding of hydrological deficits \citep{deger2023univariate, avsaroglu2022assessment, terzi2026probabilistic}. By effectively capturing complex dependence structures, including tail dependence, they provide more reliable estimates of joint and conditional return periods \citep{li2013return, li2013bivariate, sahoo2020bivariate, yin2022definition}. 

Despite the extensive application of copula-based methods in hydrological studies, the joint analysis of temperature and rainfall remains relatively less explored, particularly in complex mountainous regions such as the Northwest Himalayas (NWH). The NWH region has experienced noticeable variability in both temperature and rainfall patterns in recent decades, characterized by increasing temperature trends and irregular precipitation behaviour \citep{yaduvanshi2021temperature, negi2019appraisal, mishra2024investigating, upadhyaya2023anomalous, mishra2026spatiotemporal}. Recent studies over the Kashmir valley and broader NWH region indicate that increasing temperatures are accompanied by a decline in overall rainfall in some areas, along with significant alterations in precipitation characteristics \citep{shafiq2019temperature}. Moreover, rising temperatures, particularly during winter and spring, have been linked to reduced snowfall, shrinking glacier mass, and enhanced hydrological stress, while also contributing to an increase in extreme rainfall events in certain regions \citep{suri2023rainfall}. 

Recent studies further reveal pronounced short-term spatiotemporal variability in rainfall, characterized by strong spatial dependence and rapidly evolving temporal structures, particularly in high-altitude regions \citep{sharma2025spatio}. Concurrently, long-term analyses indicate substantial seasonal shifts in precipitation driven by changes in atmospheric circulation and temperature, reinforcing the interdependence between thermal and precipitation regimes \citep{banerjee2023solid, jena2019weakening}. Large-scale assessments also identify the NWH region as increasingly vulnerable to both droughts and floods under changing climatic conditions \citep{jena2021observed}.
Capturing such variability is further complicated by observational limitations, as the reliability of rainfall data in mountainous terrain depends strongly on the spatial configuration of rain gauge networks \citep{suri2025optimal}. These challenges highlight the need for robust joint modeling frameworks, as univariate approaches often fail to capture dependence structures, particularly during extremes. From a statistical perspective, flexible probability distributions have been developed to better model extreme temperature and rainfall behaviour and their return levels \citep{poonia2022alpha, poonia2022projection}. In this context, \citep{poonia2023new} proposed a bivariate Exponentiated Teissier distribution using the Clayton copula for temperature over the NWH region and derived corresponding joint and conditional return periods. Extending such frameworks to include rainfall alongside temperature enables a more comprehensive representation of the hydro-climatic system. Accordingly, the present study proposes two new bivariate distributions for jointly modelling temperature and rainfall over the NWH region, along with their associated return periods.

Via copulas, \citep{sklar1973random} established the foundational link between multivariate distribution functions and their univariate margins. For two random variables \( Y_1 \) and \( Y_2 \) with cumulative distribution functions (CDF) \( F_{Y_1}(y_1) \) and \( F_{Y_2}(y_2) \), the joint distribution function \( F(y_1, y_2) \) can be expressed as:
\begin{eqnarray}
	F(y_1, y_2) &=& C(F_1(y_1), F_2(y_2)) 
\end{eqnarray}
where \( C: [0,1]^2 \rightarrow [0,1] \) is a bivariate copula function. Provided the requisite derivatives exist, the corresponding joint probability density function (PDF) is given by:
\begin{eqnarray}
	f(y_1, y_2) &=& c(F_1(y_1), F_2(y_2))f_1(y_1) f_2(y_2) 
\end{eqnarray}
where \( f_{Y_1} \) and \( f_{Y_2} \) are the marginal densities, and \( c(u, v) = \frac{\partial^2 C(u, v)}{\partial u \, \partial v} \) is the copula density.

The Clayton copula is a prominent member of the Archimedean family, extensively employed in dependence modeling due to its ability to capture lower tail dependence. Its bivariate CDF is defined for marginals \( \mathscr{U}, \mathscr{V} \in [0,1] \) as:
\begin{eqnarray}
	C(\mathscr{U}, \mathscr{V}) = (\mathscr{U}^{-\delta_1} + \mathscr{V}^{-\delta_1} -1)^{-1/\delta_1}
\end{eqnarray}
where the dependence parameter is constrained to $\delta_1 \in [-1, \infty) \setminus {0}$. The corresponding copula density function is given by:
\begin{eqnarray}
	c(\mathscr{U}, \mathscr{V}) = (\delta_1+1) (\mathscr{U} \mathscr{V})^{-1-\delta_1} (\mathscr{U}^{-\delta_1} + \mathscr{V}^{-\delta_1} -1)^{-2-1/\delta_1}
\end{eqnarray}
In contrast, the Gumbel copula, another widely used Archimedean copula, is particularly suited for modeling upper tail dependence and its CDF is defined as: 
\begin{align}
	C(\mathscr{U}, \mathscr{V}) = \exp \left( -\left\{ (-\log\mathscr{U})^{\delta_2} + (-\log\mathscr{V})^{\delta_2} \right\}^{1/{\delta_2}} \right)
\end{align}
where  the dependence parameter is restricted to $\delta_2 \in [1, \infty)$. The corresponding copula density function is expressed as:
\begin{align}
	c(\mathscr{U}, \mathscr{V}) =  \frac{\exp \left( -\left\{ (-\log\mathscr{U})^{\delta_2} + (-\log\mathscr{V})^{\delta_2} \right\}^{1/{\delta_2}} \right) \left[\left\{ (-\log\mathscr{U})^{\delta_2} + (-\log\mathscr{V})^{\delta_2} \right\}^{1/{\delta_2}} + {\delta_2} -1 \right] }{ \mathscr{U} \mathscr{V}\ (\log\mathscr{U} \log\mathscr{V})^{1-{\delta_2}}  \left\{ (-\log\mathscr{U})^{\delta_2} + (-\log\mathscr{V})^{\delta_2} \right\}^{2-1/{\delta_2} }}
\end{align}

The Kumaraswamy–Teissier distribution (KTD), introduced by \citep{mishra2026}, provides a flexible framework for modeling hydro-meteorological data. Its cumulative distribution function for $x>0$ is defined as:
\begin{eqnarray}\label{ktcdf}
	\F(\x) = 1-\left( 1-\left( 1-e^{\fx} \right)^\a \right)^\b,
\end{eqnarray}
where, $\fx =  \th \x - e^{\th \x} +1 $ and $\a,\b,\th>0$. The corresponding PDF is:
\begin{eqnarray}\label{ktpdf}
	\f(\x) = \a \b \th (e^{\th \x} -1) e^{\fx} \left(1-e^{\fx} \right)^{\a-1} \left(  1-\left(1-e^{\fx} \right)^{\a} \right)^{\b-1}.
\end{eqnarray}

In this study, two new bivariate distributions based on Clayton and Gumbel copulas with KTD marginals are developed in Section \ref{bktd}. To accommodate both positive and negative dependence structures, the corresponding rotated copulas along with their tail dependence properties are presented in Section \ref{tailsec}. Parameter estimation is carried out using maximum likelihood estimation and the inference functions for margins approach, as described in Section \ref{mlifm}. The accuracy of the estimators are further assessed through a Monte Carlo simulation study in Section \ref{simu}. Finally, the practical applicability of the proposed models is demonstrated through a real-life analysis of temperature and rainfall events over the NWH in Section \ref{rla}.

\section{Bivariate Kumaraswamy Teissier Distribution} \label{bktd}
\subsection{Clayton Bivariate Kumaraswamy Teissier (CBKT) Distribution}

A random vector $(X_1, X_2)$ is said to follow Clayton copula-based bivariate Kumaraswamy-Teissier Distribution with parameter vector $\pv_1 = (a_1, b_1, \th_1, a_2, b_2, \th_2, \delta_1)$ if its CDF is given by
\begin{eqnarray} \label{cbktcdf}
    F(x_1, x_2) = \left[\left\{1-\left( 1-\left( 1-e^{\phi(\x_1;\th_1)} \right)^{a_1} \right)^{b_1} \right\}^{-\delta_1} + \left\{1-\left( 1-\left( 1-e^{\phi(\x_2;\th_2)} \right)^{a_2} \right)^{b_2} \right\}^{-\delta_1} - 1 \right]^{-1/{\delta_1}}
\end{eqnarray}

The corresponding PDF of CBKT($\pv$) is obtained by performing mixed second order partial differentiation on CDF and given by:
\begin{eqnarray}  \label{cbktpdf}
    f(x_1, x_2) &=& a_1 b_1 \th_1 a_2 b_2 \th_2 (\delta_1 +1) (e^{\th_1 x_1}-1) (e^{\th_2 x_2}-1) e^{\phi(\x_1;\th_1)} e^{\phi(\x_2;\th_2)} \left( 1-e^{\phi(\x_1;\th_1)} \right)^{a_1-1} \nonumber \\ &&
    \left( 1-e^{\phi(\x_2;\th_2)} \right)^{a_2-1} \left( 1-\left( 1-e^{\phi(\x_1;\th_1)} \right)^{a_1} \right)^{b_1-1} \left( 1-\left( 1-e^{\phi(\x_2;\th_2)} \right)^{a_2} \right)^{b_2-1} \nonumber \\ && \left\{1-\left( 1-\left( 1-e^{\phi(\x_1;\th_1)} \right)^{a_1} \right)^{b_1} \right\}^{-\delta_1-1} \left\{1-\left( 1-\left( 1-e^{\phi(\x_2;\th_2)} \right)^{a_2} \right)^{b_2} \right\}^{-\delta_1-1} \nonumber \\ && \left[\left\{1-\left( 1-\left( 1-e^{\phi(\x_1;\th_1)} \right)^{a_1} \right)^{b_1} \right\}^{-\delta_1} + \left\{1-\left( 1-\left( 1-e^{\phi(\x_2;\th_2)} \right)^{a_2} \right)^{b_2} \right\}^{-\delta_1} - 1 \right]^{\frac{(-2 \delta_1 -1)}{\delta_1}}
\end{eqnarray}
Figure \ref{Figure 1} illustrates CDF \& PDF of CBKT for a specific set of parameters.
\begin{figure}[!ht]
	\centering
	\subfloat[]{\includegraphics[width=0.49\textwidth]{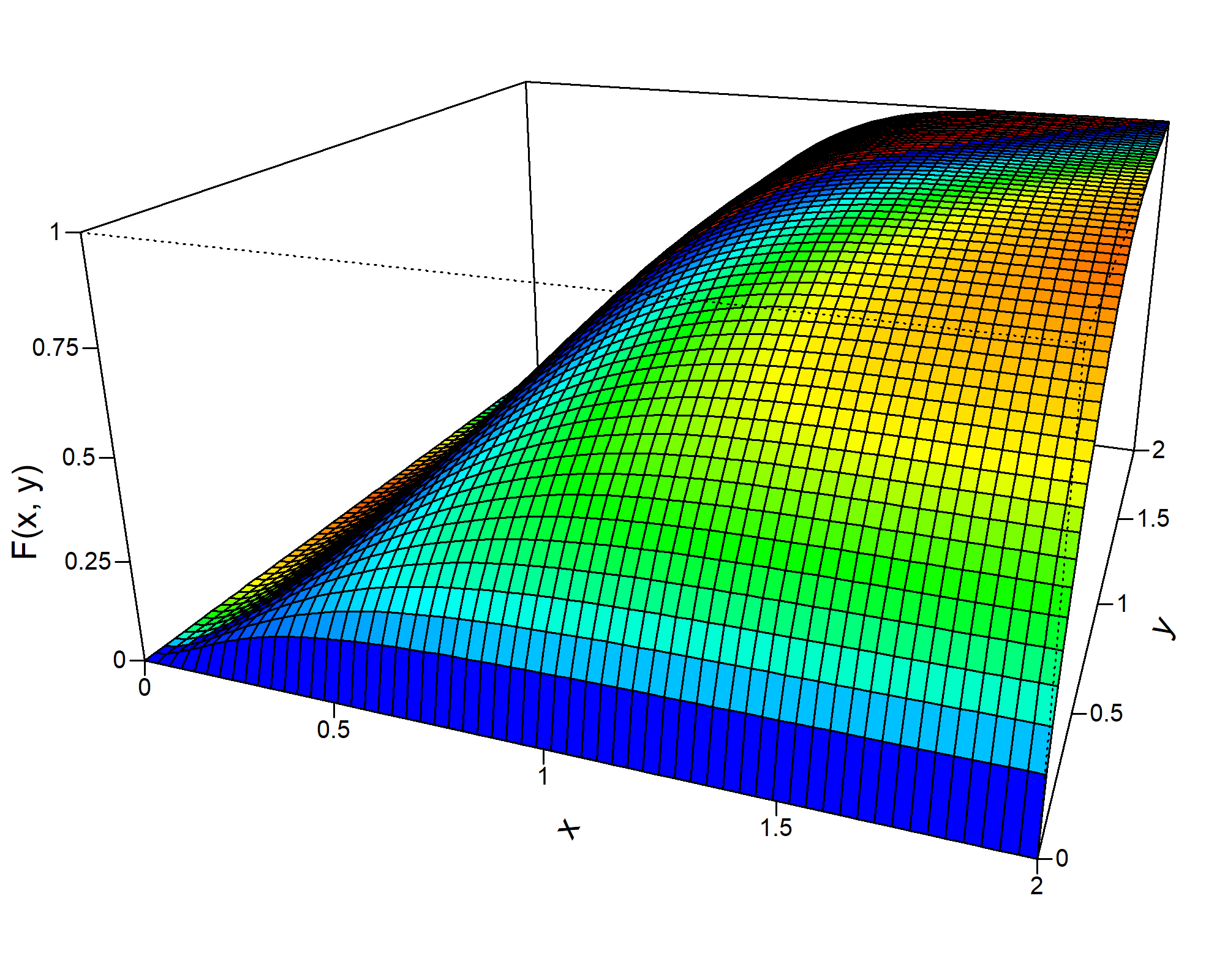}\label{fig1:f1}}
	\hfill
	\subfloat[]{\includegraphics[width=0.49\textwidth]{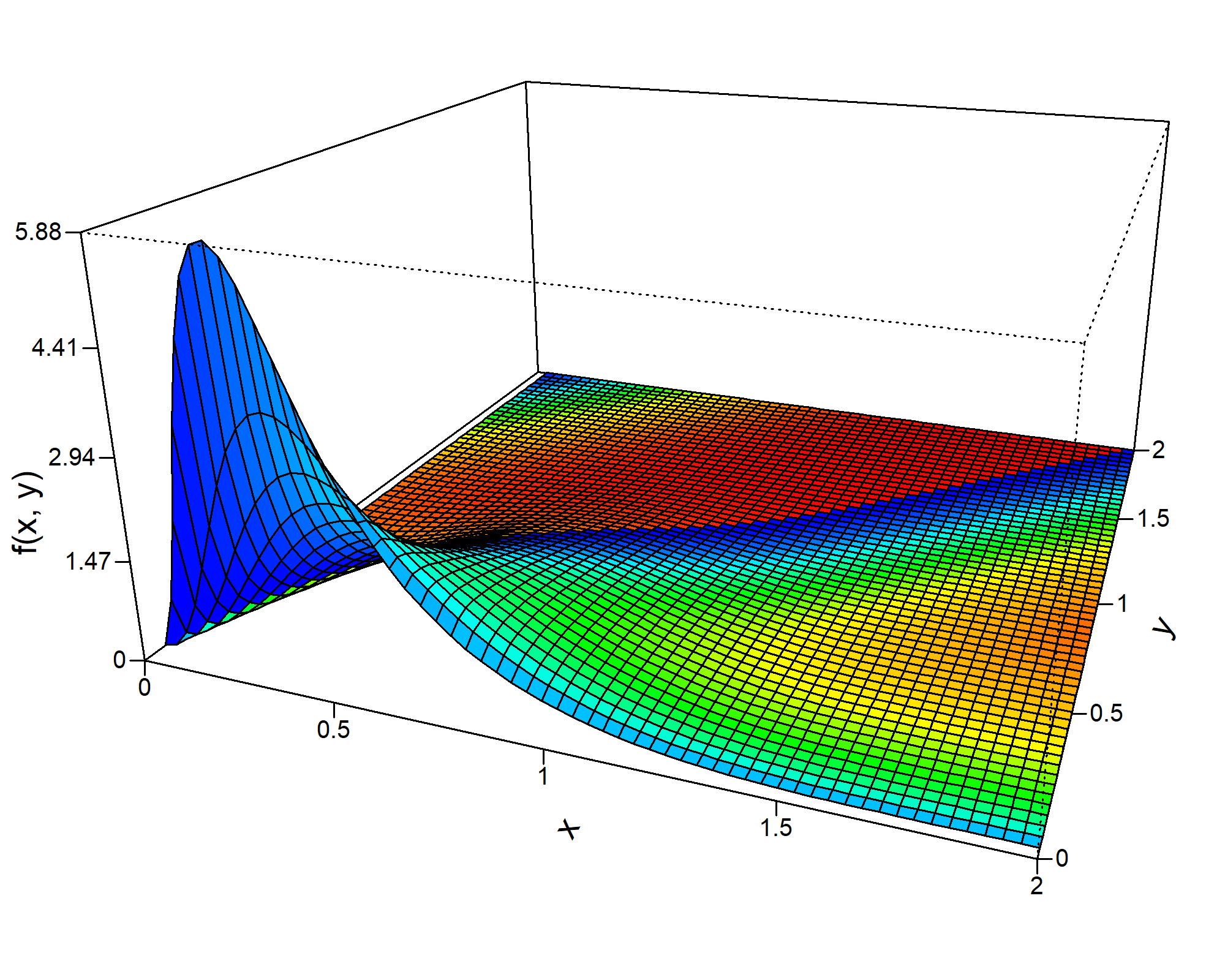}\label{fig1:f2}}
	\caption{CDF and PDF plots of the CBKT distribution for parameters $a_1 = 0.8$, $a_2 = 0.3$, $b_1 = 1.6$, $b_2 = 1.3$, $t_1 = 1.2$, $t_2 = 1.3$, and $\delta_1 = 1.2$}
	\label{Figure 1}
\end{figure}

\begin{theorem}
Let \( X_1, X_2 \sim  CBKT(a_1, b_1, \th_1, a_2, b_2, \th_2, \delta_1)  \). Then, marginal \(X_i \sim KTD(a_i, b_i, \th_i);\ i =1,2 \), i.e.,
\(F_{X_i} (x_i; a_i, b_i, \th_i ) = 1-\left( 1-\left( 1-e^{\phi(\x_i;\th_i)} \right)^{\a_i} \right)^{\b_i} \).
\end{theorem}
\begin{proof}
Since, \( 1-\left( 1-\left( 1-e^{\fx} \right)^\a \right)^\b \) is CDF of KTD. So, by using property of CDF, we get
\[\lim_{x \to \infty} 1-\left( 1-\left( 1-e^{\fx} \right)^\a \right)^\b  = 1 \] Now applying the limit on the joint CDF in eq. \eqref{cbktcdf} to obtain the marginals as follows:
\begin{align*}
F_{X_1} (x_1) &=  \lim_{x_2 \to \infty} F_{X_1,X_2} (x_1, x_2) \\ &= \lim_{x_2 \to \infty} \left[\left\{1-\left( 1-\left( 1-e^{\phi(\x_1;\th_1)} \right)^{a_1} \right)^{b_1} \right\}^{-\delta_1} + \left\{1-\left( 1-\left( 1-e^{\phi(\x_2;\th_2)} \right)^{a_2} \right)^{b_2} \right\}^{-\delta_1} - 1 \right]^{-1/{\delta_1}} \\
&= \left[\left\{1-\left( 1-\left( 1-e^{\phi(\x_1;\th_1)} \right)^{a_1} \right)^{b_1} \right\}^{-\delta_1} + 1 - 1 \right]^{-1/{\delta_1}}\\
&= 1-\left( 1-\left( 1-e^{\phi(\x_1;\th_1)} \right)^{a_1} \right)^{b_1}
\end{align*}
Similarly, \(F_{X_2} (x_2) = 1-\left( 1-\left( 1-e^{\phi(\x_2;\th_2)} \right)^{a_2} \right)^{b_2}\).
\end{proof}

\begin{theorem}
Let \( X_1, X_2 \sim  CBKT(a_1, b_1, \th_1, a_2, b_2, \th_2, \delta_1)  \). Then, the conditional PDF of \(X_1\) given \( X_2 = x_2 \) is
\begin{eqnarray}
    f (x_1| x_2)  &=& a_1 b_1 \th_1 (\delta_1 +1) (e^{\th_1 x_1}-1) e^{\phi(\x_1;\th_1)}\left( 1-e^{\phi(\x_1;\th_1)} \right)^{a_1-1}  \left( 1-\left( 1-e^{\phi(\x_1;\th_1)} \right)^{a_1} \right)^{b_1-1}\nonumber \\ &&
     \left\{1-\left( 1-\left( 1-e^{\phi(\x_1;\th_1)} \right)^{a_1} \right)^{b_1} \right\}^{-\delta_1-1} \left\{1-\left( 1-\left( 1-e^{\phi(\x_2;\th_2)} \right)^{a_2} \right)^{b_2} \right\}^{-\delta_1-1} \nonumber \\ && \left[\left\{1-\left( 1-\left( 1-e^{\phi(\x_1;\th_1)} \right)^{a_1} \right)^{b_1} \right\}^{-\delta_1} + \left\{1-\left( 1-\left( 1-e^{\phi(\x_2;\th_2)} \right)^{a_2} \right)^{b_2} \right\}^{-\delta_1} - 1 \right]^{\frac{(-2 \delta_1 -1)}{\delta_1}}
\end{eqnarray}
\end{theorem}

\begin{corollary}
    Let \( X_1, X_2 \sim  CBKT(a_1, b_1, \th_1, a_2, b_2, \th_2, \delta_1)  \). Then, the bivariate survival function can be written as 
    \begin{eqnarray} \label{cbktsf}
    \overline{F} (x_1, x_2) = 1-\left[\left\{1-\left( 1-\left( 1-e^{\phi(\x_1;\th_1)} \right)^{a_1} \right)^{b_1} \right\}^{-\delta_1} + \left\{1-\left( 1-\left( 1-e^{\phi(\x_2;\th_2)} \right)^{a_2} \right)^{b_2} \right\}^{-\delta_1} - 1 \right]^{-1/{\delta_1}}
\end{eqnarray}
\end{corollary}
\begin{corollary}
    Let \( X_1, X_2 \sim  CBKT(a_1, b_1, \th_1, a_2, b_2, \th_2, \delta_1)  \). Then, the joint reliability function can be written as 
    \begin{eqnarray} \label{cbktrf}
    R (x_1, x_2) &=& 1- \left[1-\left( 1-\left( 1-e^{\phi(\x_1;\th_1)} \right)^{a_1} \right)^{b_1} \right] - \left[ 1-\left( 1-\left( 1-e^{\phi(\x_2;\th_2)} \right)^{a_2} \right)^{b_2} \right] +\nonumber \\ && \left[\left\{1-\left( 1-\left( 1-e^{\phi(\x_1;\th_1)} \right)^{a_1} \right)^{b_1} \right\}^{-\delta_1} + \left\{1-\left( 1-\left( 1-e^{\phi(\x_2;\th_2)} \right)^{a_2} \right)^{b_2} \right\}^{-\delta_1} - 1 \right]^{-1/{\delta_1}}
\end{eqnarray}
\end{corollary}

\subsection{Gumbel Bivariate Kumaraswamy Teissier (GBKT) Distribution}
A random vector $(X_1, X_2)$ follows a bivariate Kumaraswamy-Teissier distribution with a Gumbel copula, denoted by $\text{Gum-BKTD}(\boldsymbol{\Theta})$, where $\boldsymbol{\Theta} = (a_1, b_1, \theta_1, a_2, b_2, \theta_2, \delta_2)$, if its joint CDF is defined as:
{\small
	\begin{equation}\label{gbktcdf}
		G(x_1,x_2)=
		\exp\!\left(
		-\left\{
		[-\log(1-(1-(1-e^{\phi(x_1;\theta_1)})^{a_1})^{b_1})]^{\delta_2}
		+
		[-\log(1-(1-(1-e^{\phi(x_2;\theta_2)})^{a_2})^{b_2})]^{\delta_2}
		\right\}^{1/\delta_2}
		\right)
\end{equation}}

The corresponding PDF of GBKT($\pv$) is obtained by performing mixed second order partial differentiation on CDF and expressed as:
{\small
\begin{align} \label{gbktpdf}
g(x_1, x_2) =\; 
& a_1 b_1 \th_1 a_2 b_2 \th_2 (e^{\th_1 x_1}-1) (e^{\th_2 x_2}-1) e^{\phi(\x_1;\th_1)} e^{\phi(\x_2;\th_2)} \left( 1-e^{\phi(\x_1;\th_1)} \right)^{a_1-1} \left( 1-e^{\phi(\x_2;\th_2)} \right)^{a_2-1} \notag \\
& \left( 1-\left( 1-e^{\phi(\x_1;\th_1)} \right)^{a_1} \right)^{b_1-1} \left( 1-\left( 1-e^{\phi(\x_2;\th_2)} \right)^{a_2} \right)^{b_2-1}   \notag \\
& \left\{
\left[ -\log \left( 1 - \left(1 - \left(1 - e^{\phi(x_1;\theta_1)} \right)^{a_1} \right)^{b_1} \right) \right]^{\delta_2} +
\left[ -\log \left( 1 - \left(1 - \left(1 - e^{\phi(x_2;\theta_2)} \right)^{a_2} \right)^{b_2} \right) \right]^{\delta_2}
\right\}^{-2 + 1/{\delta_2}} \notag \\
&   \frac{\left[
\left\{
\left[ -\log \left( 1 - \left(1 - \left(1 - e^{\phi(x_1;\theta_1)} \right)^{a_1} \right)^{b_1} \right) \right]^{\delta_2} +
\left[ -\log \left( 1 - \left(1 - \left(1 - e^{\phi(x_2;\theta_2)} \right)^{a_2} \right)^{b_2} \right) \right]^{\delta_2}
\right\}^{1/{\delta_2}} + {\delta_2} - 1
\right]}{
\left[
\log \left( 1 - \left(1 - \left(1 - e^{\phi(x_1;\theta_1)} \right)^{a_1} \right)^{b_1} \right)
\log \left( 1 - \left(1 - \left(1 - e^{\phi(x_2;\theta_2)} \right)^{a_2} \right)^{b_2} \right)
\right]^{ 1 - {\delta_2} }
} \notag \\
 &  \frac{
\exp\left( -\left\{ 
\left[ -\log \left( 1 - \left(1 - \left(1 - e^{\phi(x_1;\theta_1)} \right)^{a_1} \right)^{b_1} \right) \right]^{\delta_2} +
\left[ -\log \left( 1 - \left(1 - \left(1 - e^{\phi(x_2;\theta_2)} \right)^{a_2} \right)^{b_2} \right) \right]^{\delta_2}
\right\}^{1/{\delta_2}} \right)
}{
\left( 1 - \left(1 - \left(1 - e^{\phi(x_1;\theta_1)} \right)^{a_1} \right)^{b_1} \right)
\left( 1 - \left(1 - \left(1 - e^{\phi(x_2;\theta_2)} \right)^{a_2} \right)^{b_2} \right)
} \notag \\
\end{align}
}

The graphical representation for CDF \& PDF of GBKT is shown in Figure \ref{Figure 2}. 
\begin{figure}[!ht]
	\centering
	\subfloat[]{\includegraphics[width=0.49\textwidth]{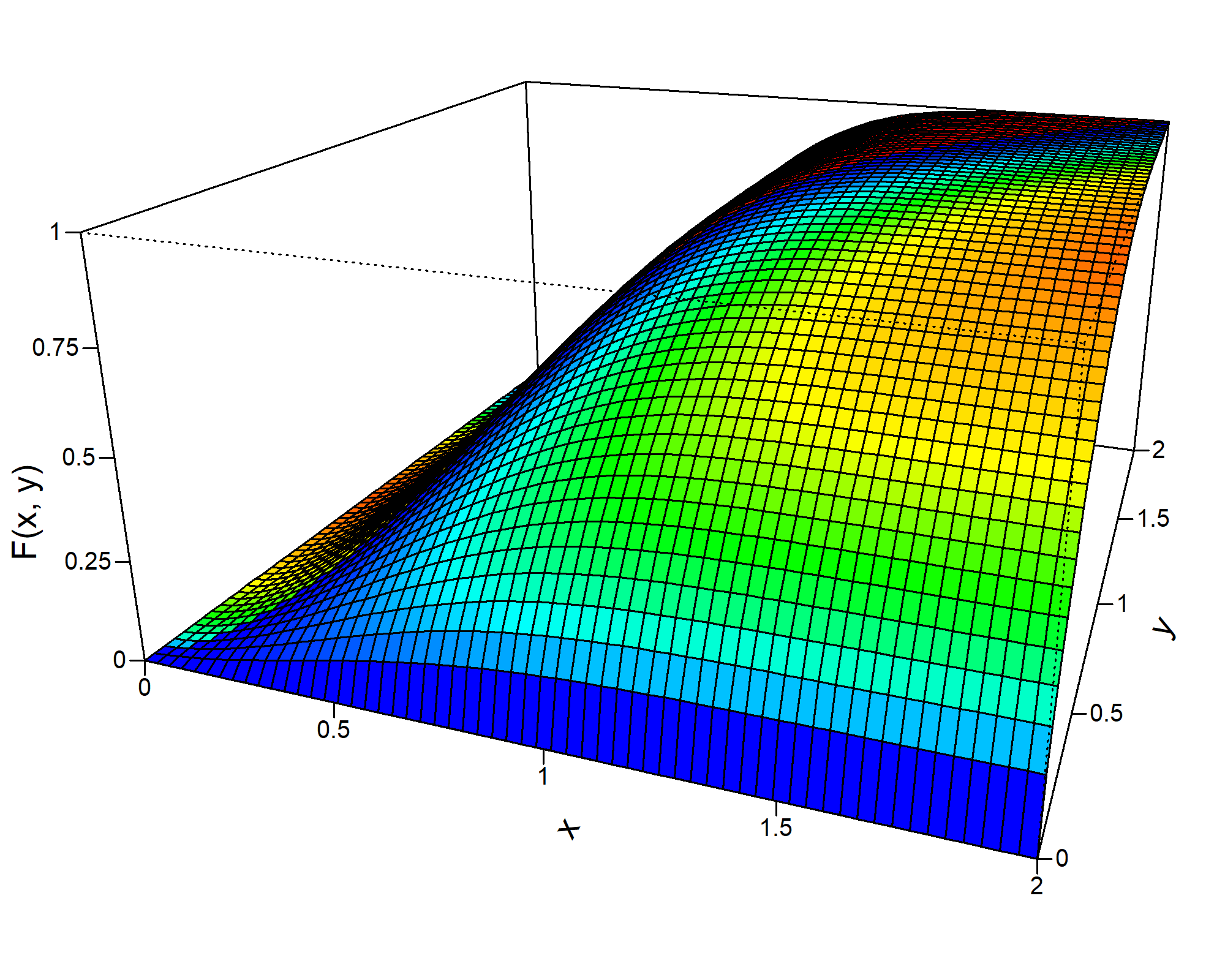}\label{fig2:f1}}
	\hfill
	\subfloat[]{\includegraphics[width=0.49\textwidth]{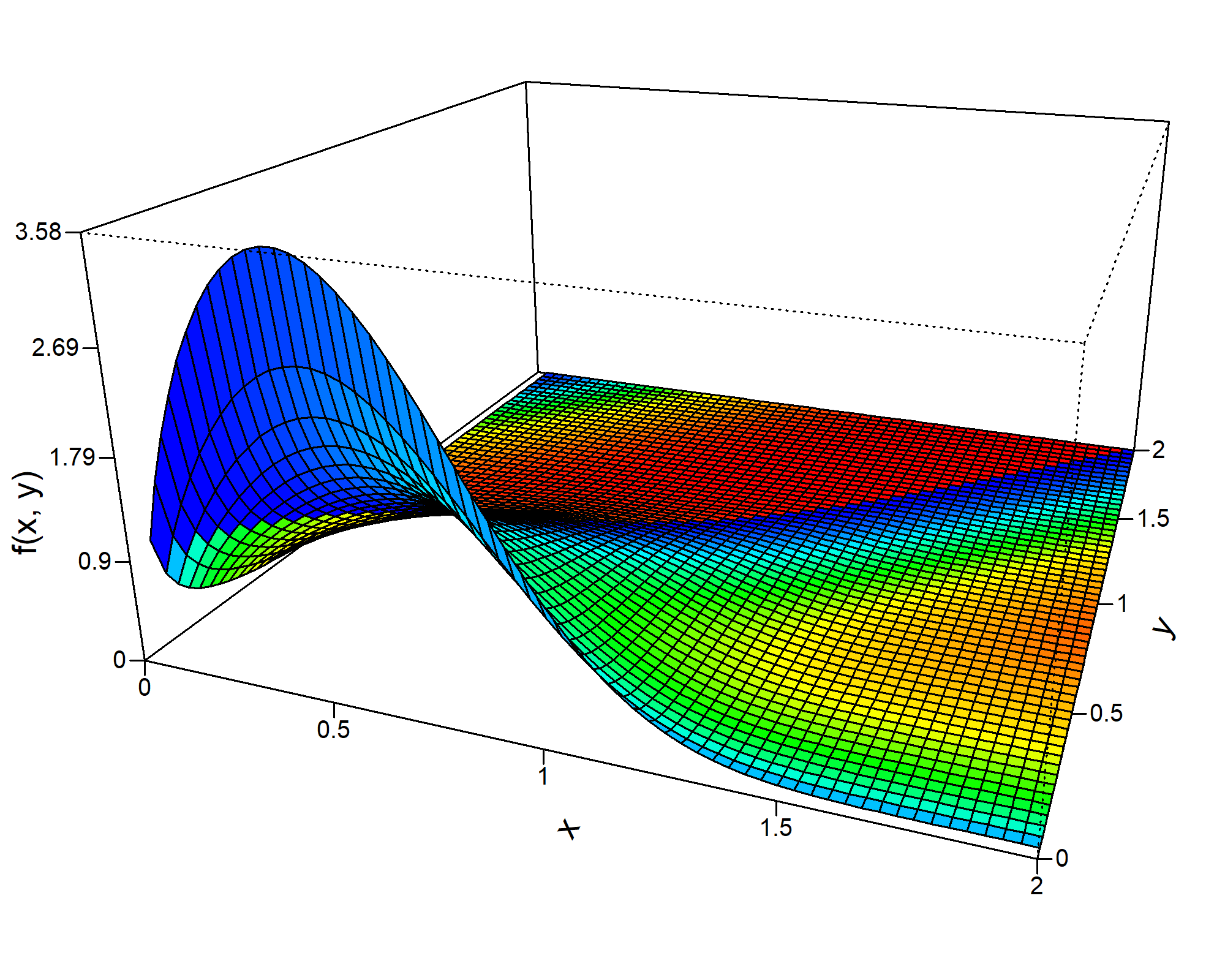}\label{fig2:f2}}
	\caption{CDF and PDF plots of the GBKT distribution for parameters $a_1 = 0.8$, $a_2 = 0.3$, $b_1 = 1.6$, $b_2 = 1.3$, $t_1 = 1.2$, $t_2 = 1.3$, and $\delta_2 = 1.2$}
	\label{Figure 2}
\end{figure}

\begin{theorem}
Let \( X_1, X_2 \sim  GBKT(a_1, b_1, \th_1, a_2, b_2, \th_2, {\delta_2})  \). Then, marginal \(X_i \sim KTD(a_i, b_i, \th_i);\ i =1,2 \), i.e.,
\(G_{X_i} (x_i; a_i, b_i, \th_i ) = 1-\left( 1-\left( 1-e^{\phi(\x_i;\th_i)} \right)^\a \right)^\b \).
\end{theorem}
\begin{proof}
Applying the limit on the joint CDF in eq. \eqref{gbktpdf} to obtain the marginals as follows:
\begin{align*}
	G_{X_1} (x_1) 
	&=  \lim_{x_2 \to \infty} G (x_1, x_2) \\ 
	&= \lim_{x_2 \to \infty} \exp\Bigg( -\Bigg\{ 
	\left[ -\log \left( 1 - \left(1 - \left(1 - e^{\phi(x_1;\theta_1)} \right)^{a_1} \right)^{b_1} \right) \right]^{\delta_2} \\
	&\qquad +
	\left[ -\log \left( 1 - \left(1 - \left(1 - e^{\phi(x_2;\theta_2)} \right)^{a_2} \right)^{b_2} \right) \right]^{\delta_2}
	\Bigg\}^{1/{\delta_2}} \Bigg) \\
	&= \exp\Bigg( -\Bigg\{ 
	\left[ -\log \left( 1 - \left(1 - \left(1 - e^{\phi(x_1;\theta_1)} \right)^{a_1} \right)^{b_1} \right) \right]^{\delta_2} +
	\left[ -\log (1) \right]^{\delta_2} \Bigg\}^{1/{\delta_2}} \Bigg)\\
	&= \exp\left( \log \left( 1 - \left(1 - \left(1 - e^{\phi(x_1;\theta_1)} \right)^{a_1} \right)^{b_1} \right) \right) \\
	&= 1-\left( 1-\left( 1-e^{\phi(x_1;\theta_1)} \right)^{a_1} \right)^{b_1}
\end{align*}
Similarly, \(G_{X_2} (x_2) = 1-\left( 1-\left( 1-e^{\phi(\x_2;\th_2)} \right)^{a_2} \right)^{b_2}\).
\end{proof}

\begin{theorem}
Let \( X_1, X_2 \sim  GBKT(a_1, b_1, \th_1, a_2, b_2, \th_2, {\delta_2})  \). Then, the conditional PDF of \(X_1\) given \( X_2 = x_2 \) is
{\small
\begin{align} 
g (x_1| x_2) =\; 
& a_1 b_1 \th_1 (e^{\th_1 x_1}-1)  e^{\phi(\x_1;\th_1)} \left( 1-e^{\phi(\x_1;\th_1)} \right)^{a_1-1}   \left( 1-\left( 1-e^{\phi(\x_1;\th_1)} \right)^{a_1} \right)^{b_1-1} \notag \\
& \left\{
\left[ -\log \left( 1 - \left(1 - \left(1 - e^{\phi(x_1;\theta_1)} \right)^{a_1} \right)^{b_1} \right) \right]^{\delta_2} +
\left[ -\log \left( 1 - \left(1 - \left(1 - e^{\phi(x_2;\theta_2)} \right)^{a_2} \right)^{b_2} \right) \right]^{\delta_2}
\right\}^{-2 + 1/{\delta_2}} \notag \\
&   \frac{\left[
\left\{
\left[ -\log \left( 1 - \left(1 - \left(1 - e^{\phi(x_1;\theta_1)} \right)^{a_1} \right)^{b_1} \right) \right]^{\delta_2} +
\left[ -\log \left( 1 - \left(1 - \left(1 - e^{\phi(x_2;\theta_2)} \right)^{a_2} \right)^{b_2} \right) \right]^{\delta_2}
\right\}^{1/{\delta_2}} + {\delta_2} - 1
\right]  }  {
\left[
\log \left( 1 - \left(1 - \left(1 - e^{\phi(x_1;\theta_1)} \right)^{a_1} \right)^{b_1} \right)
\log \left( 1 - \left(1 - \left(1 - e^{\phi(x_2;\theta_2)} \right)^{a_2} \right)^{b_2} \right)
\right]^{ 1 - {\delta_2} }
} \notag \\
 &  \frac{
\exp\left( -\left\{ 
\left[ -\log \left( 1 - \left(1 - \left(1 - e^{\phi(x_1;\theta_1)} \right)^{a_1} \right)^{b_1} \right) \right]^{\delta_2} +
\left[ -\log \left( 1 - \left(1 - \left(1 - e^{\phi(x_2;\theta_2)} \right)^{a_2} \right)^{b_2} \right) \right]^{\delta_2}
\right\}^{1/{\delta_2}} \right)
}{
\left( 1 - \left(1 - \left(1 - e^{\phi(x_1;\theta_1)} \right)^{a_1} \right)^{b_1} \right)
\left( 1 - \left(1 - \left(1 - e^{\phi(x_2;\theta_2)} \right)^{a_2} \right)^{b_2} \right)
} \notag \\
\end{align}
}
\end{theorem}

\begin{corollary}
    Let \( X_1, X_2 \sim  GBKT(a_1, b_1, \th_1, a_2, b_2, \th_2, {\delta_2})  \). Then, the bivariate survival function can be written as 
\begin{align} \label{gbktsf}
     \overline{G} (x_1, x_2) &= 1- \exp\Bigg( -\Bigg\{ 
     \left[ -\log \left( 1 - \left(1 - \left(1 - e^{\phi(x_1;\theta_1)} \right)^{a_1} \right)^{b_1} \right) \right]^{\delta_2} \notag \\ 
     &\qquad +
     \left[ -\log \left( 1 - \left(1 - \left(1 - e^{\phi(x_2;\theta_2)} \right)^{a_2} \right)^{b_2} \right) \right]^{\delta_2}
     \Bigg\}^{1/{\delta_2}} \Bigg)
\end{align}
\end{corollary}

\begin{corollary}
    Let \( X_1, X_2 \sim  GBKT(a_1, b_1, \th_1, a_2, b_2, \th_2, \delta_2)  \). Then, the joint reliability function can be written as 
{\small
\begin{align} \label{gbktrf}
    R (x_1, x_2) =\; 
& \exp\left( -\left\{ 
\left[ -\log \left( 1 - \left(1 - \left(1 - e^{\phi(x_1;\theta_1)} \right)^{a_1} \right)^{b_1} \right) \right]^{\delta_2} +
\left[ -\log \left( 1 - \left(1 - \left(1 - e^{\phi(x_2;\theta_2)} \right)^{a_2} \right)^{b_2} \right) \right]^{\delta_2}
\right\}^{1/{\delta_2}} \right) +  \notag \\ & 1- \left[1-\left( 1-\left( 1-e^{\phi(\x_1;\th_1)} \right)^{a_1} \right)^{b_1} \right] - \left[ 1-\left( 1-\left( 1-e^{\phi(\x_2;\th_2)} \right)^{a_2} \right)^{b_2} \right] 
\end{align}
} 
\end{corollary}

\section{Rotated Copulas and Tail Dependence}\label{tailsec}
The Gumbel and Clayton copulas are among the most widely used Archimedean copulas for modeling dependence between random variables. However, these copulas are restricted to positive dependence and cannot capture negative correlation structures. To overcome this limitation and achieve greater flexibility in modeling both positive and negative associations, rotated versions of the Gumbel and Clayton copulas are introduced. Three types of rotations are employed, namely $90^{\circ}$, $180^{\circ}$, and $270^{\circ}$ which are obtained by reflecting the copula. These rotated versions are mathematically defined as:
\begin{eqnarray}
C_{00}(\mathscr{U}, \mathscr{V}) &=&  C(\mathscr{U}, \mathscr{V}), \\ \label{rot}
C_{10}(\mathscr{U}, \mathscr{V}) &=&  \mathscr{V} - C(1 - \mathscr{U}, \mathscr{V}), \\ \label{rot1}
C_{11}(\mathscr{U}, \mathscr{V}) &=&  \mathscr{U} + \mathscr{V} - 1 + C(1 - \mathscr{U}, 1 - \mathscr{V}), \\ \label{rot2}
C_{01}(\mathscr{U}, \mathscr{V}) &=&  \mathscr{U} - C(\mathscr{U}, 1 - \mathscr{V}),
\end{eqnarray}
where $C_{00}$ denotes the original copula, and $C_{10}$, $C_{11}$, and $C_{01}$ correspond to the $90^{\circ}$, $180^{\circ}$, and $270^{\circ}$ rotations, respectively. These transformations allow the same copula family to describe negative or mixed dependence patterns.

\subsection*{Tail Dependence}
Tail dependence measures the degree of association between extreme values of two random variables. Unlike correlation, which quantifies overall linear dependence, tail dependence focuses on the co-occurrence of extremes — i.e., whether large (or small) values in one variable are likely to coincide with large (or small) values in another. For a copula $C$, the coefficients of tail dependence are defined as:
\begin{equation}
\begin{aligned}
\lambda_{00} &= \lim_{\nu \to 0} \frac{C(\nu, \nu)}{\nu}, \quad
\lambda_{11} = \lim_{\nu \to 0} \frac{2\nu - 1 + C(1 - \nu, 1 - \nu)}{\nu},\\[5pt]
\lambda_{10} &= \lim_{\nu \to 0} \frac{\nu - C(1 - \nu, \nu)}{\nu}, \quad
\lambda_{01} = \lim_{\nu \to 0} \frac{\nu - C(\nu, 1 - \nu)}{\nu}.
\end{aligned}
\end{equation}
\begin{figure}[!ht]
	\centering
	\includegraphics[width=1\linewidth]{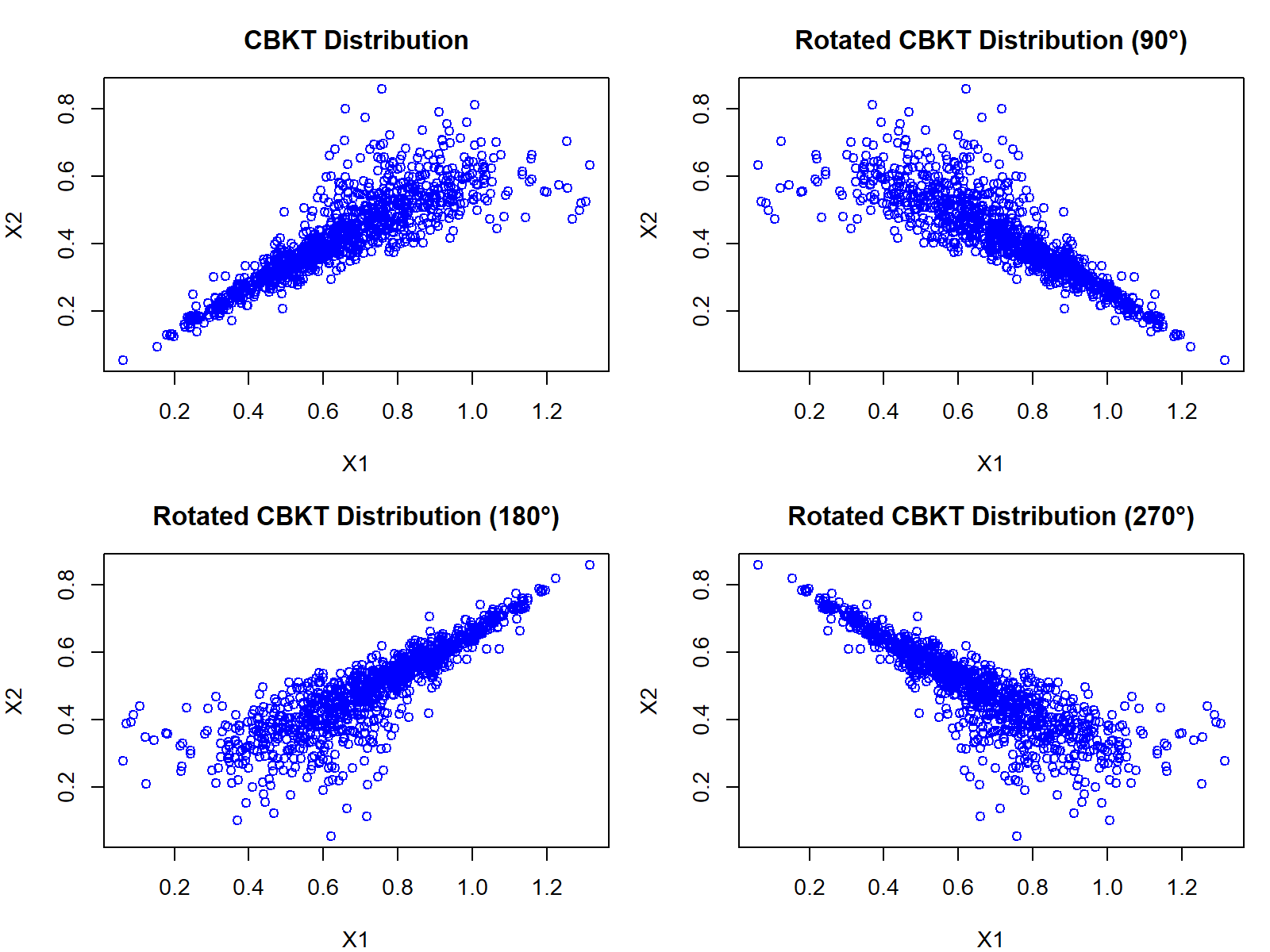} 
	\caption{Scatter plots of simulated data from CBKT model with parameters $a_1 = 1.8$, $a_2 = 1.9$, $b_1 = 1.7$, $b_2 = 1.3$, $t_1 = 1.6$, $t_2 = 2.7$, and $\delta_1 = 5$, shown under four copula rotations (0\textdegree, 90\textdegree, 180\textdegree, and 270\textdegree)}
	\label{Figure 7}
\end{figure}
Here, $\lambda_{00}$ and $\lambda_{11}$ represent the lower and upper tail dependence coefficients respectively, while $\lambda_{10}$ and $\lambda_{01}$ correspond to cross-tail dependencies arising from rotated copulas. For instance, the Clayton copula exhibits lower-tail dependence given by $\lambda_L = 2^{-1/\delta_1}$, while the Gumbel copula exhibits upper-tail dependence expressed as $\lambda_U = 2 - 2^{1/\delta_2}$.

\begin{figure}[!ht]
\centering
   \includegraphics[width=1\linewidth]{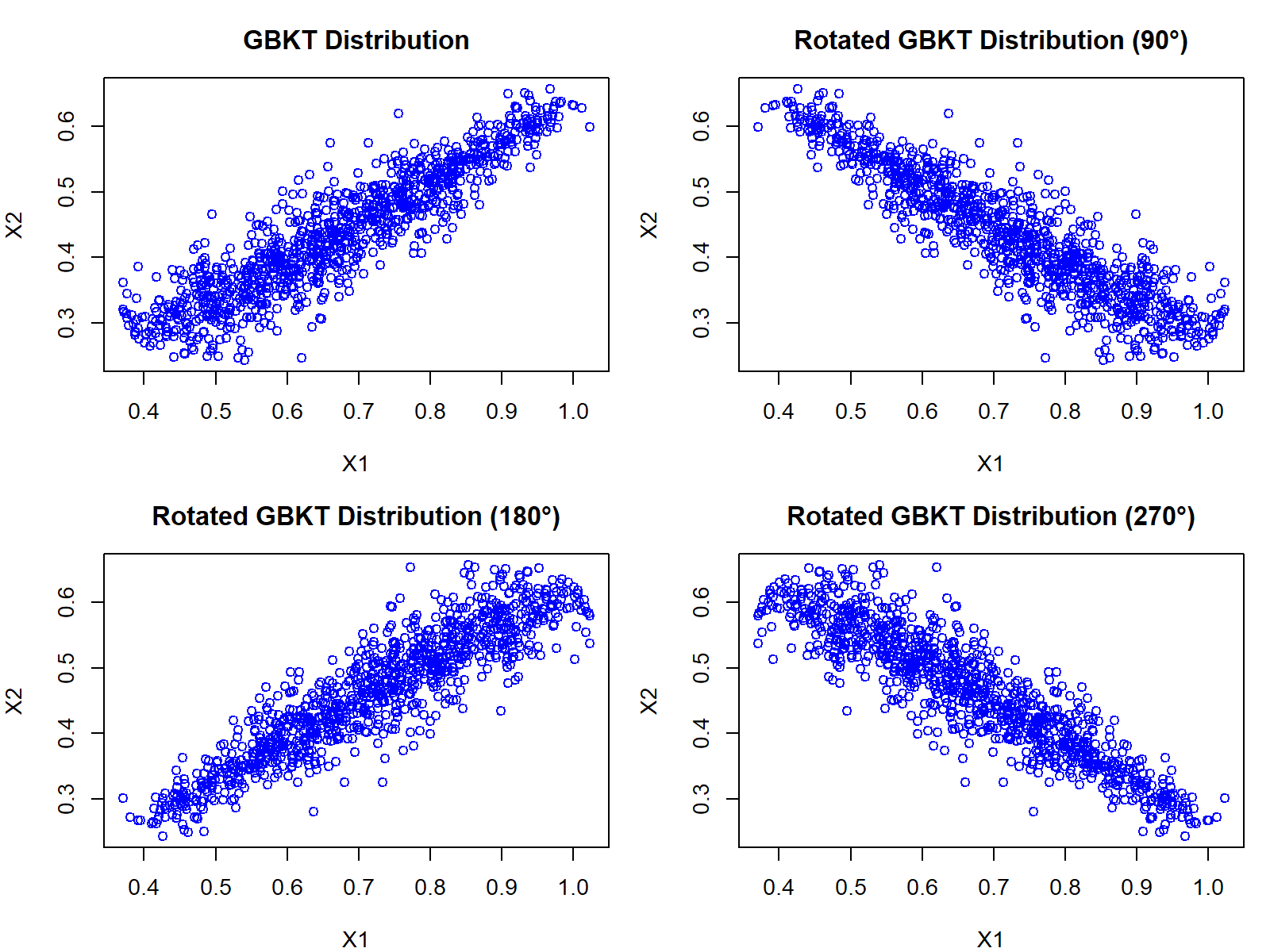} 
\caption{Scatter plots of simulated data from GBKT model with parameters $a_1 = 1.8$, $a_2 = 1.9$, $b_1 = 1.7$, $b_2 = 1.3$, $t_1 = 1.6$, $t_2 = 2.7$, and $\delta_2 = 5$, shown under four copula rotations (0\textdegree, 90\textdegree, 180\textdegree, and 270\textdegree)}
\label{Figure 8}
\end{figure}

As illustrated in Figure~\ref{Figure 7} \& \ref{Figure 8}, the rotation of copulas modifies the region of dependence they capture. The $90^{\circ}$ rotation corresponds to dependence between the lower tail of one variable and the upper tail of the other, while the $270^{\circ}$ rotation represents the opposite configuration. Specifically, the $90^{\circ}$ rotated Gumbel copula (GBKT90) models scenarios where the lower extremes of the first variable are associated with the upper extremes of the second variable, whereas the $90^{\circ}$ rotated Clayton copula (CBKT90) captures the reverse dependence structure. The $180^{\circ}$ rotation, Gumbel copula captures lower and Clayton copula captures upper tail dependence, effectively reversing the direction of dependence. These rotations extend the flexibility of copula-based models, enabling a comprehensive representation of positive, negative, and cross-tail dependencies.

\section{Parameter Estimation}\label{mlifm}
This section presents two primary estimation methodologies for parameter estimation of CBKT and GBKT distributions: Maximum Likelihood Estimation and Inference Functions for Margins.

\subsection{Maximum Likelihood Estimation}

Let $\{(X_{1i}, X_{2i})\}_{i=1}^n$ represent a random sample from the CBKT($\pv_1$) distribution. The log-likelihood function $L$, derived from PDF in equation \eqref{cbktpdf}, is expressed as:

\begin{eqnarray} \label{mlec}
\mathcal{L}(\pv_1) &=& n \log(\a_1 \b_1 \th_1 \a_2 \b_2 \th_2) + n \log(\delta_1 +1) + 
\sum_{\i=1}^{\ss} \log(e^{\th_1 \x_{1i}} -1) +  \sum_{\i=1}^{\ss} \log(e^{\th_2 \x_{2i}} -1) + \sum_{\i=1}^{\ss} \phi(\x_{1i}; \th_1) \nonumber \\ && + \sum_{\i=1}^{\ss} \phi(\x_{2i}; \th_2) + (\a_1-1)\sum_{\i=1}^{\ss}\log y_{1i} + (\a_2-1)\sum_{\i=1}^{\ss}\log y_{2i}  + (\b_1-1)\sum_{\i=1}^{\ss} \log\left(1- y_{1i} ^{\a_1} \right) \nonumber \\ && + (\b_2-1)\sum_{\i=1}^{\ss} \log\left(1- y_{2i} ^{\a_2} \right) - (\delta_1 + 1)\left( \sum_{\i=1}^{\ss} \log \left( 1 - \left(1- y_{1i} ^{a_1} \right) ^{\b_1} \right) + \sum_{\i=1}^{\ss} \log \left( 1 - \left(1- y_{2i} ^{a_2} \right) ^{\b_2} \right) \right) \nonumber \\ && - \frac{2 \delta_1 + 1}{\delta_1} \sum_{\i=1}^{\ss} \log \left( \left( 1 - \left(1- y_{1i} ^{a_1} \right) ^{\b_1} \right)^{-\delta_1} + \left( 1 - \left(1- y_{2i} ^{a_2} \right) ^{\b_2} \right)^{-\delta_1} - 1 \right)
\end{eqnarray}
where, $ 1-e^{\phi(\x_{1i}; \th_1)} = y_{1i}$ and where, $ 1-e^{\phi(\x_{2i}; \th_2)} = y_{2i}$.

Furthermore, suppose $\{(X_{1i}, X_{2i})\}_{i=1}^n$ represent a random sample from the GBKT($\pv_2$) distribution. The log-likelihood function $L$, derived from PDF in equation \eqref{gbktpdf}, is given by:
\begin{eqnarray} \label{mleg}
\mathcal{L}(\pv_2) &=& n \log(\a_1 \b_1 \th_1 \a_2 \b_2 \th_2) + 
\sum_{\i=1}^{\ss} \log(e^{\th_1 \x_{1i}} -1) +  \sum_{\i=1}^{\ss} \log(e^{\th_2 \x_{2i}} -1) + \sum_{\i=1}^{\ss} \phi(\x_{1i}; \th_1) + \sum_{\i=1}^{\ss} \phi(\x_{2i}; \th_2)  \nonumber \\ && + (\a_1-1)\sum_{\i=1}^{\ss}\log y_{1i} + (\a_2-1)\sum_{\i=1}^{\ss}\log y_{2i}  + (\b_1-1)\sum_{\i=1}^{\ss} \log\left(1- y_{1i} ^{\a_1} \right) + (\b_2-1)\sum_{\i=1}^{\ss} \log\left(1- y_{2i} ^{\a_2} \right) \nonumber \\ &&  - (2-1/{\delta_2}) \left( \sum_{\i=1}^{\ss} \log \left[ \left\{-\log \left( 1 - \left(1- y_{1i} ^{a_1} \right) ^{\b_1} \right) \right\}^{\delta_2} + \left\{-\log \left( 1 - \left(1- y_{2i} ^{a_2} \right) ^{\b_2} \right) \right\}^{\delta_2} \right] \right) \nonumber \\ &&+ \sum_{\i=1}^{\ss} \log \left(\left[ \left\{-\log \left( 1 - \left(1- y_{1i} ^{a_1} \right) ^{\b_1} \right) \right\}^{\delta_2} + \left\{-\log \left( 1 - \left(1- y_{2i} ^{a_2} \right) ^{\b_2} \right) \right\}^{\delta_2} \right]^{1/{\delta_2}} + {\delta_2} -1 \right) \nonumber \\ && - (1-{\delta_2}) \left( \sum_{\i=1}^{\ss} \log \left[ \left\{\log \left( 1 - \left(1- y_{1i} ^{a_1} \right) ^{\b_1} \right) \right\}  \left\{\log \left( 1 - \left(1- y_{2i} ^{a_2} \right) ^{\b_2} \right) \right\} \right] \right)  \nonumber \\ && -  \sum_{\i=1}^{\ss} \log \left[ \left( 1 - \left(1- y_{1i} ^{a_1} \right) ^{\b_1} \right)  \left( 1 - \left(1- y_{2i} ^{a_2} \right) ^{\b_2} \right) \right] \nonumber \\ && -  \sum_{\i=1}^{\ss} \left[ \left\{-\log \left( 1 - \left(1- y_{1i} ^{a_1} \right) ^{\b_1} \right) \right\}^{\delta_2} + \left\{-\log \left( 1 - \left(1- y_{2i} ^{a_2} \right) ^{\b_2} \right) \right\}^{\delta_2} \right]^{1/{\delta_2}} 
\end{eqnarray}

For the CBKT distribution, the parameter vector $\pv_1 = (a_1, b_1, \th_1, a_2, b_2, \th_2, \delta_1)$ is estimated by maximizing the likelihood function presented in equation \eqref{mlec}. Similarly, parameter estimation for the GBKT distribution involves optimizing the likelihood function in equation \eqref{mleg} to obtain estimates for the parameter vector $\pv_2 = (a_1, b_1, \th_1, a_2, b_2, \th_2, \delta_2)$. The analytical complexity inherent in both likelihood functions $\mathcal{L}(\pv_1)$ and $\mathcal{L}(\pv_2)$, combined with the high-dimensional parameter space comprising seven parameters each, presents substantial theoretical challenges. Specifically, the intricate functional forms preclude straightforward analytical proofs regarding the existence, uniqueness, and asymptotic properties of the maximum likelihood estimators. The absence of tractable closed-form solutions further necessitates the application of sophisticated numerical optimization techniques to solve the corresponding systems of seven nonlinear estimating equations detailed in Appendix A \& B. Computational implementation was achieved using the \texttt{optim()} function within the R software with a suitable set of initial values.

\subsection{Inference functions for margins (IFM)}

The Inference Functions for Margins (IFM) method employs a two-stage estimation procedure for parameter determination. This approach first estimates the marginal distributions independently, then subsequently estimates the copula dependence structure using the marginal parameter estimates. For comprehensive details on the IFM methodology, refer to \citep{joe1996estimation}. \\
\textbf{Step 1: Marginal Distribution Estimation}\\
Given the KTD marginal distributions within the CBKT and GBKT frameworks, the log-likelihood functions for the marginal distributions, denoted as $L_j$ for $j = 1, 2$, are formulated as follows:
\begin{eqnarray} \label{MLE}
	\mathcal{L}(\pv) &=& \ss \log(\a_j \b_j \th_j) + \sum_{\i=1}^{\ss} \log(e^{\phi(\x_{ji}; \th_j)} -1) +  \sum_{\i=1}^{\ss} \phi(\x_{ji}; \th_j) + (\a_j-1)\sum_{\i=1}^{\ss}\log \left(1-e^{\phi(\x_{ji}; \th_j)} \right) \nonumber \\ &&  + (\b_j-1)\sum_{\i=1}^{\ss} \log\left(1- \left(1-e^{\phi(\x_{ji}; \th_j)} \right) ^{\a_j} \right)  \ ; \ j=1,2
\end{eqnarray}
The maximum likelihood estimates for the marginal parameters $(a_1, b_1, \theta_1)$ and $(a_2, b_2, \theta_2)$ are obtained by independently maximizing the log-likelihood functions $\mathcal{L}(\pv)$.\\
\textbf{Step 2: Copula Dependence Parameter Estimation}\\
Building upon the marginal parameter estimates obtained in Step 1, the copula dependence parameters for both models are subsequently estimated. For the CBKT model, the dependence parameter $\delta_1$ is estimated by maximizing the copula log-likelihood function specified in equation \eqref{ifmc}. 
\begin{eqnarray} \label{ifmc}
\mathcal{L}(\pv_1) &=& n \log(\hat{a_1} \hat{b_1} \hat{\th_1} \hat{a_2} \hat{b_2} \hat{\th_2}) + n \log(\delta_1 +1) + 
\sum_{\i=1}^{\ss} \log(e^{\hat{\th_1} \x_{1i}} -1) +  \sum_{\i=1}^{\ss} \log(e^{\hat{\th_2} \x_{2i}} -1) + \sum_{\i=1}^{\ss} \phi(\x_{1i}; \hat{\th_1}) \nonumber \\ && + \sum_{\i=1}^{\ss} \phi(\x_{2i}; \hat{\th_2}) + (\hat{a_1}-1)\sum_{\i=1}^{\ss}\log y_{1i} + (\hat{a_2}-1)\sum_{\i=1}^{\ss}\log y_{2i}  + (\hat{b_1}-1)\sum_{\i=1}^{\ss} \log\left(1- y_{1i} ^{\hat{a_1}} \right) \nonumber \\ && + (\hat{b_2}-1)\sum_{\i=1}^{\ss} \log\left(1- y_{2i} ^{\hat{a_2}} \right) - (\delta_1 + 1)\left( \sum_{\i=1}^{\ss} \log \left( 1 - \left(1- y_{1i} ^{\hat{a_1}} \right) ^{\hat{b_1}} \right) + \sum_{\i=1}^{\ss} \log \left( 1 - \left(1- y_{2i} ^{\hat{a_2}} \right) ^{\hat{b_2}} \right) \right) \nonumber \\ && - \frac{2 \delta_1 + 1}{\delta_1} \sum_{\i=1}^{\ss} \log \left( \left( 1 - \left(1- y_{1i} ^{\hat{a_1}} \right) ^{\hat{b_1}} \right)^{-\delta_1} + \left( 1 - \left(1- y_{2i} ^{\hat{a_2}} \right) ^{\hat{b_2}} \right)^{-\delta_1} - 1 \right)
\end{eqnarray}

Correspondingly, the GBKT model's dependence parameter $\delta_2$ is derived through optimization of the copula log-likelihood function presented in equation \eqref{ifmg}.
\begin{eqnarray} \label{ifmg}
\mathcal{L}(\pv_2) &=& n \log(\hat{a_1} \hat{b_1} \hat{\th_1} \hat{a_2} \hat{b_2} \hat{\th_2}) + 
\sum_{\i=1}^{\ss} \log(e^{\hat{\th_1} \x_{1i}} -1) +  \sum_{\i=1}^{\ss} \log(e^{\hat{\th_2} \x_{2i}} -1) + \sum_{\i=1}^{\ss} \phi(\x_{1i}; \hat{\th_1}) + \sum_{\i=1}^{\ss} \phi(\x_{2i}; \hat{\th_2})  \nonumber \\ && + (\hat{a_1}-1)\sum_{\i=1}^{\ss}\log y_{1i} + (\hat{a_2}-1)\sum_{\i=1}^{\ss}\log y_{2i}  + (\hat{b_1}-1)\sum_{\i=1}^{\ss} \log\left(1- y_{1i} ^{\hat{a_1}} \right) + (\hat{b_2}-1)\sum_{\i=1}^{\ss} \log\left(1- y_{2i} ^{\hat{a_2}} \right) \nonumber \\ &&  - (2-1/{\delta_2}) \left( \sum_{\i=1}^{\ss} \log \left[ \left\{-\log \left( 1 - \left(1- y_{1i} ^{\hat{a_1}} \right) ^{\hat{b_1}} \right) \right\}^{\delta_2} + \left\{-\log \left( 1 - \left(1- y_{2i} ^{\hat{a_2}} \right) ^{\hat{b_2}} \right) \right\}^{\delta_2} \right] \right) \nonumber \\ &&+ \sum_{\i=1}^{\ss} \log \left(\left[ \left\{-\log \left( 1 - \left(1- y_{1i} ^{\hat{a_1}} \right) ^{\hat{b_1}} \right) \right\}^{\delta_2} + \left\{-\log \left( 1 - \left(1- y_{2i} ^{\hat{a_2}} \right) ^{\hat{b_2}} \right) \right\}^{\delta_2} \right]^{1/{\delta_2}} + {\delta_2} -1 \right) \nonumber \\ && - (1-{\delta_2}) \left( \sum_{\i=1}^{\ss} \log \left[ \left\{\log \left( 1 - \left(1- y_{1i} ^{\hat{a_1}} \right) ^{\hat{b_1}} \right) \right\}  \left\{\log \left( 1 - \left(1- y_{2i} ^{\hat{a_2}} \right) ^{\hat{b_2}} \right) \right\} \right] \right)  \nonumber \\ && -  \sum_{\i=1}^{\ss} \log \left[ \left( 1 - \left(1- y_{1i} ^{\hat{a_1}} \right) ^{\hat{b_1}} \right)  \left( 1 - \left(1- y_{2i} ^{\hat{a_2}} \right) ^{\hat{b_2}} \right) \right] \nonumber \\ && -  \sum_{\i=1}^{\ss} \left[ \left\{-\log \left( 1 - \left(1- y_{1i} ^{\hat{a_1}} \right) ^{\hat{b_1}} \right) \right\}^{\delta_2} + \left\{-\log \left( 1 - \left(1- y_{2i} ^{\hat{a_2}} \right) ^{\hat{b_2}} \right) \right\}^{\delta_2} \right]^{1/{\delta_2}} 
\end{eqnarray}
Consistent with the marginal estimation procedure, both dependence parameters $\delta_1$ and $\delta_2$ are computed numerically using nonlinear optimization algorithms.

\section{Simulation Study}\label{simu}
In this section, we present a Monte Carlo simulation study to estimate parameters of the proposed copula-based Kumaraswamy–Teissier distributions (CBKT and GBKT). The primary objectives are twofold: (i) to evaluate the accuracy of the Maximum Likelihood Estimator (MLE) and the Inference Function for Margins (IFM) in estimating the model parameters, and (ii) to investigate the behavior of these estimators across different sample sizes. A total of 1,000 random samples are generated with sizes $n=50, 100, 200,$ and $300$. For the simulation design, random variates from the CBKT and GBKT models are generated using the conditional distribution approach described by \citep{nelsen2006introduction}. The procedure for data generation is identical for both copula models, differing only in the functional forms of the copula. The simulation steps can be summarized as follows: 

\begin{table}[!ht]
\centering

\caption{\small Mean, Bias and RMSE of parameter estimates of CBKT under MLE and IFM}
\label{table1}
\resizebox{1.1\textwidth}{!}
{%
\begin{tabular}{@{}lccccccccccccccc@{}}
\toprule
\toprule
 & n & \multicolumn{7}{c}{MLE} & \multicolumn{7}{c}{IFM}  \\ 
\cmidrule(lr){3-9} \cmidrule(lr){10-16} 
{} & {} & $\hat{a_1}$ & $\hat{b_1}$ & $\hat{\theta_1}$ & $\hat{a_2}$ & $\hat{b_2}$ & $\hat{\theta_2}$ & $\hat{\delta_1}$ &  
$\hat{a_1}$ & $\hat{b_1}$ & $\hat{\theta_1}$ & $\hat{a_2}$ & $\hat{b_2}$ & $\hat{\theta_2}$ & $\hat{\delta_1}$ \\
\midrule
\multicolumn{16}{c}{True Values: $a_1=1.5$, $b_1=1.8$, $\theta_1=1.4$, $a_2=1.6$, $b_2=1.9$, $\theta_2=1.3$, $\delta_1=1.2$}  \\
\midrule

\multirow{4}{*}{Mean} 
& 50 & 1.5543 & 1.8399 & 1.4893 & 1.6566 & 1.9051 & 1.3867 & 1.2663 
 & 1.5627 & 1.8638 & 1.4843 & 1.6613 & 1.9062 & 1.3870 & 1.2394 \\ 
& 100 & 1.5338 & 1.8520 & 1.4745 & 1.6340 & 1.8715 & 1.3853 & 1.2343 
 & 1.5366 & 1.8143 & 1.4848 & 1.6387 & 1.8692 & 1.3873 & 1.2195 \\ 
& 200 & 1.5159 & 1.8469 & 1.4688 & 1.6184 & 1.8705 & 1.3790 & 1.2208 
 & 1.5155 & 1.8528 & 1.4704 & 1.6205 & 1.8583 & 1.3840 & 1.2136 \\ 
& 300 & 1.5090 & 1.8781 & 1.4572 & 1.6113 & 1.9117 & 1.3664 & 1.2178 
 & 1.5085 & 1.8892 & 1.4579 & 1.6122 & 1.9134 & 1.3670 & 1.2131 \\ 
\midrule

\multirow{4}{*}{Bias} 
& 50 & 0.0543 & 0.0399 & 0.0893 & 0.0566 & 0.0051 & 0.0867 & 0.0663 
 & 0.0627 & 0.0638 & 0.0843 & 0.0613 & 0.0062 & 0.0870 & 0.0394 \\ 
& 100 & 0.0338 & 0.0520 & 0.0745 & 0.0340 & 0.0285 & 0.0853 & 0.0343 
 & 0.0366 & 0.0143 & 0.0848 & 0.0387 & 0.0308 & 0.0873 & 0.0195 \\ 
& 200 & 0.0159 & 0.0469 & 0.0688 & 0.0184 & 0.0295 & 0.0790 & 0.0208 
 & 0.0155 & 0.0528 & 0.0704 & 0.0205 & 0.0417 & 0.0840 & 0.0136 \\ 
& 300 & 0.0090 & 0.0781 & 0.0572 & 0.0113 & 0.0117 & 0.0664 & 0.0178 
 & 0.0085 & 0.0892 & 0.0579 & 0.0122 & 0.0134 & 0.0670 & 0.0131 \\ 
\midrule

\multirow{4}{*}{RMSE} 
& 50 & 0.2419 & 0.9451 & 0.2631 & 0.2477 & 0.9428 & 0.2350 & 0.3502 
 & 0.2550 & 0.9477 & 0.2608 & 0.2557 & 0.9439 & 0.2378 & 0.3362 \\ 
& 100 & 0.1822 & 0.9323 & 0.2530 & 0.1873 & 0.9311 & 0.2329 & 0.2622 
 & 0.1881 & 0.9248 & 0.2565 & 0.1928 & 0.9351 & 0.2340 & 0.2568 \\ 
& 200 & 0.1297 & 0.9119 & 0.2443 & 0.1390 & 0.9182 & 0.2274 & 0.1869 
 & 0.1330 & 0.9314 & 0.2496 & 0.1425 & 0.9274 & 0.2311 & 0.1855 \\ 
& 300 & 0.1062 & 0.9122 & 0.2383 & 0.1131 & 0.9129 & 0.2223 & 0.1579 
 & 0.1095 & 0.9354 & 0.2435 & 0.1148 & 0.9216 & 0.2230 & 0.1570 \\ 
\midrule
\midrule
\multicolumn{16}{c}{True Values: $a_1=1.2$, $b_1=2.6$, $\theta_1=0.8$, $a_2=1.1$, $b_2=2.3$, $\theta_2=0.7$, $\delta_1=1.9$}  \\
\midrule
\multirow{4}{*}{Mean} 
& 50  & 1.2430 & 2.5496 & 0.8266 & 1.1392 & 2.2558 & 0.7262 & 1.9326  
      & 1.2441 & 2.5260 & 0.8295 & 1.1401 & 2.2399 & 0.7277 & 1.8945 \\ 
& 100 & 1.2198 & 2.5818 & 0.8147 & 1.1156 & 2.2697 & 0.7171 & 1.9304  
      & 1.2218 & 2.5741 & 0.8162 & 1.1184 & 2.2638 & 0.7186 & 1.9029 \\ 
& 200 & 1.2067 & 2.5773 & 0.8118 & 1.1033 & 2.2619 & 0.7142 & 1.9290  
      & 1.2089 & 2.5777 & 0.8126 & 1.1056 & 2.2771 & 0.7129 & 1.9122 \\ 
& 300 & 1.2027 & 2.5787 & 0.8099 & 1.1005 & 2.2980 & 0.7088 & 1.9262  
      & 1.2046 & 2.5789 & 0.8106 & 1.1012 & 2.2848 & 0.7106 & 1.9149 \\ 
\midrule

\multirow{4}{*}{Bias} 
& 50  & 0.0430 & 0.0504 & 0.0266 & 0.0392 & 0.0442 & 0.0262 & 0.0326 
      & 0.0441 & 0.0740 & 0.0295 & 0.0401 & 0.0601 & 0.0277 & 0.0055 \\ 
& 100 & 0.0198 & 0.0182 & 0.0147 & 0.0156 & 0.0303 & 0.0171 & 0.0304 
      & 0.0218 & 0.0259 & 0.0162 & 0.0184 & 0.0362 & 0.0186 & 0.0029 \\ 
& 200 & 0.0067 & 0.0227 & 0.0118 & 0.0033 & 0.0381 & 0.0142 & 0.0290 
      & 0.0089 & 0.0223 & 0.0126 & 0.0056 & 0.0229 & 0.0129 & 0.0122 \\ 
& 300 & 0.0027 & 0.0213 & 0.0099 & 0.0005 & 0.0020 & 0.0088 & 0.0262 
      & 0.0046 & 0.0211 & 0.0106 & 0.0012 & 0.0152 & 0.0106 & 0.0149 \\ 
\midrule

\multirow{4}{*}{RMSE} 
& 50  & 0.1652 & 0.4771 & 0.0860 & 0.1532 & 0.4858 & 0.0838 & 0.3614  
      & 0.1843 & 0.4847 & 0.0922 & 0.1741 & 0.4890 & 0.0848 & 0.3570 \\ 
& 100 & 0.1178 & 0.4743 & 0.0683 & 0.1076 & 0.4819 & 0.0701 & 0.3039  
      & 0.1275 & 0.4817 & 0.0713 & 0.1177 & 0.4886 & 0.0727 & 0.3004 \\ 
& 200 & 0.0832 & 0.4702 & 0.0617 & 0.0766 & 0.4785 & 0.0639 & 0.2351  
      & 0.0887 & 0.4820 & 0.0640 & 0.0814 & 0.4829 & 0.0639 & 0.2334 \\ 
& 300 & 0.0697 & 0.4622 & 0.0585 & 0.0639 & 0.4751 & 0.0607 & 0.2063  
      & 0.0739 & 0.4745 & 0.0603 & 0.0665 & 0.4827 & 0.0611 & 0.2059 \\ 

\bottomrule

\bottomrule
\end{tabular}%
}
\end{table}

\begin{table}[!ht]
\centering

\caption{\small Mean, Bias and RMSE of parameter estimates of GBKT under MLE and IFM}
\label{table2}
\resizebox{1.1\textwidth}{!}{%
\begin{tabular}{@{}lccccccccccccccc@{}}
\toprule
\toprule
 & n & \multicolumn{7}{c}{MLE} & \multicolumn{7}{c}{IFM}  \\ 
\cmidrule(lr){3-9} \cmidrule(lr){10-16} 
{} & {} & $\hat{a_1}$ & $\hat{b_1}$ & $\hat{\theta_1}$ & $\hat{a_2}$ & $\hat{b_2}$ & $\hat{\theta_2}$ & $\hat{\delta_2}$ &  
$\hat{a_1}$ & $\hat{b_1}$ & $\hat{\theta_1}$ & $\hat{a_2}$ & $\hat{b_2}$ & $\hat{\theta_2}$ & $\hat{\delta_2}$ \\
\midrule
\midrule
\multicolumn{16}{c}{True Values: $a_1=2.4$, $b_1=1.7$, $\theta_1=1.8$, $a_2=2.3$, $b_2=1.8$, $\theta_2=1.3$, $\delta_2=2.5$}  \\
\midrule

\multirow{4}{*}{Mean} 
& 50  & 1.8385 & 2.4002 & 2.6743 & 1.8605 & 2.1975 & 2.0410 & 1.5319 
      & 2.7725 & 1.6629 & 1.9860 & 2.6564 & 1.6944 & 1.4424 & 2.1276 \\ 
& 100 & 1.8985 & 2.3888 & 2.7755 & 1.9413 & 2.1307 & 2.1476 & 1.2361 
      & 2.7667 & 1.5332 & 2.0231 & 2.6066 & 1.6127 & 1.4515 & 2.1323 \\ 
& 200 & 2.0484 & 2.3010 & 2.7750 & 2.0940 & 2.0211 & 2.1515 & 1.2223 
      & 2.7435 & 1.4600 & 2.0391 & 2.5790 & 1.5792 & 1.4556 & 2.1353 \\ 
& 300 & 2.0951 & 2.2411 & 2.7721 & 2.1351 & 1.9954 & 2.1422 & 1.2117 
      & 2.7333 & 1.4414 & 2.0417 & 2.5691 & 1.5583 & 1.4588 & 2.1362 \\ 
\midrule

\multirow{4}{*}{Bias} 
& 50  & 0.5615 & 0.7002 & 0.8743 & 0.4395 & 0.3975 & 0.7410 & 0.9681 
      & 0.3725 & 0.0371 & 0.1860 & 0.3564 & 0.1056 & 0.1424 & 0.3724 \\ 
& 100 & 0.5015 & 0.6888 & 0.9755 & 0.3587 & 0.3307 & 0.8476 & 1.2639 
      & 0.3667 & 0.1668 & 0.2231 & 0.3066 & 0.1873 & 0.1515 & 0.3677 \\ 
& 200 & 0.3516 & 0.6010 & 0.9750 & 0.2060 & 0.2211 & 0.8515 & 1.2777 
      & 0.3435 & 0.2400 & 0.2391 & 0.2790 & 0.2208 & 0.1556 & 0.3647 \\ 
& 300 & 0.3049 & 0.5411 & 0.9721 & 0.1649 & 0.1954 & 0.8422 & 1.2883 
      & 0.3333 & 0.2586 & 0.2417 & 0.2691 & 0.2417 & 0.1588 & 0.3638 \\ 
\midrule

\multirow{4}{*}{RMSE} 
& 50  & 0.8004 & 0.8536 & 0.9897 & 0.7643 & 0.6630 & 0.8620 & 1.6572 
      & 0.5826 & 0.9387 & 0.3823 & 0.5657 & 0.9519 & 0.2699 & 0.3753 \\ 
& 100 & 0.6630 & 0.7489 & 0.9869 & 0.6182 & 0.5059 & 0.8634 & 1.3290 
      & 0.5304 & 0.9397 & 0.4000 & 0.4757 & 0.9436 & 0.2723 & 0.3688 \\ 
& 200 & 0.5320 & 0.6505 & 0.9806 & 0.4742 & 0.3900 & 0.8612 & 1.3043 
      & 0.4690 & 0.9178 & 0.4027 & 0.4019 & 0.9410 & 0.2723 & 0.3650 \\ 
& 300 & 0.4717 & 0.5823 & 0.9740 & 0.4296 & 0.3541 & 0.8481 & 1.2889 
      & 0.4460 & 0.9157 & 0.4010 & 0.3685 & 0.9391 & 0.2722 & 0.3640 \\ 
\midrule
\midrule
\multicolumn{16}{c}{True Values: $a_1=1.2$, $b_1=2.6$, $\theta_1=0.8$, $a_2=1.1$, $b_2=2.3$, $\theta_2=0.7$, $\delta_2=1.9$}  \\
\midrule
\multirow{4}{*}{Mean} 
& 50  & 1.1144 & 2.3645 & 0.7989 & 1.0197 & 2.1200 & 0.6794 & 2.7348 
     & 1.2449 & 2.5287 & 0.8297 & 1.1396 & 2.2485 & 0.7269 & 1.8197 \\ 
& 100 & 1.0355 & 2.3835 & 0.8348 & 0.9388 & 2.1191 & 0.7179 & 2.4623 
     & 1.2242 & 2.5737 & 0.8175 & 1.1161 & 2.2459 & 0.7199 & 1.8228 \\ 
& 200 & 0.9558 & 2.3662 & 0.8803 & 0.8621 & 2.0991 & 0.7641 & 2.1663 
     & 1.2117 & 2.5555 & 0.8165 & 1.1055 & 2.2352 & 0.7185 & 1.8259 \\ 
& 300 & 0.9027 & 2.3196 & 0.9204 & 0.8137 & 2.1009 & 0.7966 & 1.9535 
     & 1.2079 & 2.5597 & 0.8144 & 1.1028 & 2.2474 & 0.7160 & 1.8266 \\ 
\midrule

\multirow{4}{*}{Bias} 
& 50  & 0.0856 & 0.2355 & 0.0011 & 0.0803 & 0.1800 & 0.0206 & 0.8348 
     & 0.0449 & 0.0713 & 0.0297 & 0.0396 & 0.0515 & 0.0269 & 0.0803 \\ 
& 100 & 0.1645 & 0.2165 & 0.0348 & 0.1612 & 0.1809 & 0.0179 & 0.5623 
     & 0.0242 & 0.0263 & 0.0175 & 0.0161 & 0.0541 & 0.0199 & 0.0772 \\ 
& 200 & 0.2442 & 0.2338 & 0.0803 & 0.2379 & 0.2009 & 0.0641 & 0.2663 
     & 0.0117 & 0.0445 & 0.0165 & 0.0055 & 0.0648 & 0.0185 & 0.0741 \\ 
& 300 & 0.2973 & 0.2804 & 0.1204 & 0.2863 & 0.1991 & 0.0966 & 0.0535 
     & 0.0079 & 0.0403 & 0.0144 & 0.0028 & 0.0526 & 0.0160 & 0.0734 \\ 
\midrule

\multirow{4}{*}{RMSE} 
& 50  & 0.4231 & 0.4814 & 0.3515 & 0.4169 & 0.4830 & 0.3388 & 2.3263 
     & 0.1826 & 0.4860 & 0.0893 & 0.1713 & 0.4915 & 0.0826 & 0.0927 \\ 
& 100 & 0.4340 & 0.4754 & 0.3283 & 0.4307 & 0.4731 & 0.3163 & 2.1190 
     & 0.1285 & 0.4858 & 0.0710 & 0.1196 & 0.4843 & 0.0706 & 0.0829 \\ 
& 200 & 0.4453 & 0.4697 & 0.3060 & 0.4391 & 0.4619 & 0.2942 & 1.8684 
     & 0.0881 & 0.4786 & 0.0642 & 0.0792 & 0.4792 & 0.0637 & 0.0769 \\ 
& 300 & 0.4512 & 0.4623 & 0.2941 & 0.4408 & 0.4476 & 0.2767 & 1.6653 
     & 0.0742 & 0.4738 & 0.0598 & 0.0663 & 0.4764 & 0.0604 & 0.0752 \\ 

\bottomrule
\bottomrule
\end{tabular}%
}
\end{table}

\begin{steps}
    \item Fix the values of all the parameters. 
    \item Generate two independent random variables, $\mathscr{U}$ and $\mathscr{R}$, from the uniform distribution $U(0,1)$. 
    \item For the Clayton copula, compute $\mathscr{V}$ as
    \begin{eqnarray} \label{cda}
   \mathscr{V} = \left[ \mathscr{U}^{-\delta_1}  \left( \mathscr{R}^{{-\delta_1}/({\delta_1 +1})} -1  \right) + 1 \right]^{-1/{\delta_1}} .
\end{eqnarray}
    For the Gumbel copula, $\mathscr{V}$ is obtained by
    \begin{eqnarray} \label{gda}
   \mathscr{V} = \exp \left[- \left\{ (\delta_2-1) ~W_0 \left( \frac{ (\mathscr{R} \mathscr{U})^{1/{1-\delta_2}} (-\log \mathscr{U})}{(\delta_2-1)}  \right) - (-\log \mathscr{U})^{\delta_2} \right\}^{1/{\delta_2}} \right],
\end{eqnarray}
    where $W_0(\cdot)$ denotes the principal branch of the Lambert-W function.
    \item The marginal samples from the Kumaraswamy–Teissier distribution are obtained using its quantile function:
    \begin{eqnarray} \label{kda}
x_1 = \frac{1}{\theta_1} \ln \left( -W_{-1} \left( \frac{ \left(1-(1-p)^{1/{b_1}} \right)^{1/{a_1}}-1}{e} \right) \right),
\end{eqnarray}
\begin{eqnarray} \label{kdb}
x_2 = \frac{1}{\theta_2} \ln \left( -W_{-1} \left( \frac{ \left(1-(1-p)^{1/{b_2}} \right)^{1/{a_2}}-1}{e} \right) \right).
\end{eqnarray}
    \item The resulting bivariate observation is $(x_1,x_2)$. Steps (ii)–(iv) are repeated until the desired sample size is obtained.
\end{steps}
For empirical assessment, random samples are drawn from CBKT$(\Theta)$ under two different parameter settings:  
$(a_1=1.5, b_1=1.8, \theta_1=1.4, a_2=1.6, b_2=1.9, \theta_2=1.3, \delta_1=1.2)$ and $(a_1=1.2, b_1=2.6, \theta_1=0.8, a_2=1.1, b_2=2.3, \theta_2=0.7, \delta_1=1.9)$. For each case, the mean, absolute bias, and root mean squared error (RMSE) of the parameter estimates obtained via MLE and IFM are reported in Table~\ref{table1}. Similarly, simulations are performed for GBKT$(\Theta)$ with parameter sets $(a_1=2.4, b_1=1.7, \theta_1=1.8, a_2=2.3, b_2=1.8, \theta_2=1.3, \delta_2=2.5)$ and $(a_1=1.2, b_1=2.6, \theta_1=0.8, a_2=1.1, b_2=2.3, \theta_2=0.7, \delta_2=1.9)$. The results are summarized in Table~\ref{table2}. All computations are carried out in the R statistical software. The findings from both simulation tables indicate that the sample means of the parameter estimates converge to their true values as the sample size increases. Moreover, both bias and RMSE decrease with larger sample sizes, confirming the consistency of the MLE \& IFM estimators. 


\section{Real life application} \label{rla}

\subsection{Data analysis and pre-processing}
We utilize gridded monthly rainfall and temperature data with a spatial resolution of 0.5$^{\circ} \times$ 0.625$^{\circ}$, covering 101 grid points across the Northwest Himalaya (NWH) region, encompassing Jammu and Kashmir (JK), Himachal Pradesh (HP), and Uttarakhand (UK). The dataset is obtained from the National Aeronautics and Space Administration (NASA) database (\url{https://power.larc.nasa.gov/data-access-viewer/}) for the period 1984–2024. For seasonal analysis, the monthly data are further classified into two major climatological seasons: summer, comprising June–September (JJAS), and winter, comprising December–February (DJF). This seasonal segregation allows for a clearer understanding of the distinct hydro-climatic characteristics prevailing during the monsoon-dominated summer period and the western-disturbance-influenced winter period in the NWH region.
Figures~\ref{Figure 3} and~\ref{Figure 4} present the spatial distribution of key summary statistics for temperature and rainfall during the summer and winter seasons, respectively, including the mean, skewness, and kurtosis across the study region. 
\begin{figure}[ht]
	\begin{minipage}[t]{0.33\textwidth}
		\centering
		\includegraphics[width=1\textwidth]{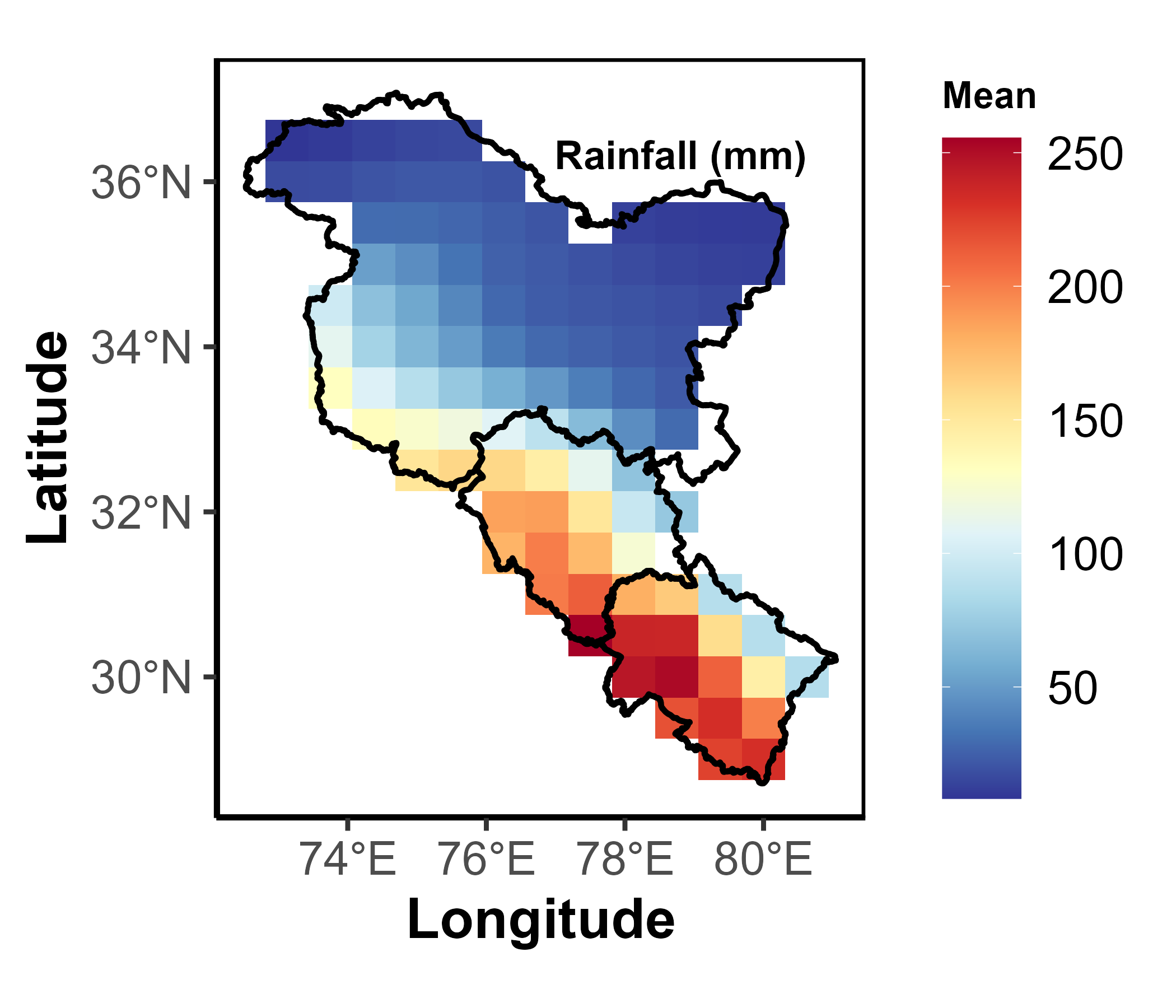}
		\caption*{(a)}
	\end{minipage}\hfill
	\begin{minipage}[t]{0.33\textwidth}
		\centering
		\includegraphics[width=1\textwidth]{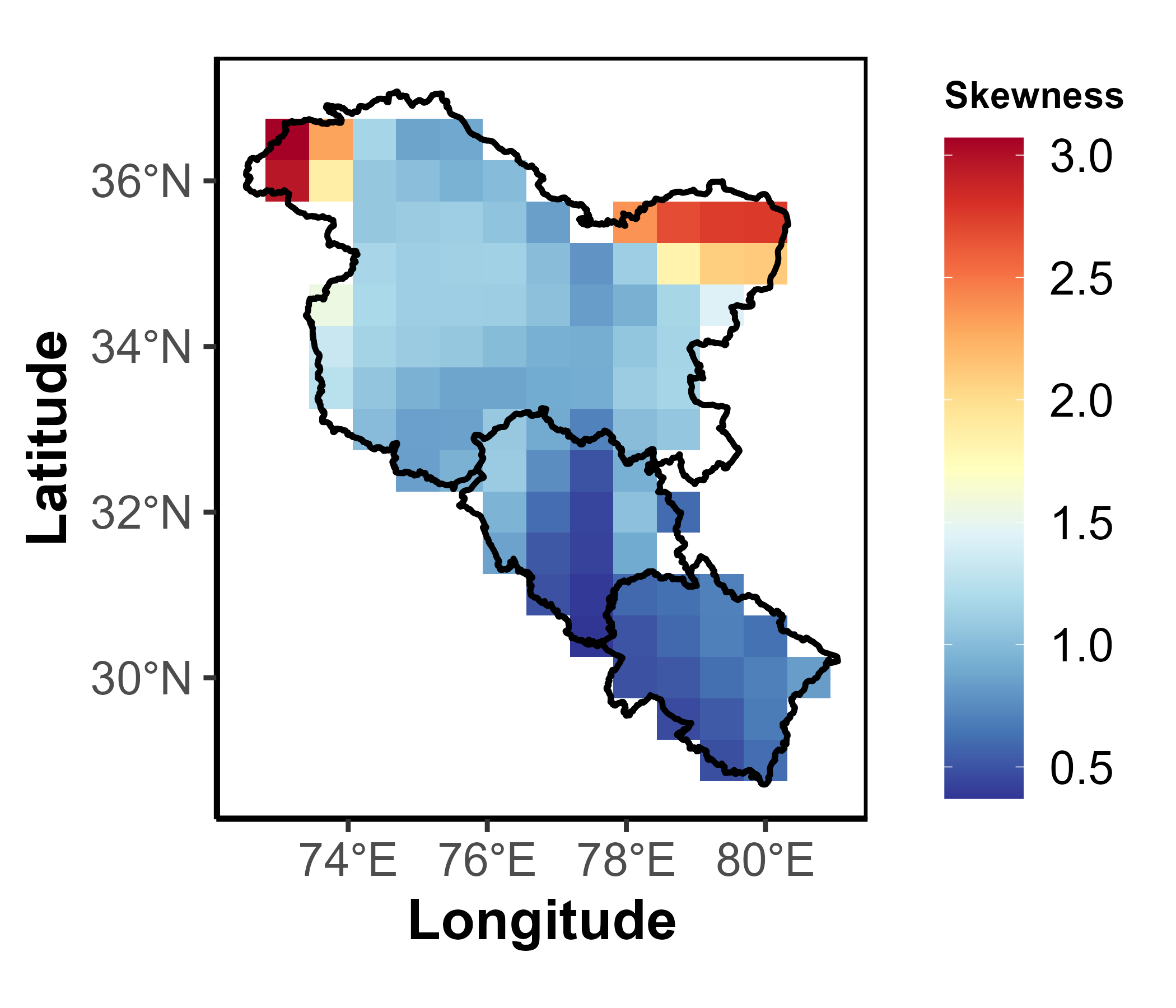}
		\caption*{(b)}
	\end{minipage}\hfill
	\begin{minipage}[t]{0.33\textwidth}
		\centering
		\includegraphics[width=1\textwidth]{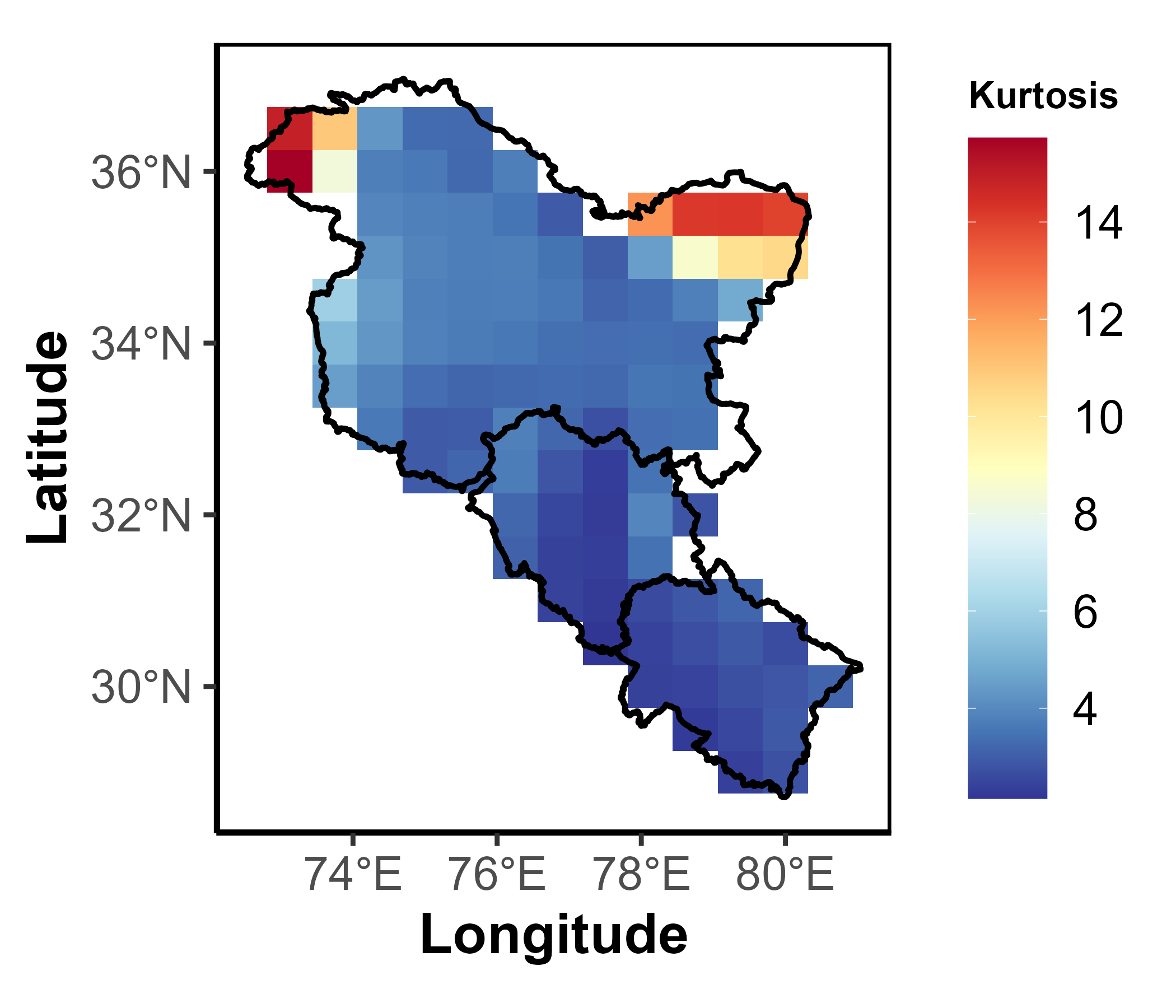}
		\caption*{(c)}
	\end{minipage}
	\medskip
	\begin{minipage}[t]{0.33\textwidth}
		\centering
		\includegraphics[width=1\textwidth]{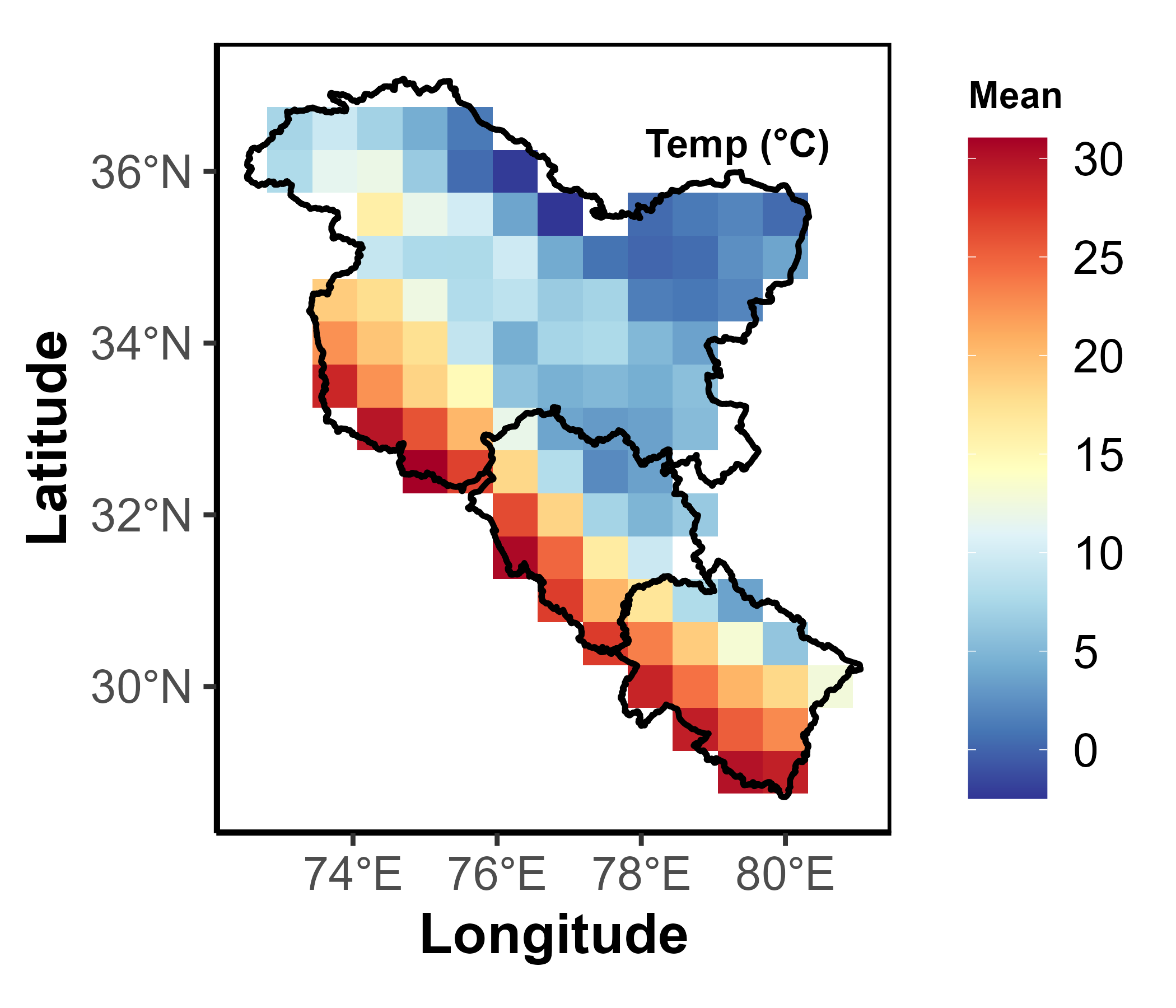}
		\caption*{(d)}
	\end{minipage}\hfill
	\begin{minipage}[t]{0.33\textwidth}
		\centering
		\includegraphics[width=1\textwidth]{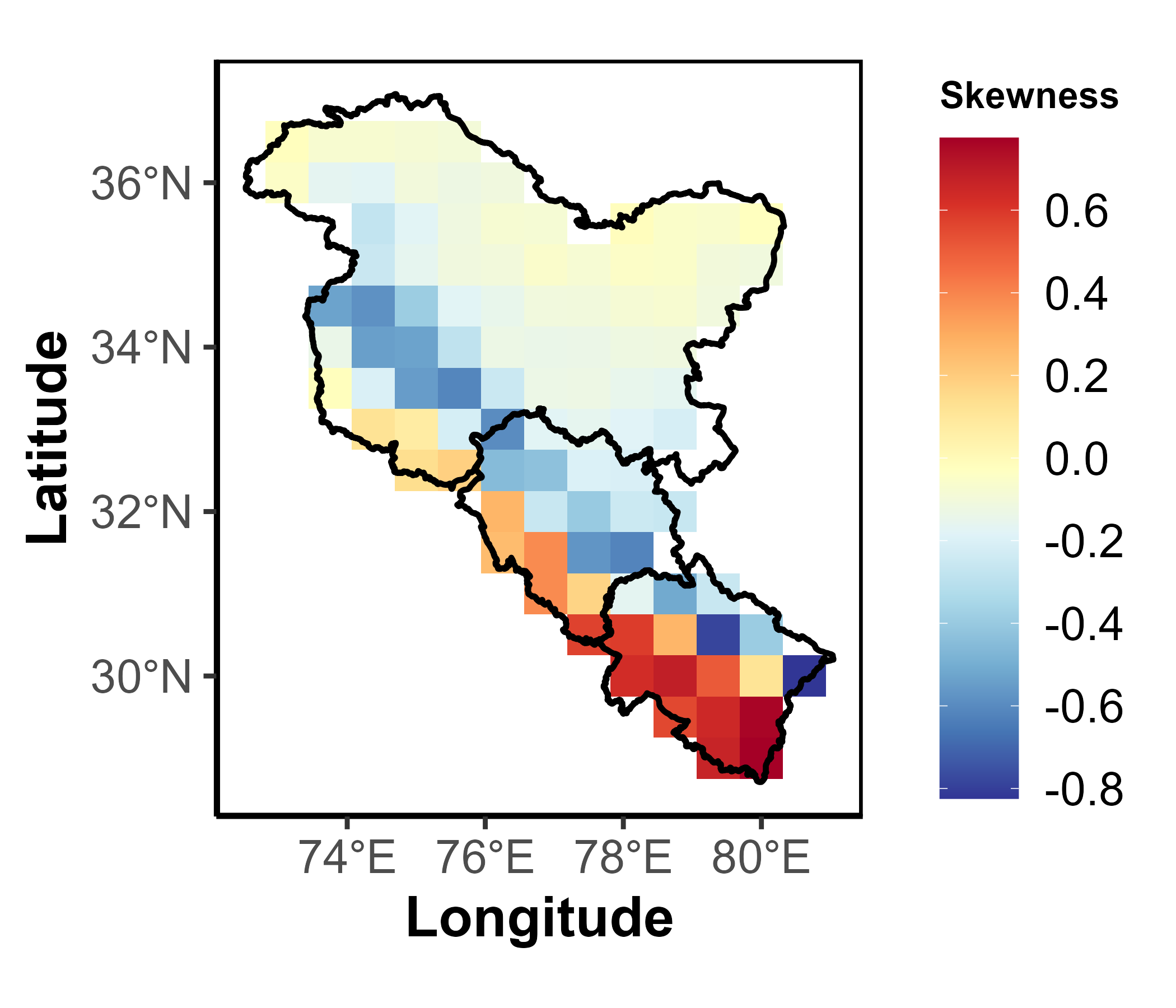}
		\caption*{(e)}
	\end{minipage}\hfill
\begin{minipage}[t]{0.33\textwidth}
	\centering
	\includegraphics[width=1\textwidth]{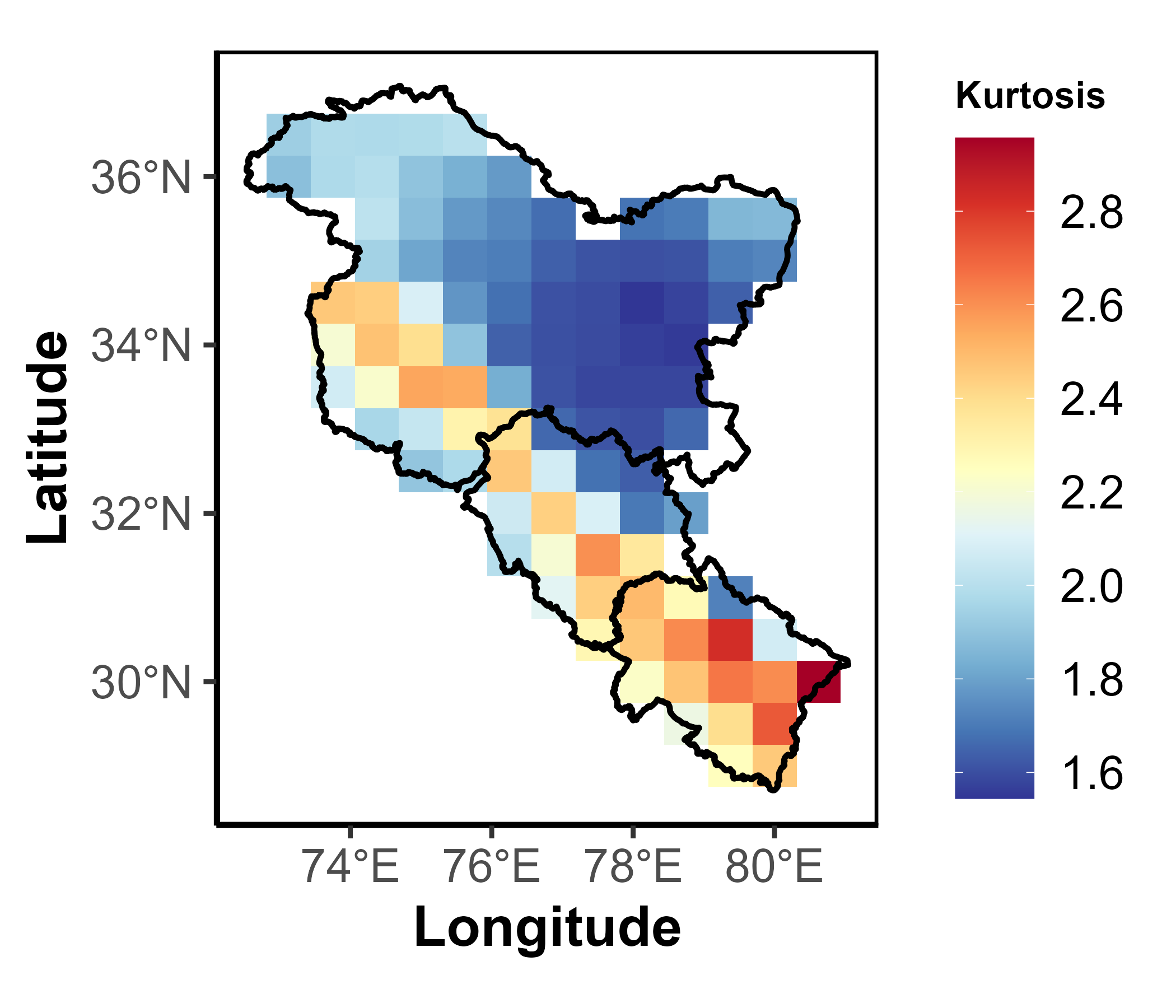}
	\caption*{(f)}
\end{minipage}
	\caption{Descriptive statistics of summer rainfall and temperature across the NWH Himalaya.}
	\label{Figure 3}
\end{figure}

\begin{figure}[ht]
	\begin{minipage}[t]{0.33\textwidth}
		\centering
		\includegraphics[width=1\textwidth]{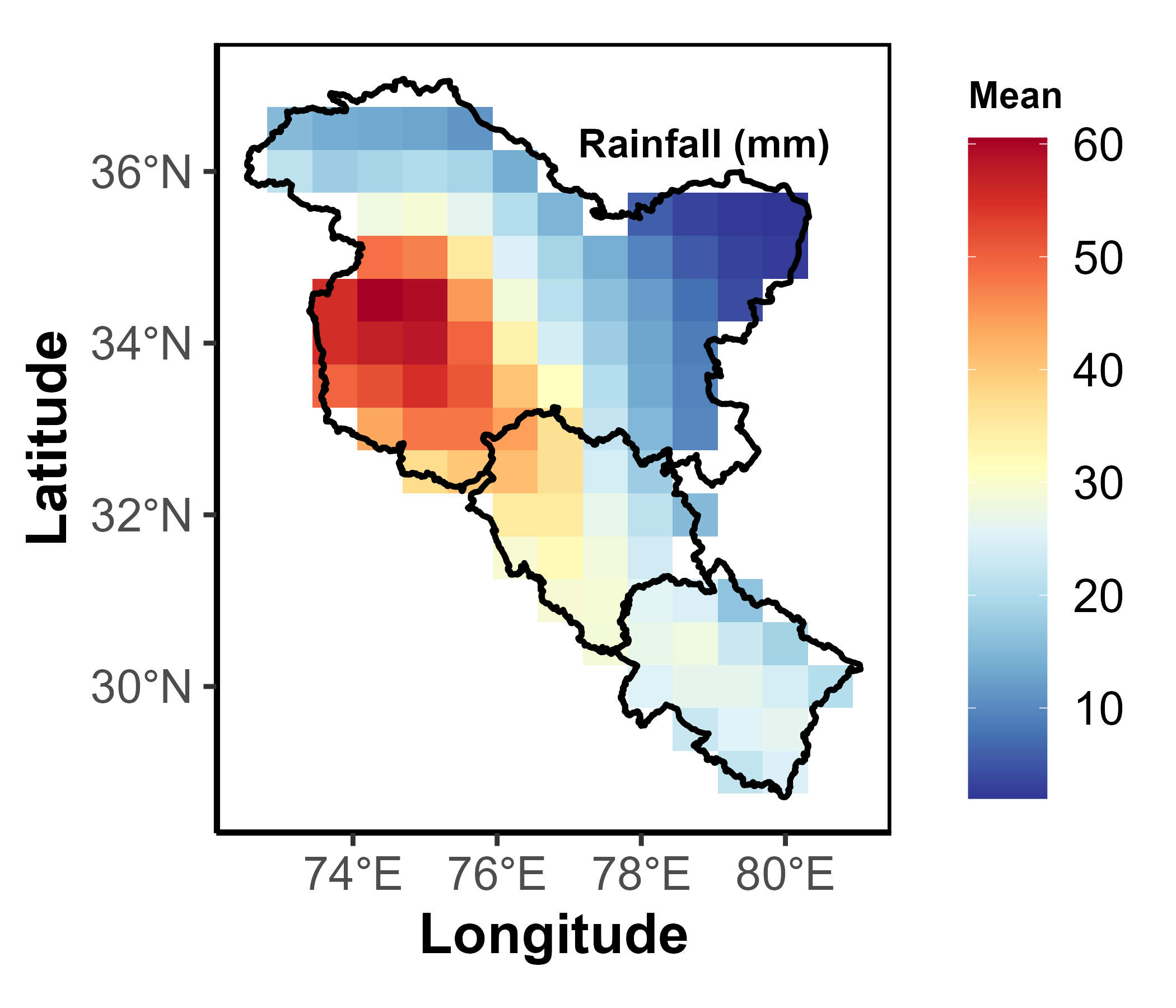}
		\caption*{(a)}
	\end{minipage}\hfill
	\begin{minipage}[t]{0.33\textwidth}
		\centering
		\includegraphics[width=1\textwidth]{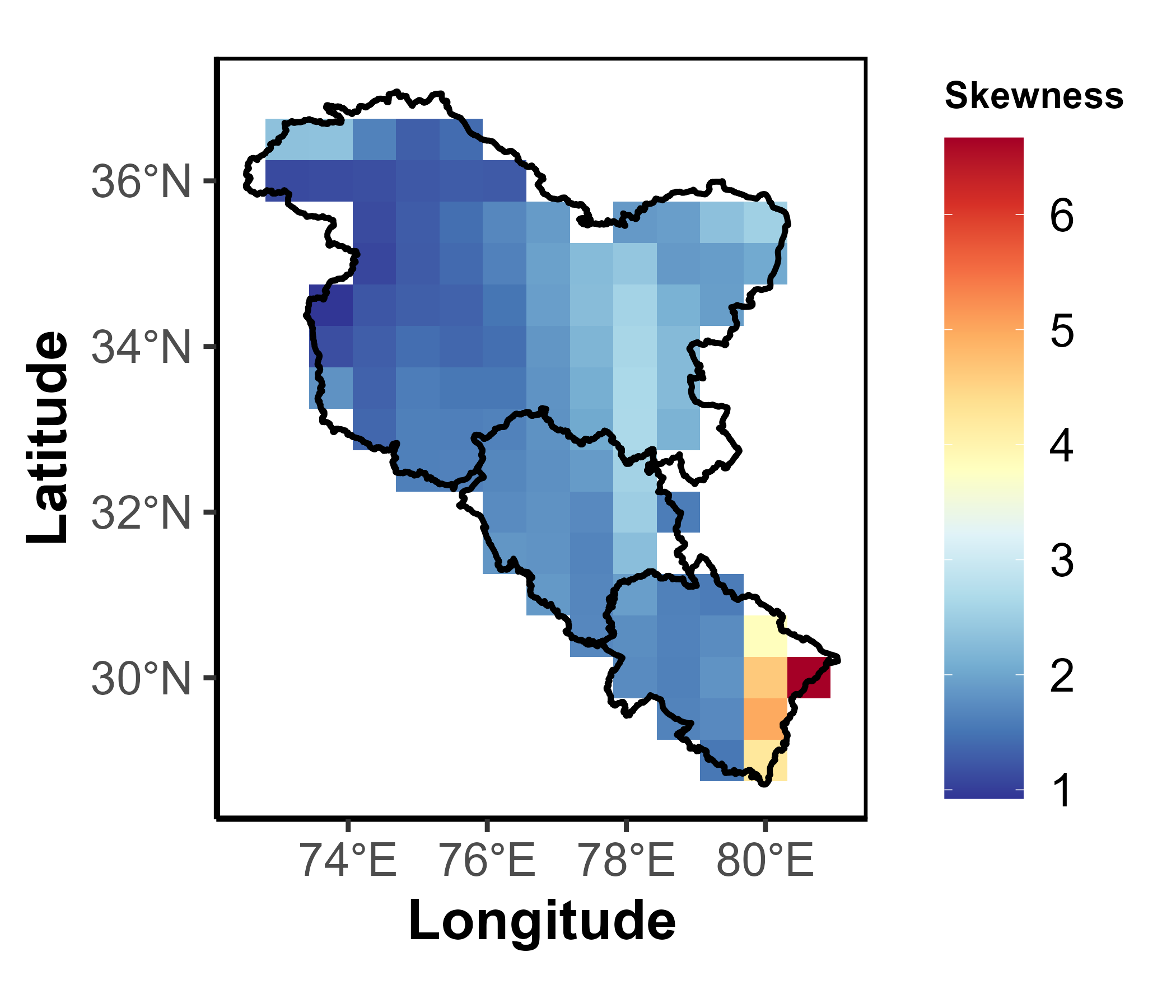}
		\caption*{(b)}
	\end{minipage}\hfill
	\begin{minipage}[t]{0.33\textwidth}
		\centering
		\includegraphics[width=1\textwidth]{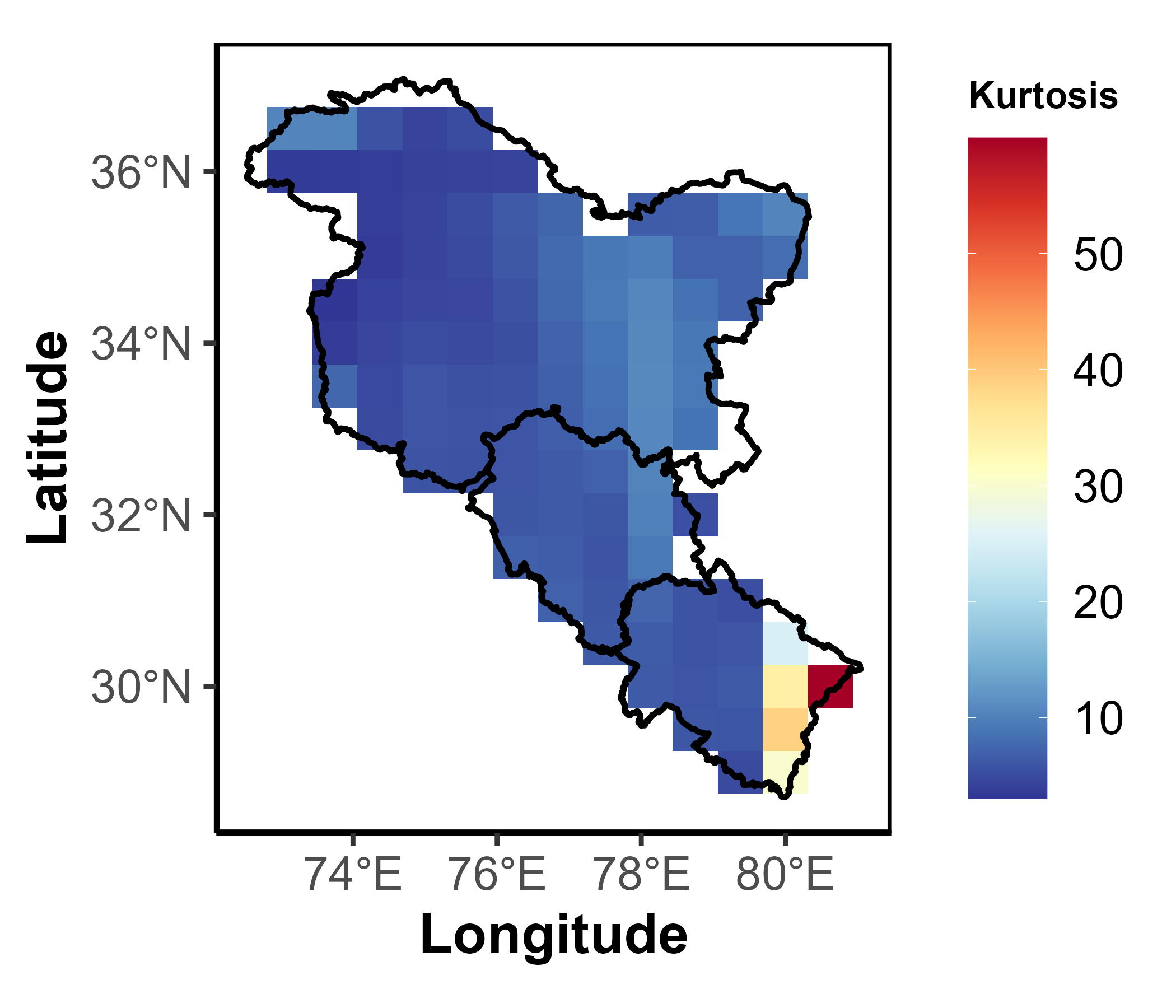}
		\caption*{(c)}
	\end{minipage}
	\medskip
	\begin{minipage}[t]{0.33\textwidth}
		\centering
		\includegraphics[width=1\textwidth]{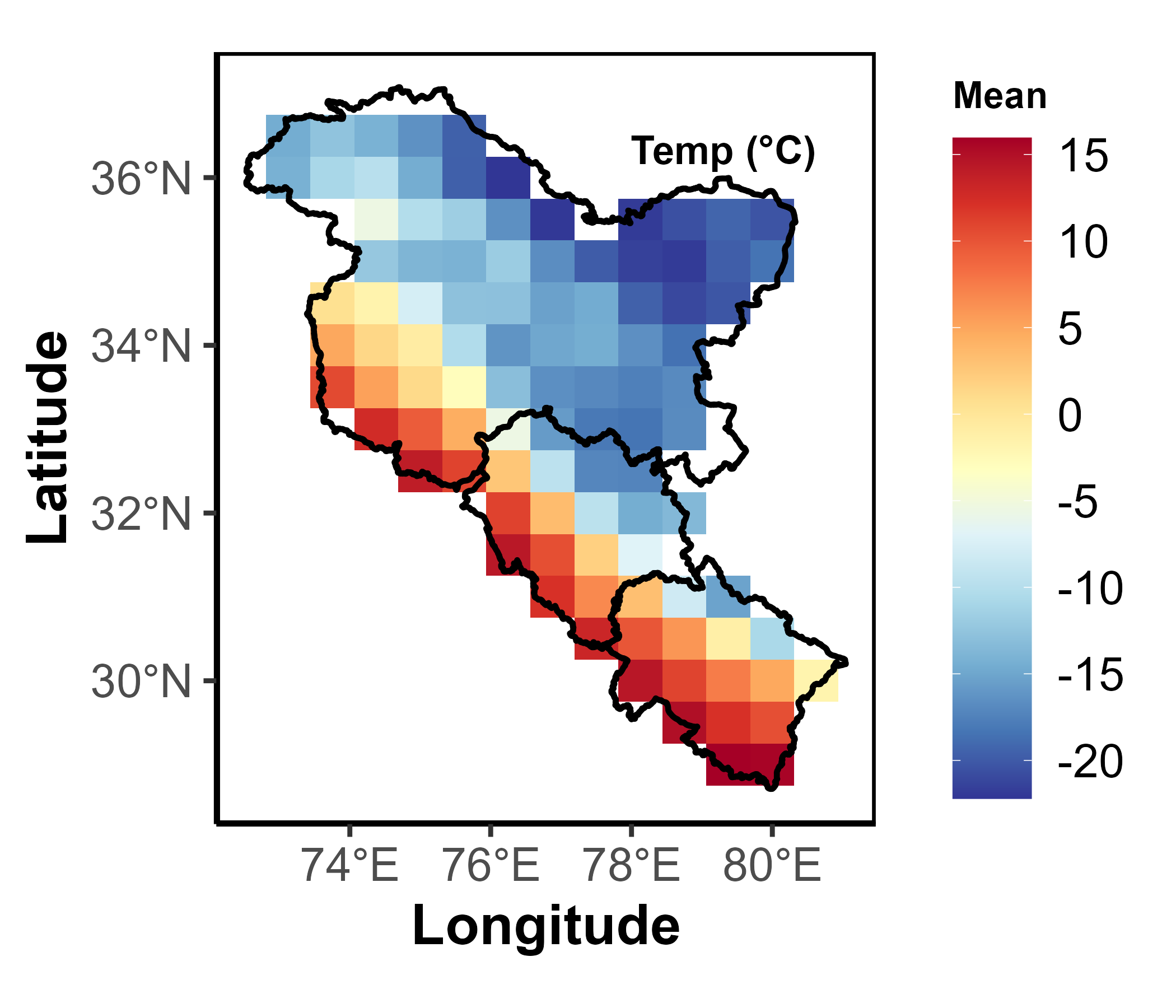}
		\caption*{(d)}
	\end{minipage}\hfill
	\begin{minipage}[t]{0.33\textwidth}
		\centering
		\includegraphics[width=1\textwidth]{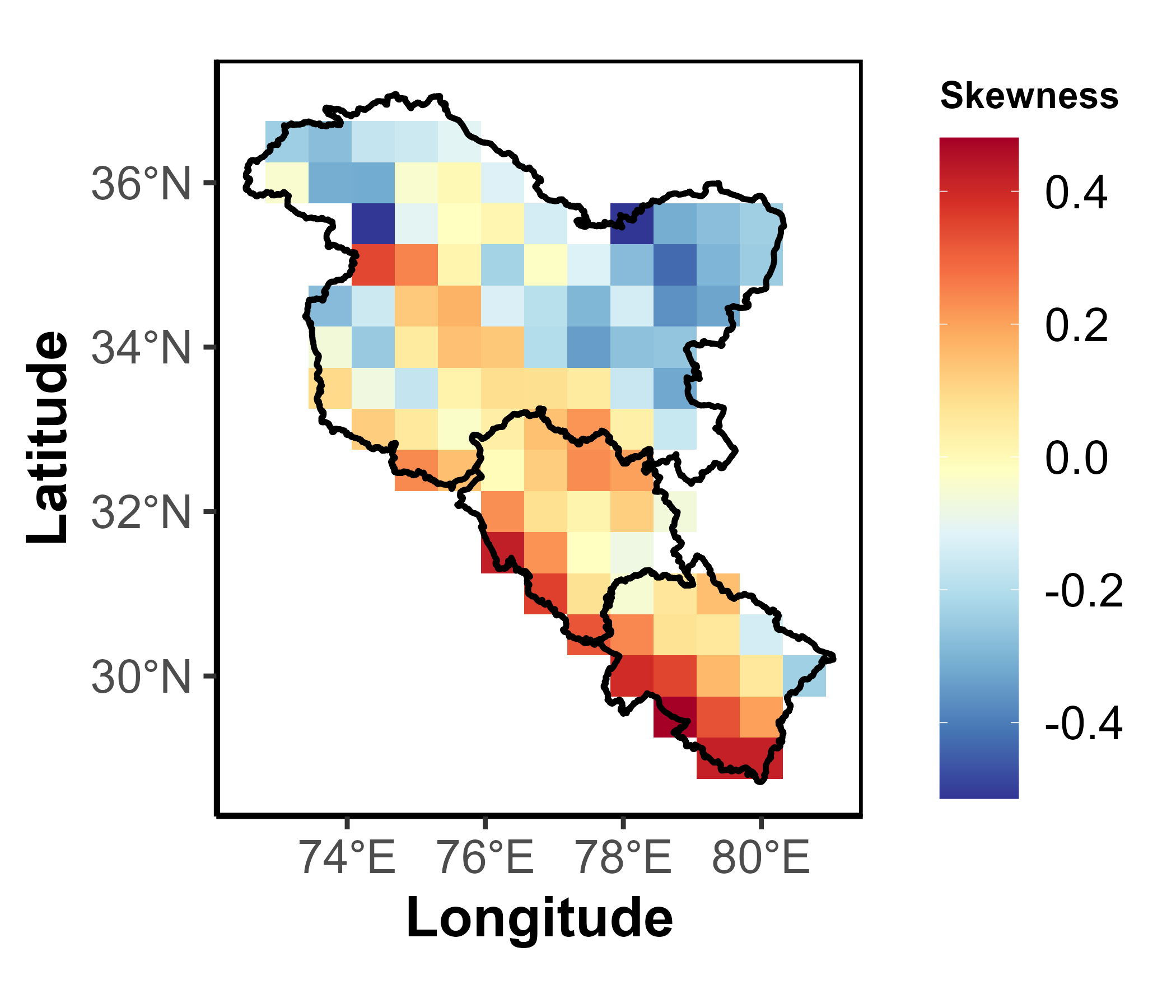}
		\caption*{(e)}
	\end{minipage}\hfill
\begin{minipage}[t]{0.33\textwidth}
	\centering
	\includegraphics[width=1\textwidth]{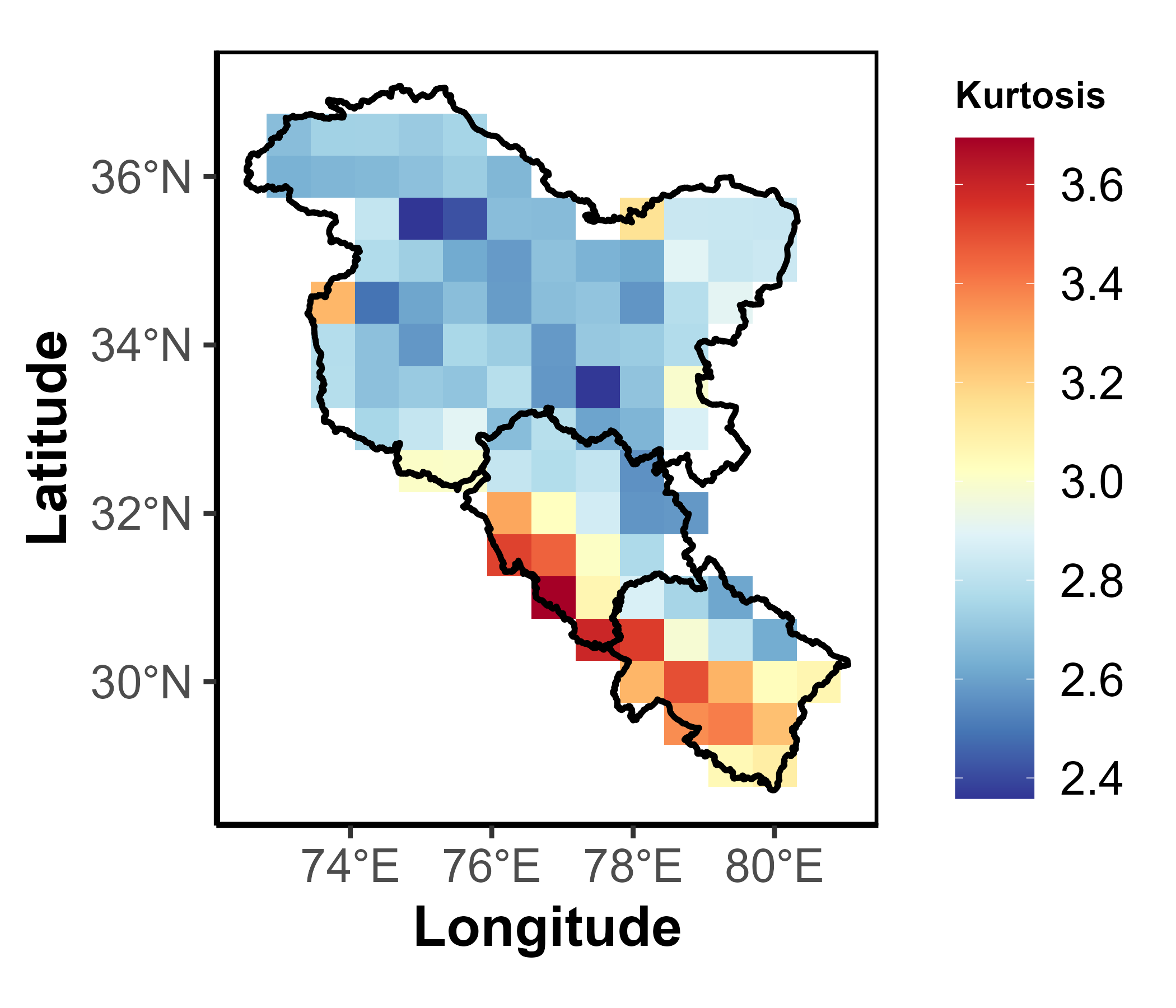}
	\caption*{(f)}
\end{minipage}
	\caption{Descriptive statistics of winter rainfall and temperature across the NWH Himalaya.}
	\label{Figure 4}
\end{figure}
During summer (JJAS), both rainfall and temperature exhibit higher mean values in the lower-Himalayan (LH) regions compared to the Upper Himalayan (UH) areas, reflecting stronger monsoonal influence and relatively warmer conditions in the foothills and mid-altitudes. In contrast, the winter (DJF) rainfall pattern shifts, with higher precipitation observed mainly over the Upper Himalayan (UH) region, particularly in JK, likely due to western disturbances. However, the spatial pattern of temperature remains elevation-dependent in both seasons, with lower temperatures in higher altitudes and comparatively warmer conditions in lower regions. The higher-order moments further highlight the distributional characteristics of these variables. Rainfall shows predominantly positive skewness across both seasons, indicating right-skewed distributions influenced by occasional heavy rainfall events, along with generally leptokurtic behavior suggesting heavier tails. Temperature skewness exhibits a contrasting spatial structure: negative skewness in upper regions implies the occurrence of extreme low-temperature events, while positive skewness in lower regions reflects relatively higher temperature extremes. Kurtosis values for summer temperature are mostly below three, indicating a platykurtic distribution, whereas winter temperatures show kurtosis exceeding three in several LH grids, suggesting greater extremal variability during the colder season.
\begin{figure}[ht]
	\begin{minipage}[t]{0.33\textwidth}
		\centering
		\includegraphics[width=1\textwidth]{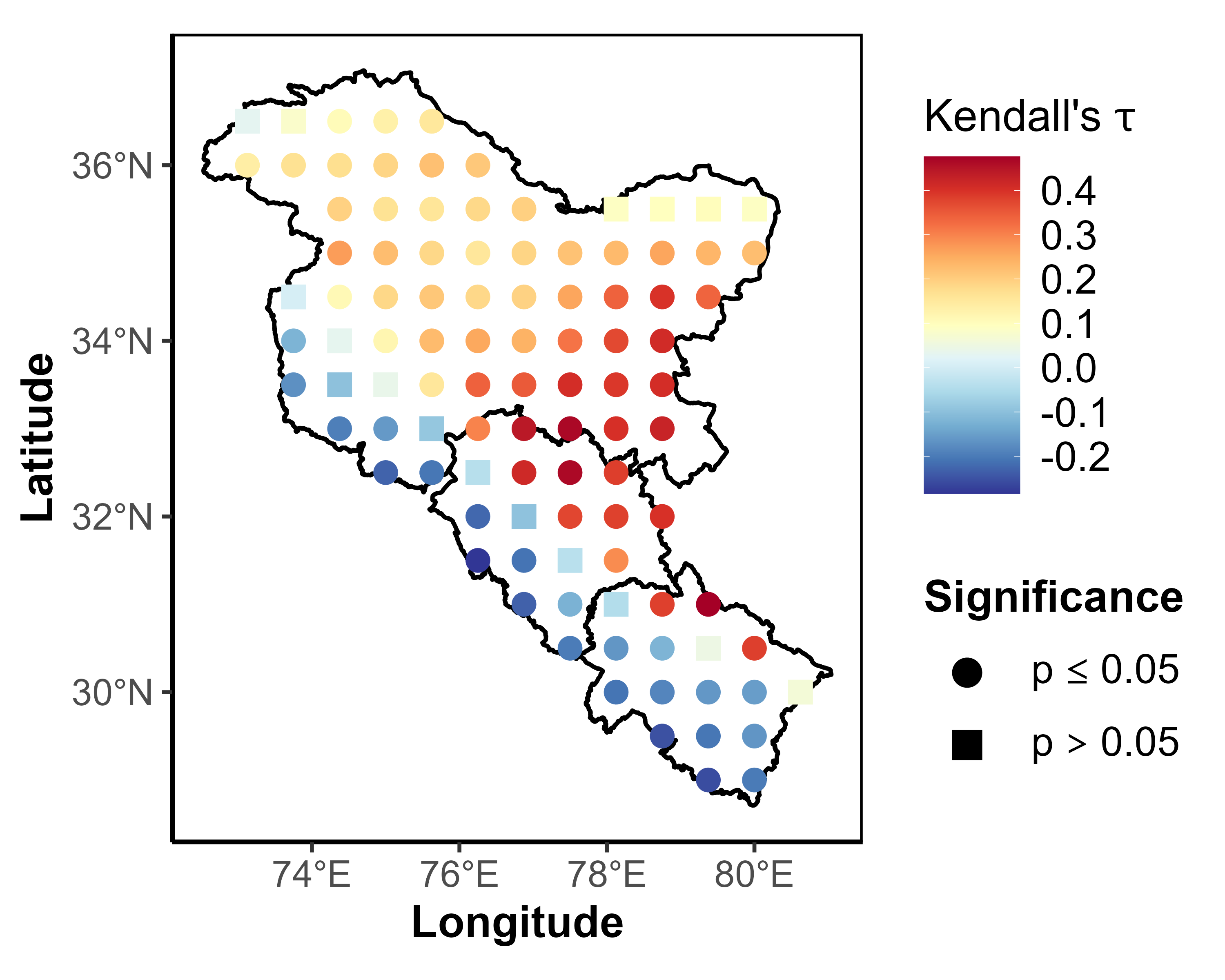}
		\caption*{(a)} 
	\end{minipage}\hfill
	\begin{minipage}[t]{0.33\textwidth}
		\centering
		\includegraphics[width=1\textwidth]{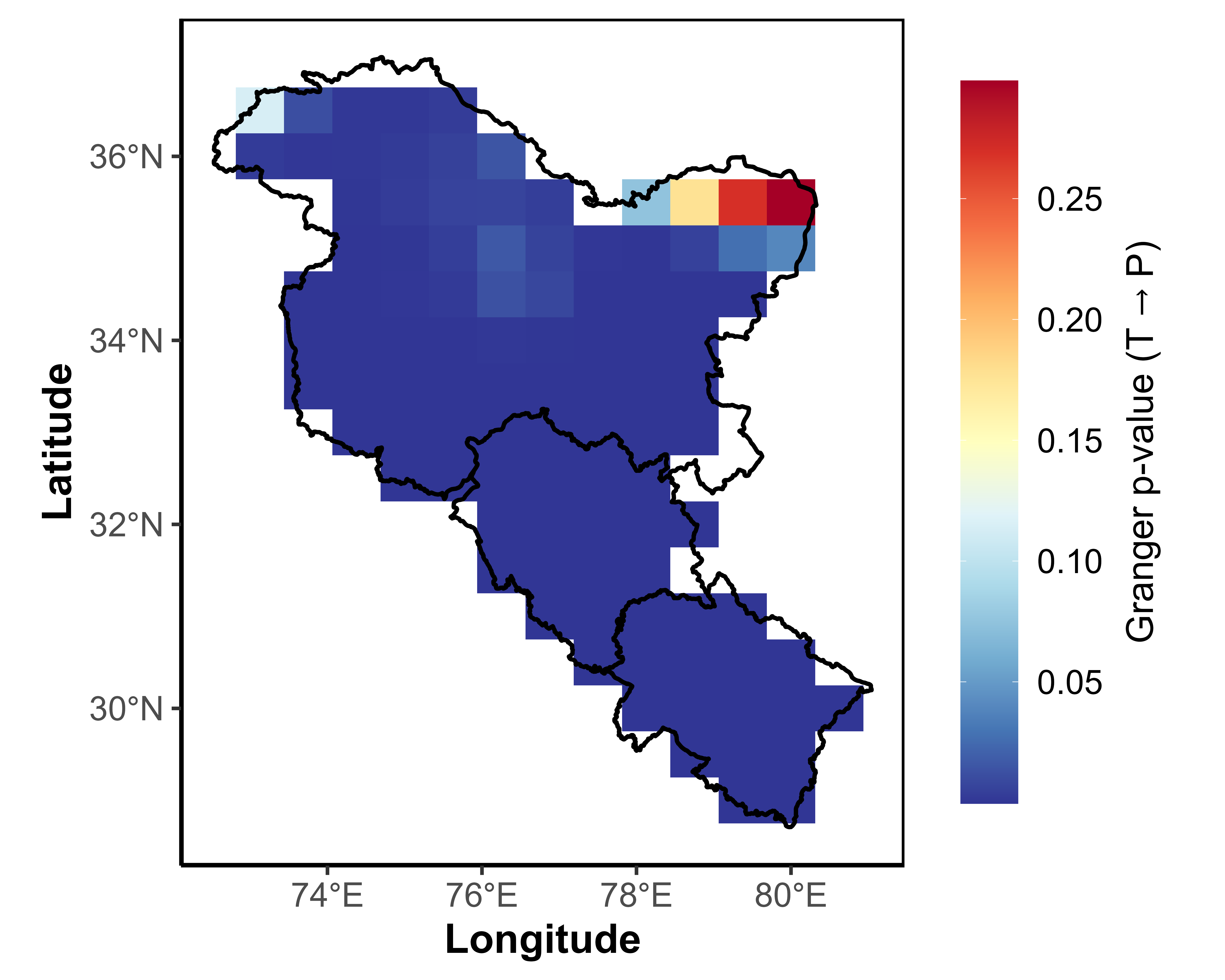}
		\caption*{(b)}
	\end{minipage}\hfill
	\begin{minipage}[t]{0.33\textwidth}
		\centering
		\includegraphics[width=1\textwidth]{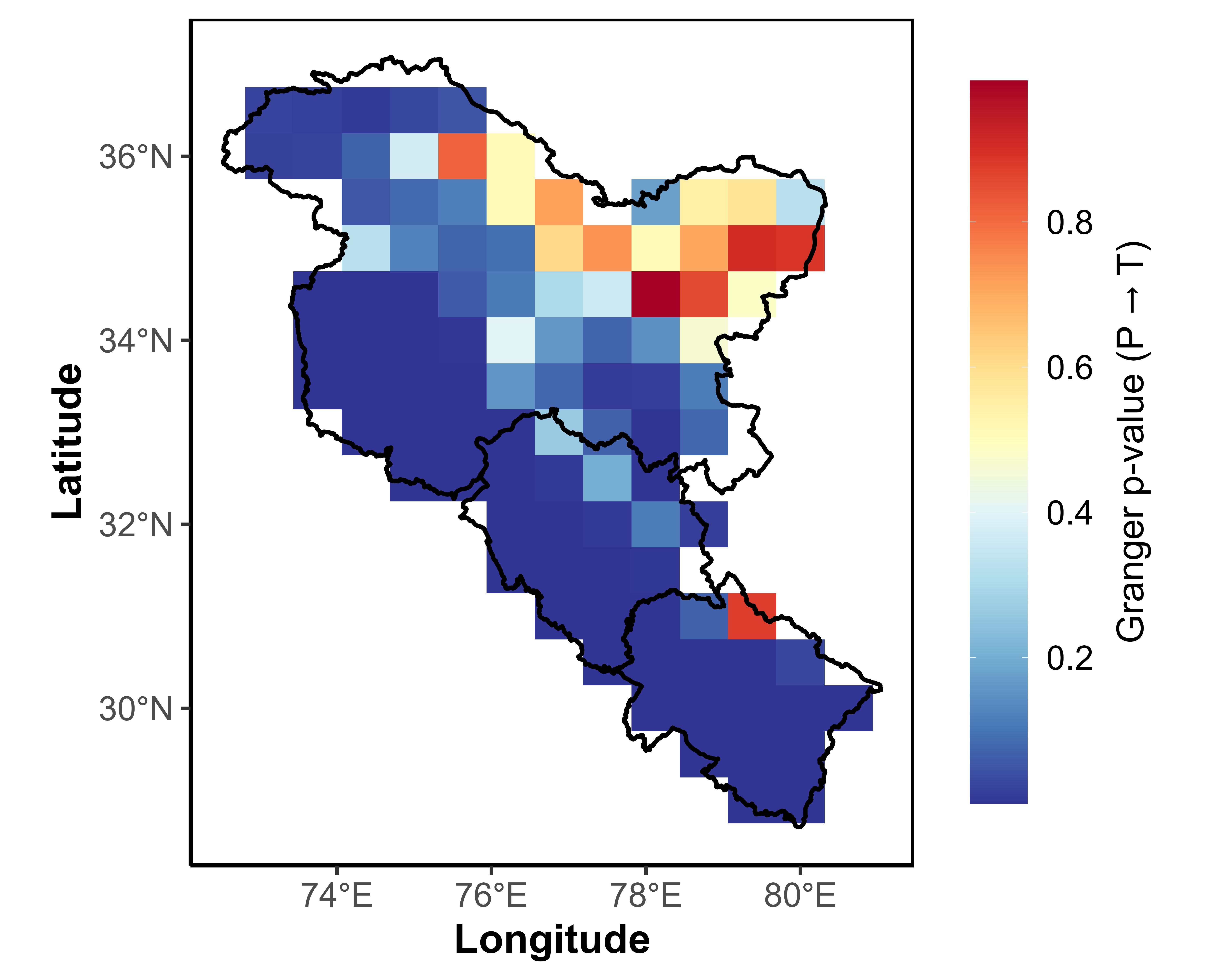}
		\caption*{(c)} 
	\end{minipage}
	\caption{Spatial distribution of (a) Kendall’s correlation coefficient between summer temperature and rainfall along with the corresponding p-values, (b) Granger causality test p-values for temperature influencing rainfall ($T \rightarrow P$), and (c) Granger causality test p-values for rainfall influencing temperature ($P \rightarrow T$) across the NWH region.}
	\label{Figure 5}
\end{figure}

\begin{figure}[ht]
	\begin{minipage}[t]{0.33\textwidth}
		\centering
		\includegraphics[width=1\textwidth]{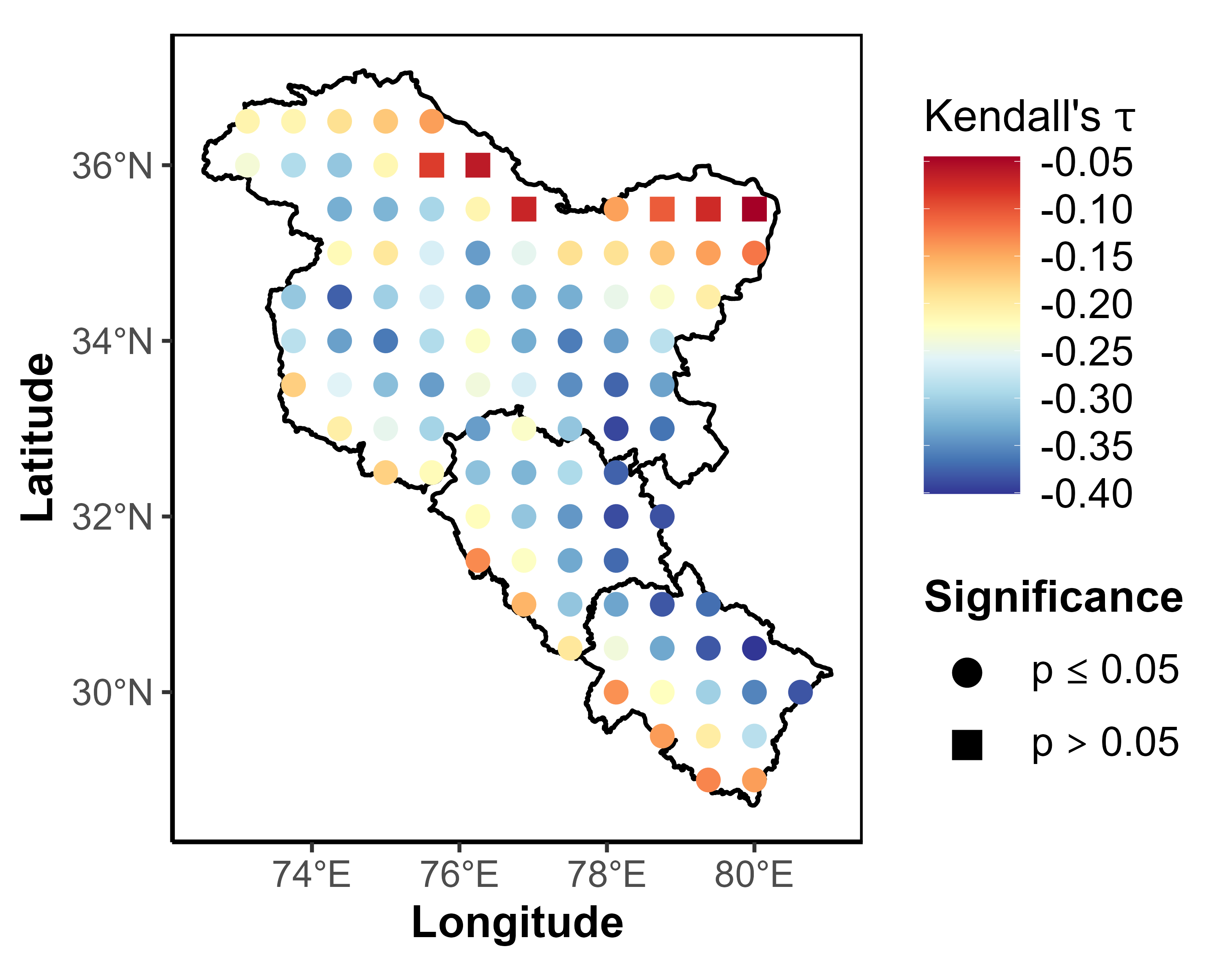}
		\caption*{(a)} 
	\end{minipage}\hfill
	\begin{minipage}[t]{0.33\textwidth}
		\centering
		\includegraphics[width=1\textwidth]{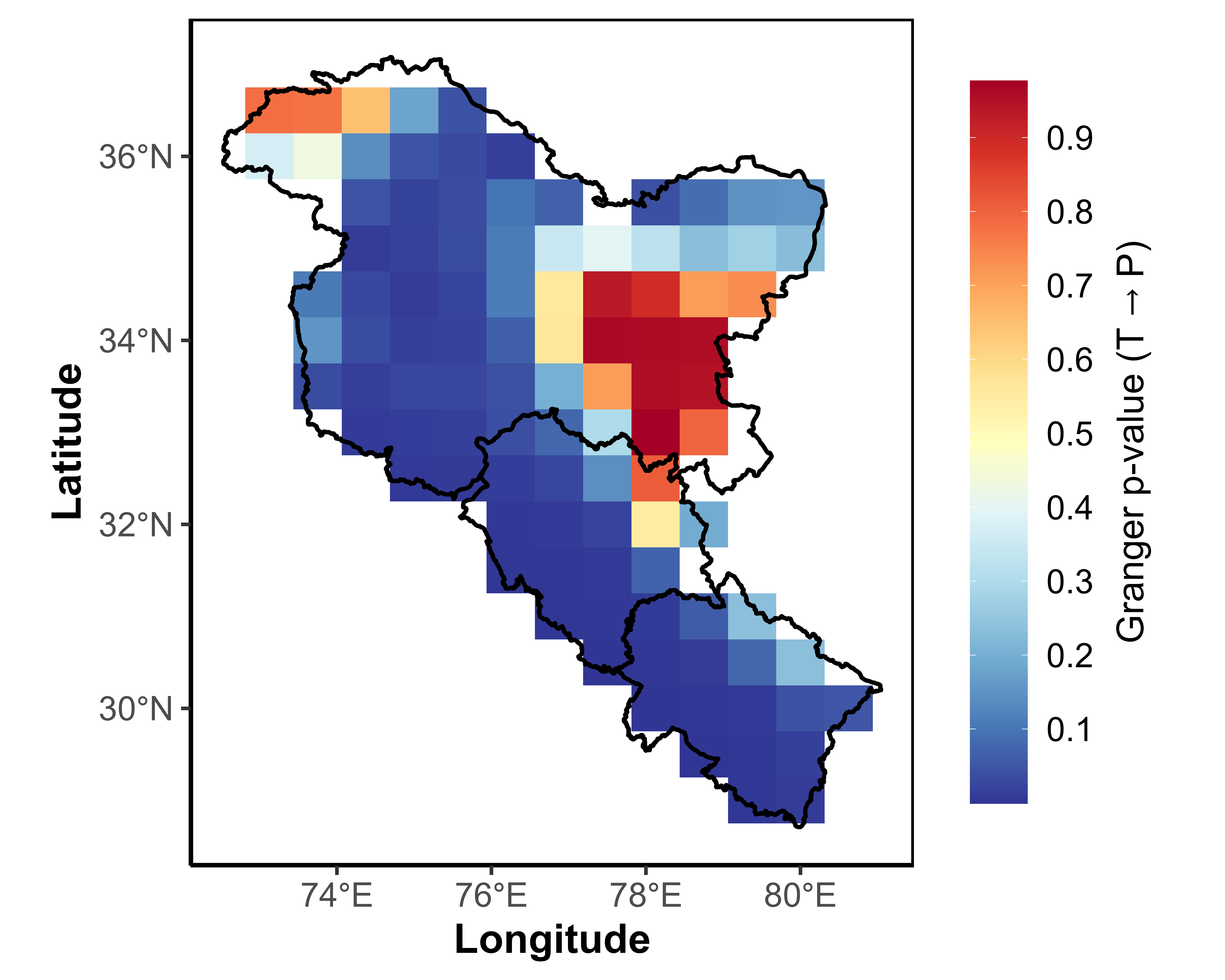}
		\caption*{(b)}
	\end{minipage}\hfill
	\begin{minipage}[t]{0.33\textwidth}
		\centering
		\includegraphics[width=1\textwidth]{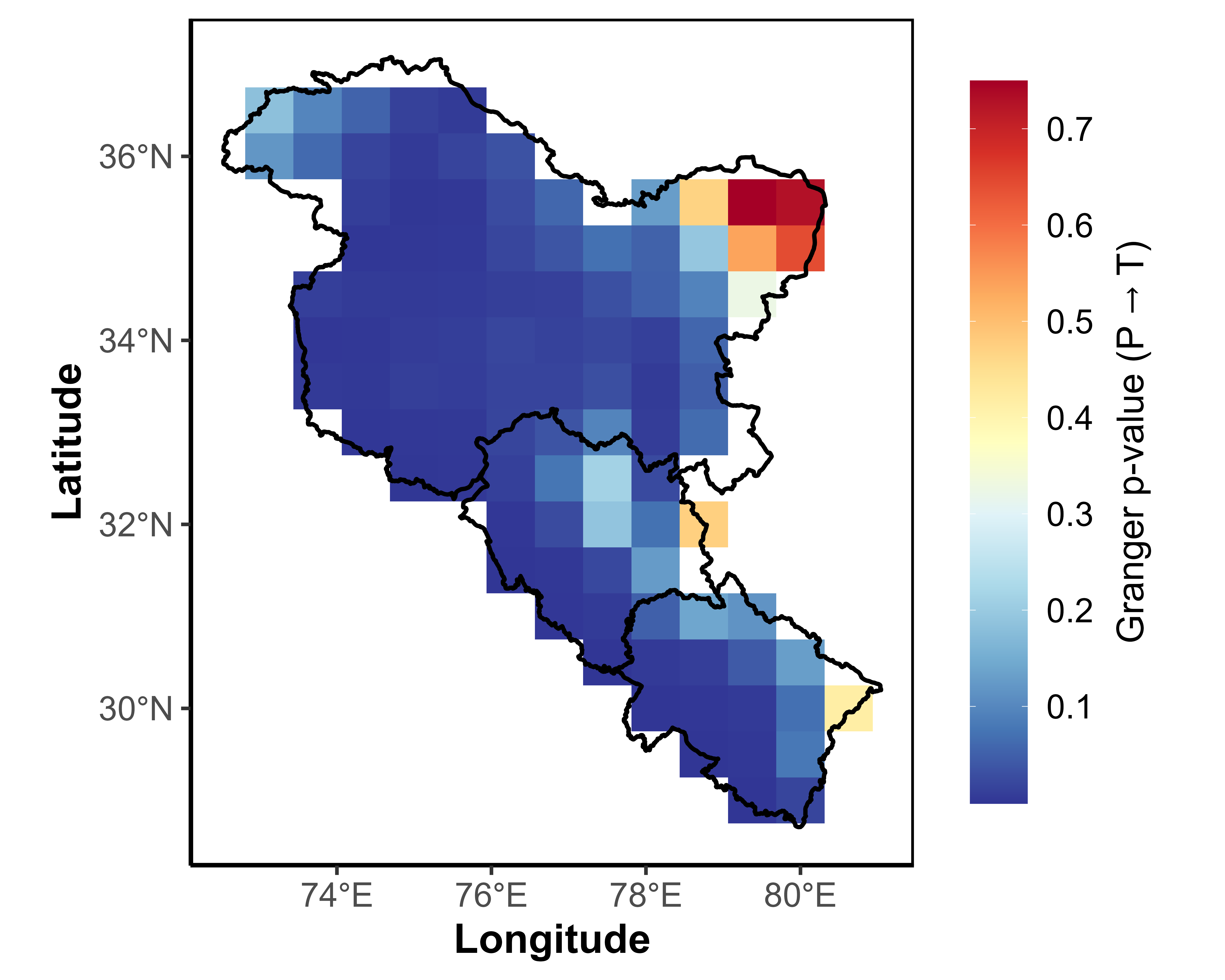}
		\caption*{(c)} 
	\end{minipage}
	\caption{Spatial distribution of (a) Kendall’s correlation coefficient between winter temperature and rainfall along with the corresponding p-values, (b) Granger causality test p-values for temperature influencing rainfall ($T \rightarrow P$), and (c) Granger causality test p-values for rainfall influencing temperature ($P \rightarrow T$) across the NWH region.}
	\label{Figure 9}
\end{figure}
Figures~\ref{Figure 5}(a) and~\ref{Figure 9}(a) present the spatial distribution of Kendall’s rank correlation between temperature and rainfall for the summer (JJAS) and winter (DJF) seasons respectively, along with their associated p-values. During summer, the LH region exhibits predominantly negative correlation, whereas the Middle Himalayan (MH) and UH regions show positive association between the two variables. Approximately 83\% of the grids demonstrate statistically significant dependence at the 5\% level. In contrast, the winter season is characterized by a consistently negative correlation across the entire NWH region, with nearly 94\% of the grids showing statistically significant association. To further investigate the directional relationship between the variables, the Granger causality test is employed. This test examines whether past values of one time series provide statistically significant predictive information about another series beyond its own history. Figures~\ref{Figure 5}(b)–(c) and~\ref{Figure 9}(b)–(c) display the spatial distribution of Granger causality p-values for the directions $T \rightarrow P$ and $P \rightarrow T$ during summer and winter, respectively. In summer, temperature significantly Granger-causes rainfall in about 95\% of the grids, with limited exceptions in parts of eastern JK, whereas rainfall significantly influences temperature in approximately 55\% of the grids, primarily over the western part of NWH region. The winter pattern shows a reversal in directional dominance: rainfall exerts a stronger predictive influence on temperature, with about 68\% of grids showing significance, while temperature significantly predicts rainfall in nearly 50\% of the region, mainly over the western sector. Overall, the results highlight pronounced seasonal asymmetry in both dependence structure and directional predictability between temperature and rainfall across the NWH Himalaya.

To model joint behaviour of temperature and rainfall, copula-based bivariate models are employed. Prior to model fitting, appropriate data transformation and scaling procedures are implemented to summer and winter datasets to ensure compatibility with model assumptions and to facilitate numerical optimization. The pre-processing steps are summarized as follows:
\begin{itemize}
	\item Since the proposed models are defined for strictly positive values, the rainfall data are shifted by adding a small constant (0.01) to all observations to avoid zero values.
	\item Similarly, temperature values, which are negative in several months across many grids, are shifted by adding $|\min(T)| + 0.01$ to each grid-specific series, ensuring that all transformed temperature observations are strictly positive.
	\item For numerical stability and improved optimization performance, both temperature and rainfall series are scaled by normalizing each observation with respect to the corresponding grid-specific maximum value.
\end{itemize}

\subsection{Joint Modeling of Temperature and Rainfall}
Following pre-processing, the choice of copula structure is guided by the sign of dependence. For grids exhibiting positive correlation between temperature and rainfall, the CBKT and GBKT copula models are fitted. In contrast, for grids displaying negative dependence, rotated versions of the CBKT and GBKT copulas (specifically the $90^\circ$ rotations) are employed, as defined in equation~\eqref{rot}, to adequately capture inverse dependence structures.  
For comparative assessment, several existing bivariate distributions are also considered, including the Clayton Bivariate Exponentiated Teissier (CBET) distribution~\citep{poonia2023new}, which is a special case of the CBKT model, the Clayton Bivariate Rayleigh (CBR) distribution~\citep{el2019bivariate}, the Bivariate Generalized Exponential (BGE) distribution~\citep{mirhosseini2015new}, and the Bivariate Kumaraswamy Exponential (BKE) distribution~\citep{bakouch2019bivariate}. 

The adequacy of selected bivariate copula models for joint temperature--rainfall analysis is evaluated using and Cram\'er--von Mises (CVM) goodness-of-fit tests. Let $C_{\text{emp}}$ denote the empirical copula
and $C_{\Theta}$ denote the fitted copula with parameter $\Theta$ then, CVM statistics $S_n$ is obtained via 
\[
S_n = \sum_{k=1}^{n} \left[ C_{\text{emp}}(u_k) - C_{\Theta}(u_k) \right]^2,
\]
which quantifies the integrated squared difference between the empirical and fitted copulas. The null hypothesis states that the empirical copula belongs to the specified copula family. The significance of the statistics is assessed using parametric bootstrap procedures at the 5\% significance level. Models with $p$-values greater than 0.05 are considered to adequately represent the dependence structure. 

Furthermore, model selection is performed using the maximum log-likelihood value $(\mathcal{L})$, the Akaike Information Criterion (AIC), and the Bayesian Information Criterion (BIC). The AIC and BIC are computed as \(\text{AIC} = 2k - 2l(\hat{\Theta}), ~\text{BIC} = k \log(n) - 2l(\hat{\Theta}),\) where $l(\hat{\Theta})$ denotes the maximized log-likelihood obtained via MLE, $k$ represents the number of estimated parameters, and $n$ is the sample size. In general, models with larger log-likelihood values, smaller AIC and BIC values are preferred, as they indicate better goodness-of-fit while penalizing model complexity. 

\begin{table}[ht]
	\caption{The ML estimates, $\mathcal{L}$, AIC, and BIC of the grid with coordinates (79.375, 31)}
	\label{table3}
	\centering
	\resizebox{1.1\textwidth}{!}{%
		\begin{tabular}{lcccc}
			\hline
			Model & $\mathcal{L}$ & AIC & BIC & Parameter Estimates \\ 
			\hline
			CBKT & 208.3576 & -402.7152 &  -381.0161  & $(\hat{\alpha}_1, \hat{\alpha}_2, \hat{\beta}_1, \hat{\beta}_2, \hat{\theta}_1, \hat{\theta}_2, \hat{\delta}_1) = (3.5810, 2553.61, 0.5419, 0.8147, 20.8419, 0.6021, 1.2349)$ \\[3pt]
			GBKT & 208.3866 & -402.7732 & -381.0741 & $(\hat{\alpha}_1, \hat{\alpha}_2, \hat{\beta}_1, \hat{\beta}_2, \hat{\theta}_1, \hat{\theta}_2, \hat{\delta}_2) = (3.5810, 2553.61, 0.5419, 0.8147, 20.8419, 0.6021, 1.75698)$ \\[3pt]
			CBET & 198.1452 & -386.2903 & -370.791  & $(\hat{\alpha}_1, \hat{\alpha}_2, \hat{\theta}_1, \hat{\theta}_2, \hat{\delta}) = (6.2444, 0.6141, 2.0264, 2.5539, 1.3274)$ \\[3pt]
			BKE &  7.92049 & -7.8409 &  4.55849  & $(\hat{\alpha}_1, \hat{\alpha}_2, \hat{\alpha}_3, \hat{\beta}) = (8.9685, 30.2394,  0.3259,  3.2387)$ \\[3pt]
			BGE & -93.41532 & 192.8306 & 202.1302   & $(\hat{\lambda}_1, \hat{\lambda}_2, \hat{\alpha}) = (1.2865, 3.2493, 0.9990)$ \\[3pt]
			CBR & 40.86746 & -75.7349 & -66.4353  & $(\hat{\lambda}_1, \hat{\lambda}_2, \hat{\theta}) = (1.2397, 2.8734, 0.2519)$ \\[3pt]
			\hline
		\end{tabular}
	}
\end{table}
The results presented in Table~\ref{table3} summarize the maximum likelihood estimates and model selection criteria for the grid located at $(79.375, 31)$ for summer temperature and rainfall. Among the competing models, the GBKT distribution attains the highest log-likelihood value along with the smallest AIC and BIC values, indicating that it provides the best fit for the joint temperature–rainfall data at this location according to the adopted selection criteria. The CBKT model yields comparable parameter estimates; however, its log-likelihood is smaller and its AIC and BIC values are bit larger relative to CBKT. The remaining competing models (CBET, BKE, BGE, and CBR) produce substantially lower log-likelihood values and larger information criteria which provide a comparatively weaker fit. Based on these measures, the GBKT model is selected as the most appropriate specification for this grid.

\begin{figure}[ht]
	\begin{minipage}[t]{0.49\textwidth}
		\centering
		\includegraphics[width=1\textwidth]{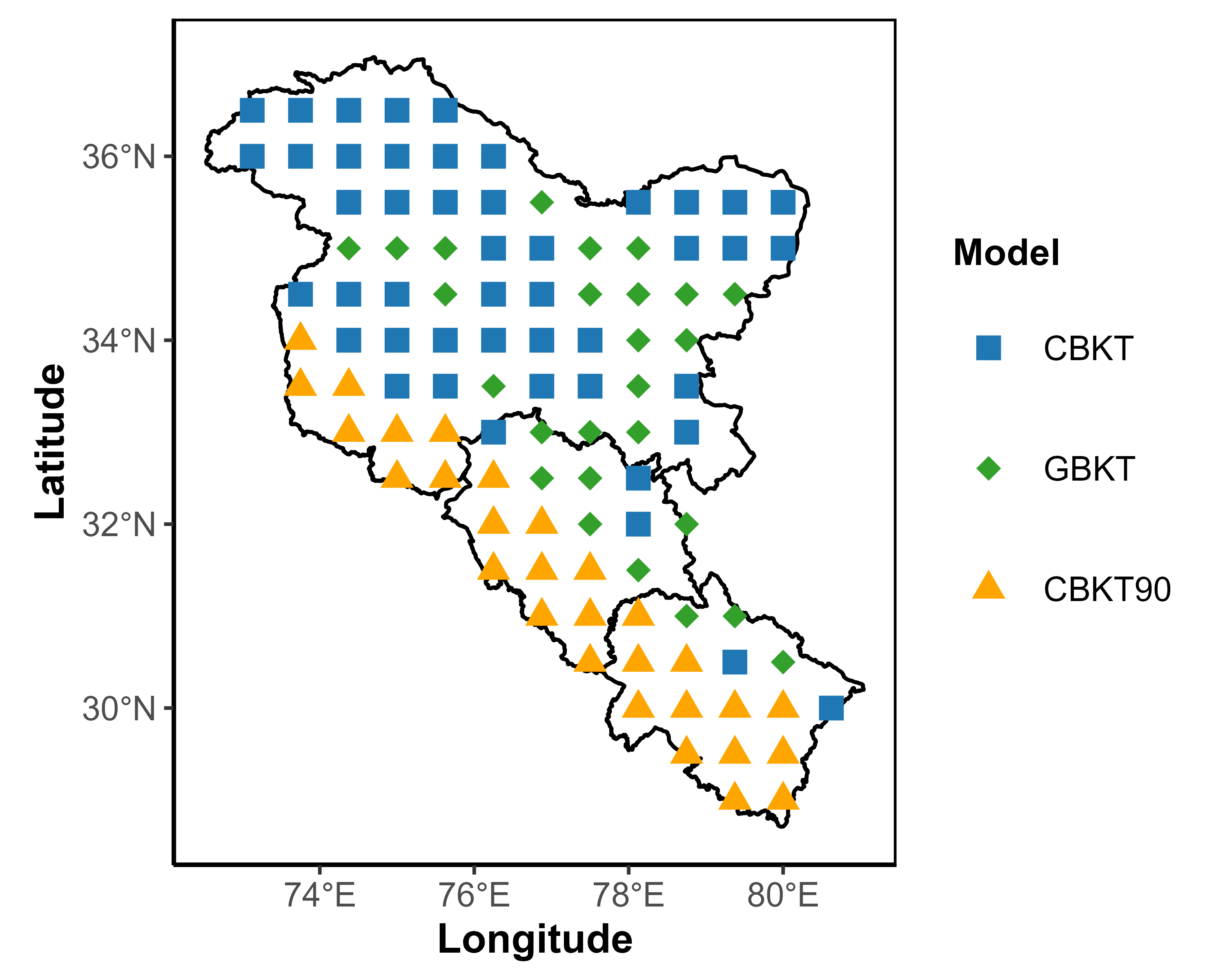}
		\caption*{(a)} 
	\end{minipage}\hfill
	\begin{minipage}[t]{0.49\textwidth}
		\centering
		\includegraphics[width=1\textwidth]{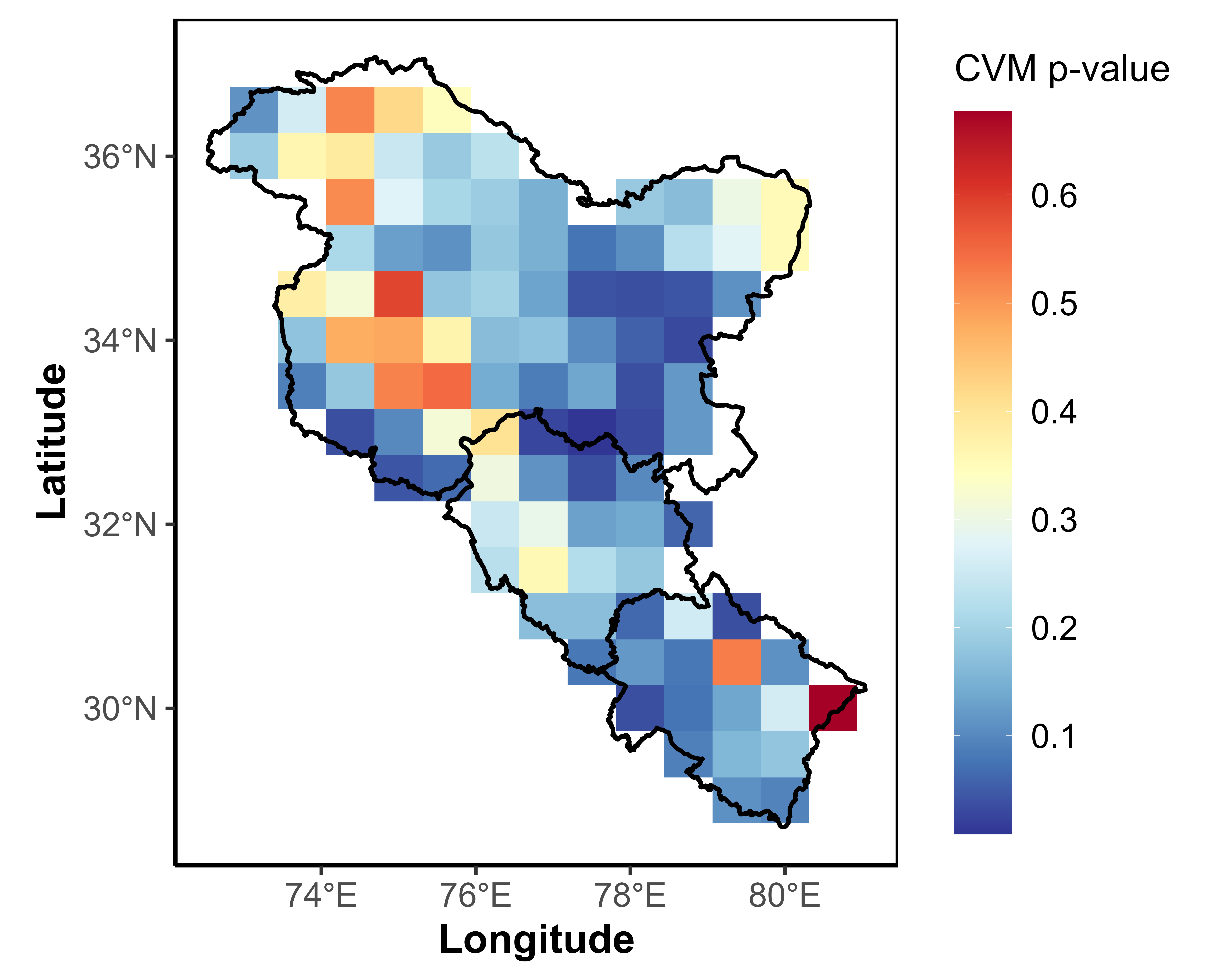}
		\caption*{(b)}
	\end{minipage}
	\caption{Spatial representation of (a) the selected best-fit bivariate model for joint summer temperature–rainfall data, (b) the Cram\'er–von Mises (CVM) p-value used to evaluate model adequacy across the NWH region.}
	\label{Figure 10}
\end{figure}

\begin{figure}[ht]
	\begin{minipage}[t]{0.49\textwidth}
		\centering
		\includegraphics[width=1\textwidth]{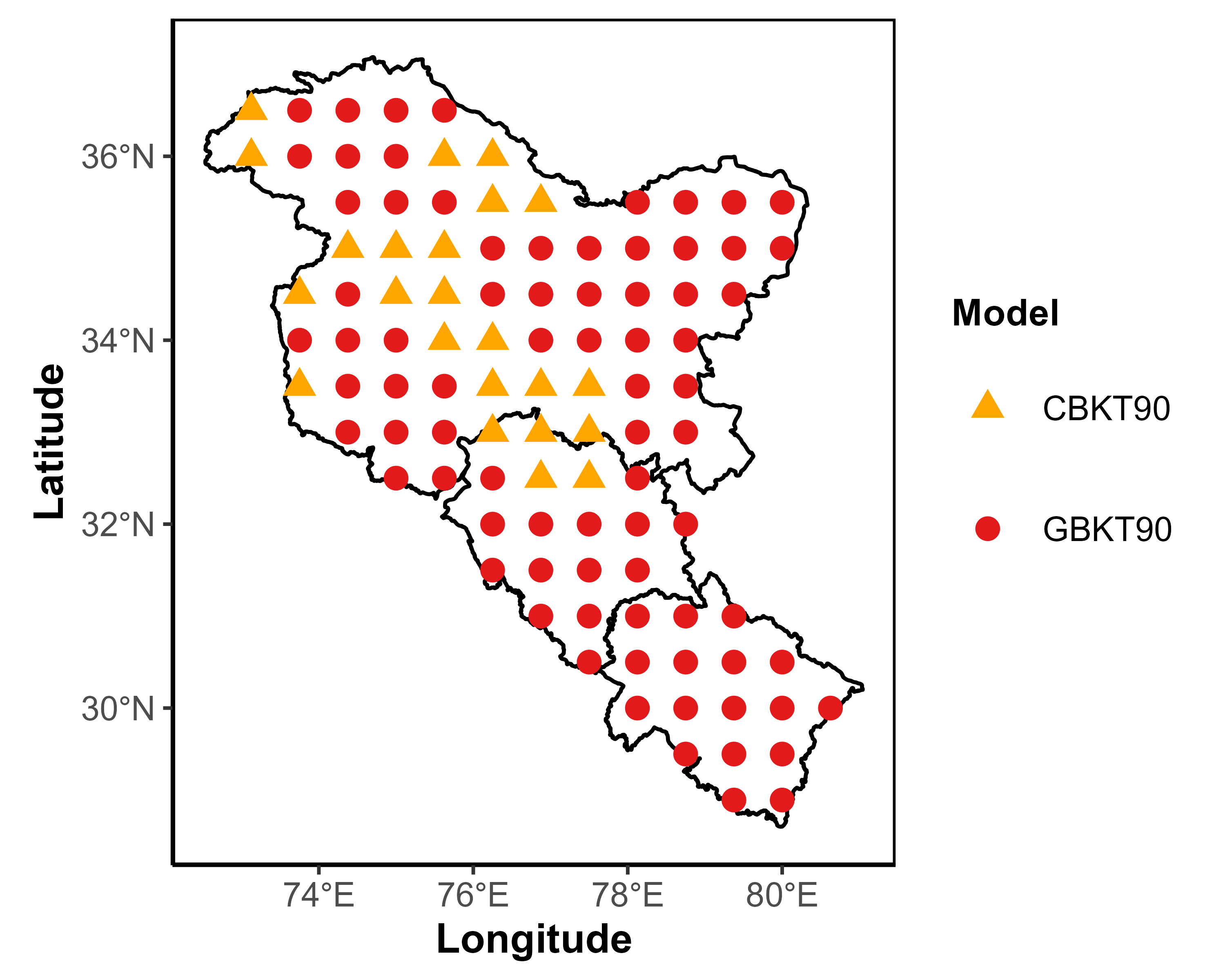}
		\caption*{(a)} 
	\end{minipage}\hfill
	\begin{minipage}[t]{0.49\textwidth}
		\centering
		\includegraphics[width=1\textwidth]{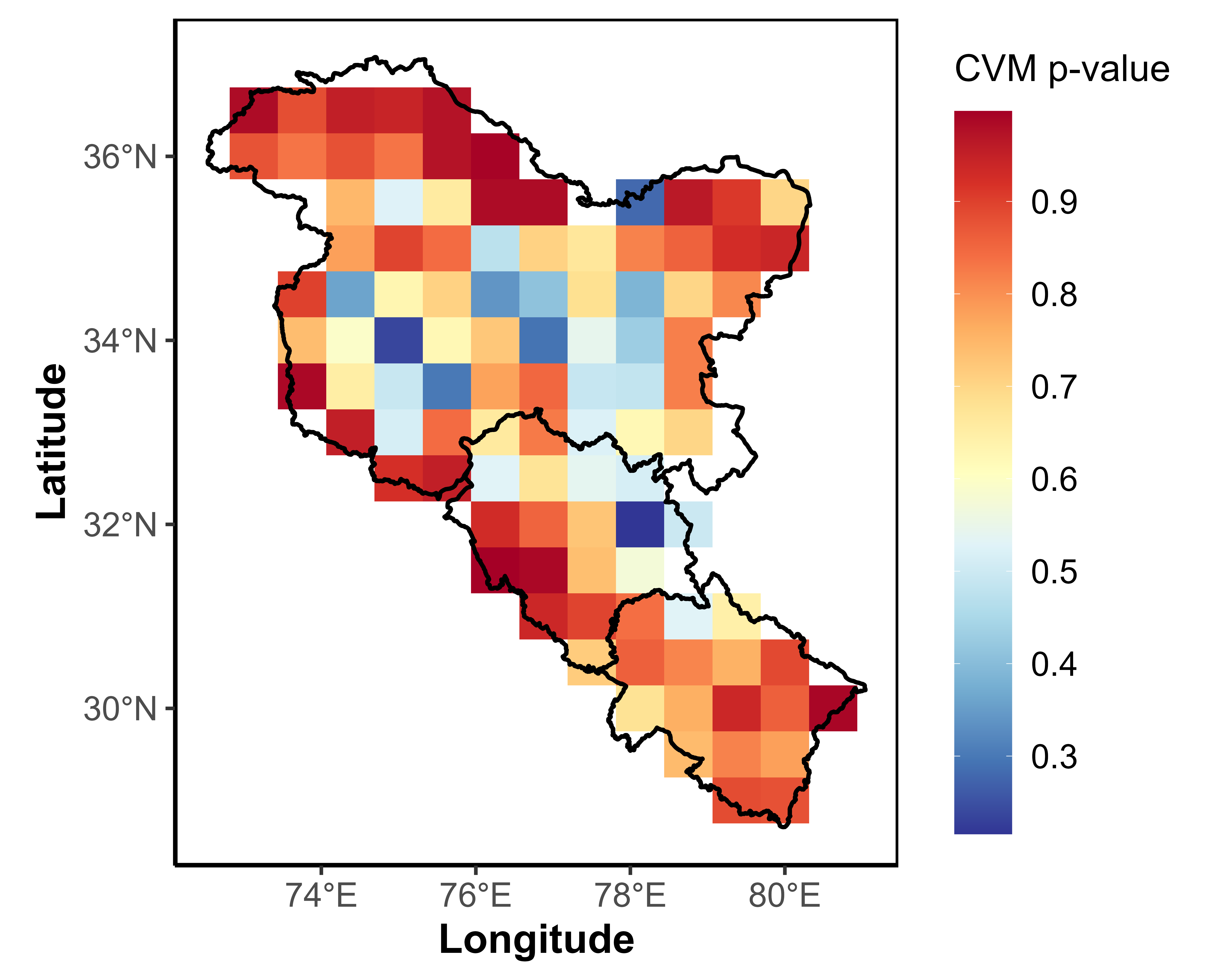}
		\caption*{(b)}
	\end{minipage}
	\caption{Spatial representation of (a) the selected best-fit bivariate model for joint winter temperature–rainfall data, (b) the Cram\'er–von Mises (CVM) p-value used to evaluate model adequacy across the NWH region.}
	\label{Figure 12}
\end{figure}

Table~\ref{table3} presents the model comparison results for a representative grid however, the same estimation and model selection framework is systematically applied to all 101 grids across the NWH region for both summer and winter seasons. For each grid, the optimal copula model is selected based on the maximized log-likelihood, AIC and BIC. The resulting spatial distribution of the best-fitting copula models is displayed in Figures~\ref{Figure 10}(a) and~\ref{Figure 11}(a) for summer and winter, respectively. During the summer season, grids along the LH region are predominantly best described by the rotated CBKT90 copula, reflecting negative dependence between temperature and rainfall in these areas. In contrast, grids exhibiting positive dependence are primarily fitted by either the CBKT or GBKT copula. Model adequacy is further assessed using the CVM goodness-of-fit test. The spatial distribution of the corresponding p-values is shown in Figure~\ref{Figure 10}(b). Across all grid locations, the p-values exceed the 5\% significance level, indicating that the null hypothesis—that the empirical copula belongs to the selected copula family—cannot be rejected at any location. For the winter season, the dependence structure shows a different spatial configuration. The central region of JK and the upper parts of HP are best fitted by the rotated CBKT90 copula, while the remaining areas are predominantly characterized by the rotated GBKT90 copula. The CVM goodness-of-fit p-values, presented in Figure~\ref{Figure 12}(b), similarly exceed the 5\% significance threshold for all grids. Thus, the null hypothesis is not rejected anywhere in the study domain, providing strong statistical evidence that the selected copula models adequately capture the dependence structure between temperature and rainfall across the NWH region in both seasons.

\subsection{Estimated Tail Dependence from the Best-Fitted Copula}
Correlation quantifies the linear association and its statistical significance; however, it does not distinguish whether dependence arises from joint lower extremes, joint upper extremes, or cross-extreme interactions. In this context, copula-based tail dependence provides deeper insight into extremal co-movement. The fitted copulas—CBKT, GBKT, CBKT90, and GBKT90—correspond respectively to the tail coefficients $\lambda_{LL}$, $\lambda_{UU}$, $\lambda_{UL}$, and $\lambda_{LU}$ (as described in Section \ref{tailsec}). Figure \ref{Figure 11} illustrates the spatial distribution of the tail dependence coefficient ($\lambda$) across NWH region for summer and winter seasons. The color gradient represents the magnitude of $\lambda$, while marker shapes denote the associated tail type. 

\begin{figure}[ht]
	\begin{minipage}[t]{0.49\textwidth}
		\centering
		\includegraphics[width=1\textwidth]{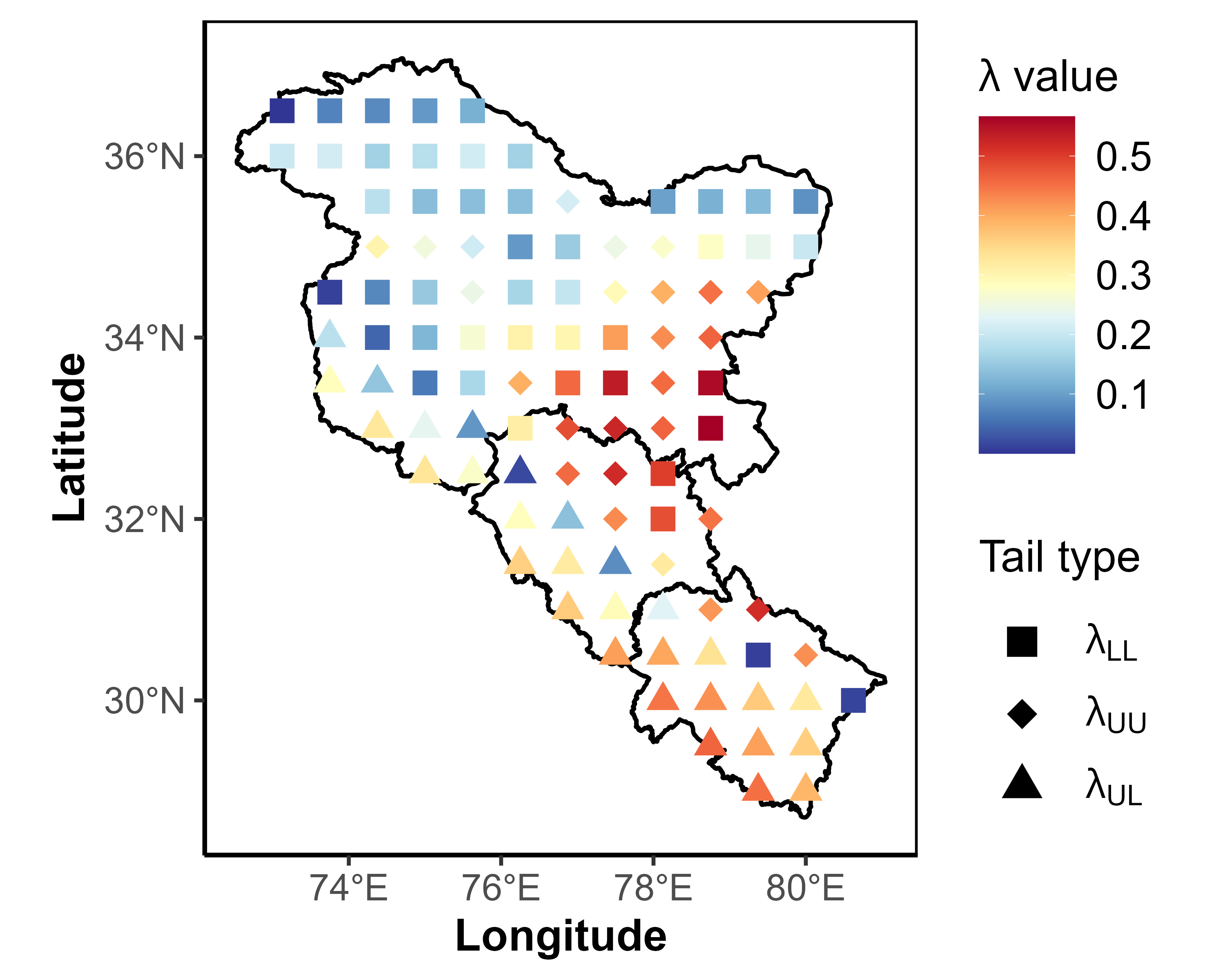}
		\caption*{(a)} 
	\end{minipage}\hfill
	\begin{minipage}[t]{0.49\textwidth}
		\centering
		\includegraphics[width=1\textwidth]{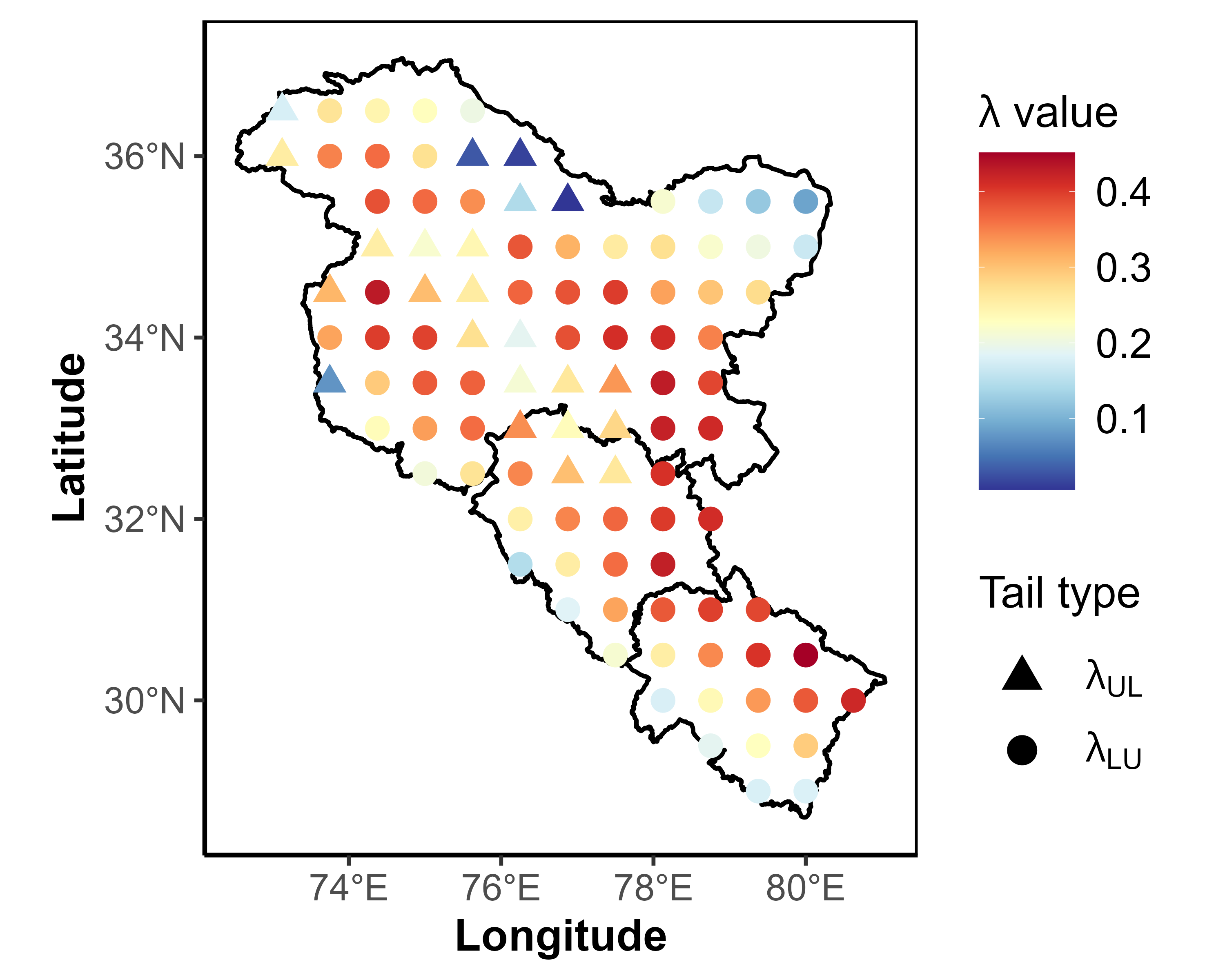}
		\caption*{(b)}
	\end{minipage}
	\caption{Estimated tail dependence coefficients ($\lambda$) obtained from the selected copula models across the NWH region for (a) summer and (b) winter season.}
	\label{Figure 11}
\end{figure}
During summer, the LH region is predominantly characterized by the rotated CBKT90 copula, indicating upper–lower tail dependence ($\lambda_{UL}$). This suggests the co-occurrence of high-temperature extremes with low-rainfall events, with relatively stronger dependence observed along the boundary areas. In contrast, the eastern parts of HP and JK are mainly fitted by the GBKT copula, reflecting upper–upper tail dependence ($\lambda_{UU}$), where high temperature and high rainfall extremes tend to occur jointly. The remaining portions of JK and adjoining areas of HP are best described by the CBKT copula, implying lower–lower tail dependence ($\lambda_{LL}$), corresponding to simultaneous low-temperature and low-rainfall extremes.

In winter, the dependence structure shifts notably. The middle Himalayan regions of JK and HP are primarily fitted by the rotated CBKT90 copula, indicating upper–lower tail dependence, although the magnitude of the tail coefficient is comparatively weak. Conversely, large parts of UK, HP, and the remaining areas of JK are best characterized by the rotated GBKT90 copula, corresponding to lower–upper tail dependence ($\lambda_{LU}$). This reflects the co-occurrence of low-temperature extremes with high-rainfall events, particularly pronounced in eastern UK and JK, where stronger negative dependence is observed, suggesting an increase in winter rainfall across much of the NWH Himalayan region.

\subsection{Return Periods of Temperature and Rainfall}
\subsubsection{Univariate Return Periods}
Following the joint modeling and tail dependence analysis, the extremal behavior of temperature $(X_1)$ and rainfall $(X_2)$ is further evaluated using univariate and bivariate return periods. The marginal return periods corresponding to temperature and rainfall are defined as
\begin{equation}
	T_{X_1} = \frac{1}{P(X_1 \geq x_1)} 
	= \frac{1}{1 - F_{X_1}(x_1)}
\end{equation}
\begin{equation}
	T_{X_2} = \frac{1}{P(X_2 \geq x_2)} 
	= \frac{1}{1 - F_{X_2}(x_2)}
\end{equation}
where $F_{X_1}$ and $F_{X_2}$ denote the marginal distribution functions of temperature and rainfall, respectively. 
\begin{figure}[ht]
	\begin{minipage}[t]{0.45\textwidth}
		\centering
		\includegraphics[width=1\textwidth]{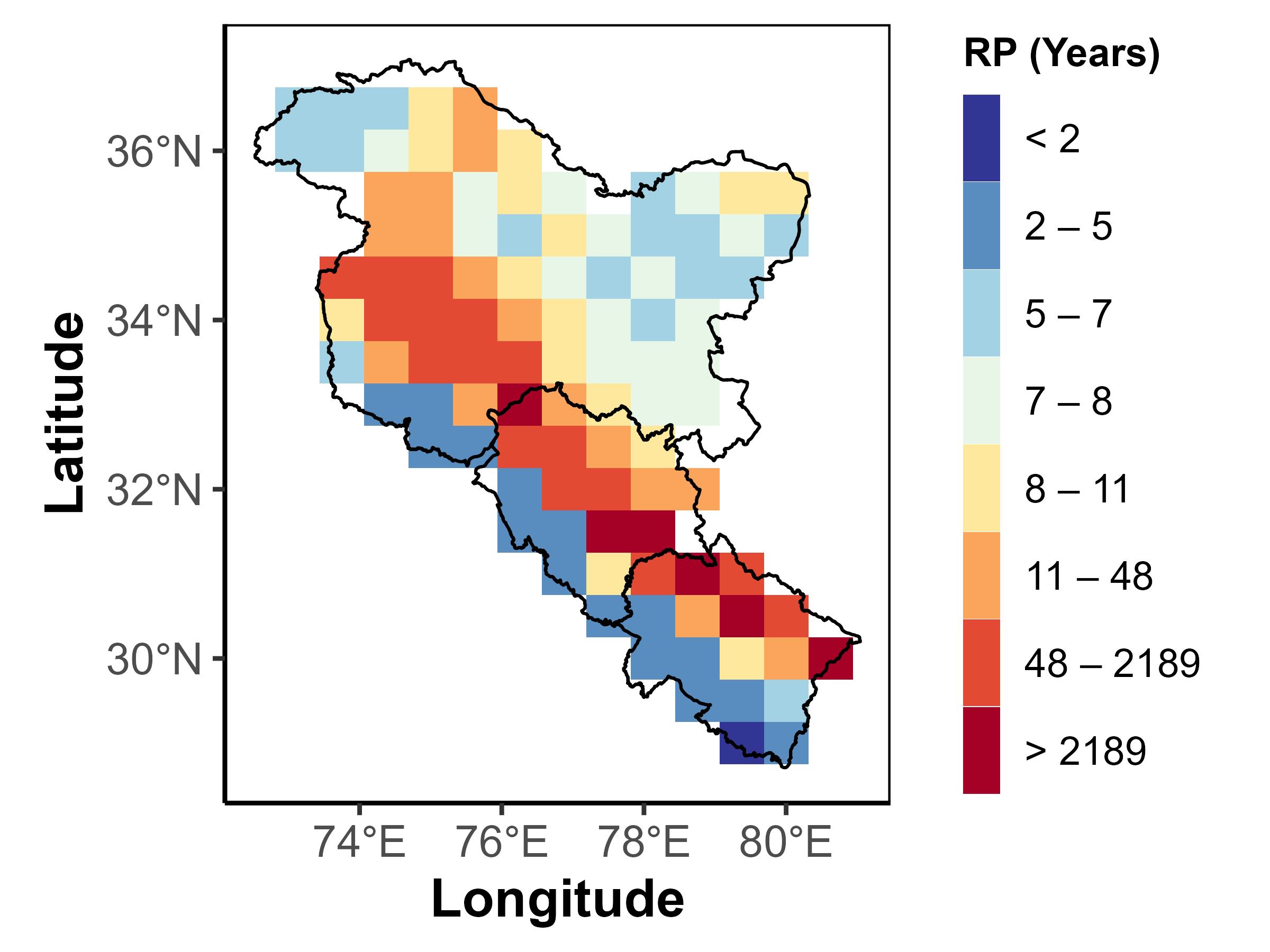}
		\caption*{(a)} 
	\end{minipage}\hfill
	\begin{minipage}[t]{0.45\textwidth}
		\centering
		\includegraphics[width=1\textwidth]{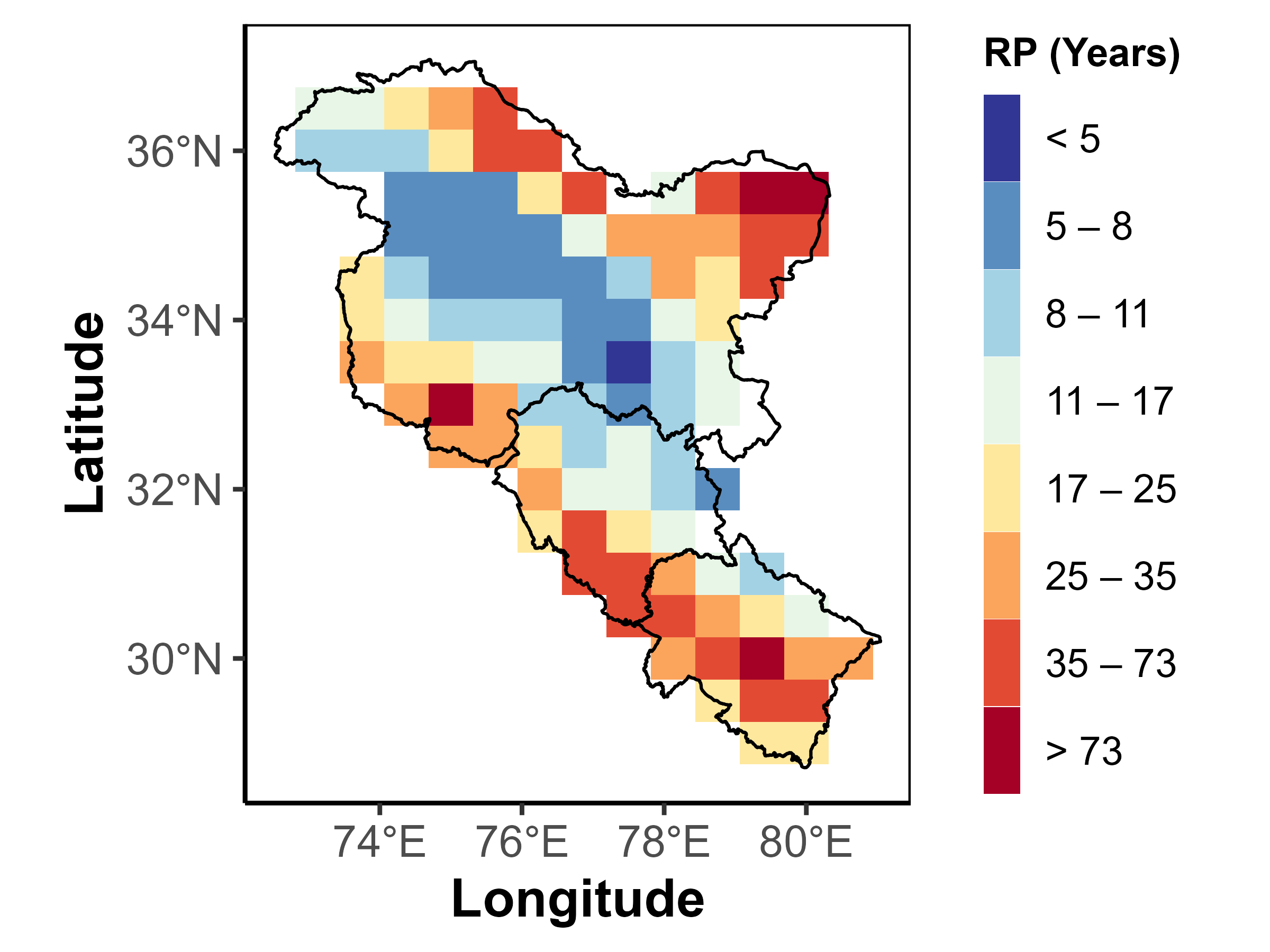}
		\caption*{(b)} 
	\end{minipage}
	\caption{Univariate return periods calculated from monthly temperature data: (a, b) summer and winter maximum temperatures surpassing the mean temperature by $4^\circ$C, respectively.}
	\label{Figure 13}
\end{figure}

As shown in Figure \ref{Figure 13}, the univariate return periods are derived from monthly temperature and rainfall data. Panels (a) and (b) depict the return periods of summer and winter maximum temperatures exceeding the mean temperature by $4^\circ$C, respectively. From Figure \ref{Figure 13} (a) and (b), it is evident that exceedances of summer maximum temperature above the mean are more frequent (i.e., associated with shorter return periods) than in winter over the LH and eastern parts of JK. In contrast, this pattern reverses across the MH region, where winter exceedances occur more often, reflected by comparatively smaller return periods than those observed during summer. 

\begin{figure}[ht]
	\begin{minipage}[t]{0.45\textwidth}
		\centering
		\includegraphics[width=1\textwidth]{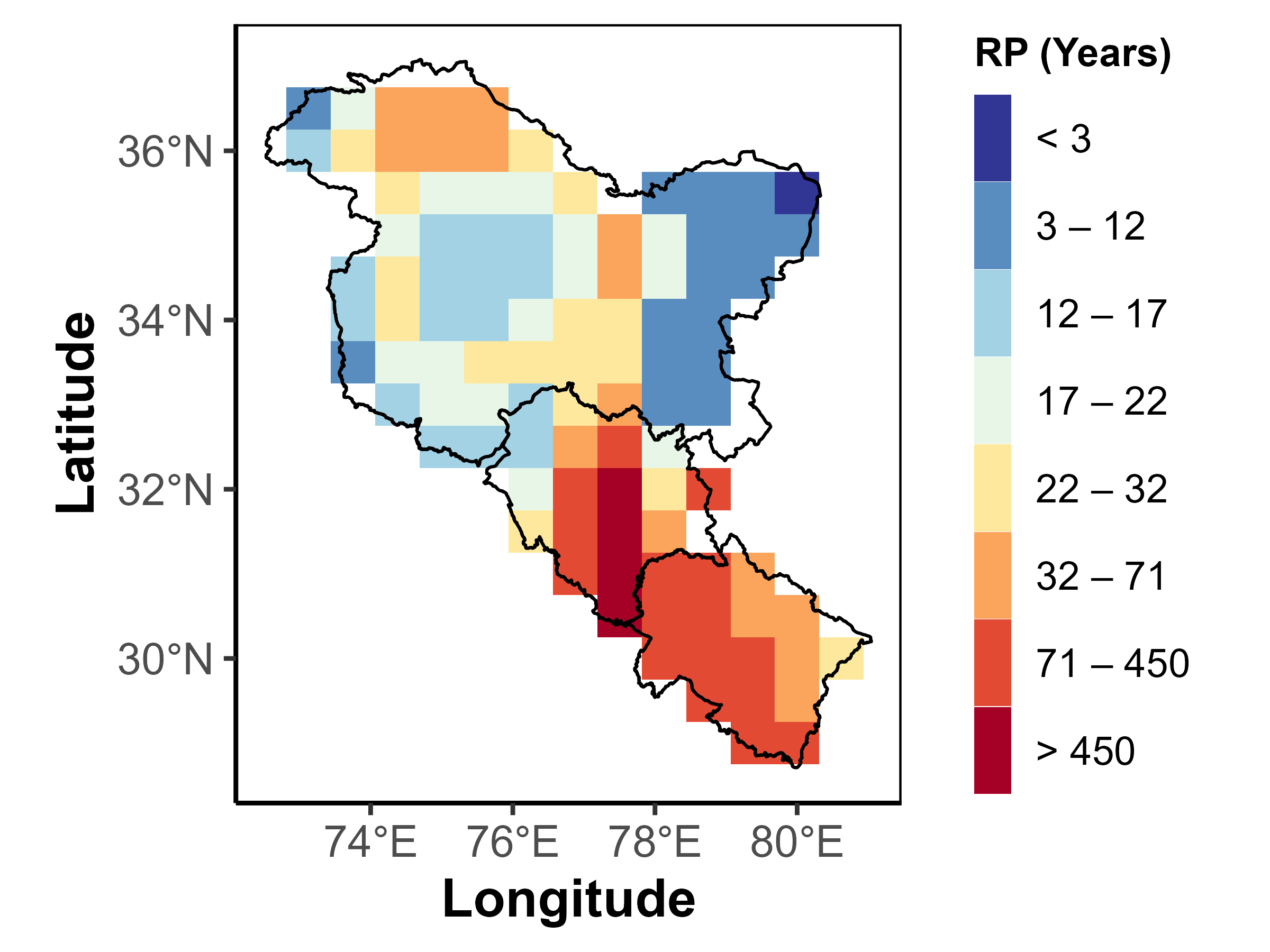}
		\caption*{(a)}
	\end{minipage} \hfill
	\begin{minipage}[t]{0.45\textwidth}
		\centering
		\includegraphics[width=1\textwidth]{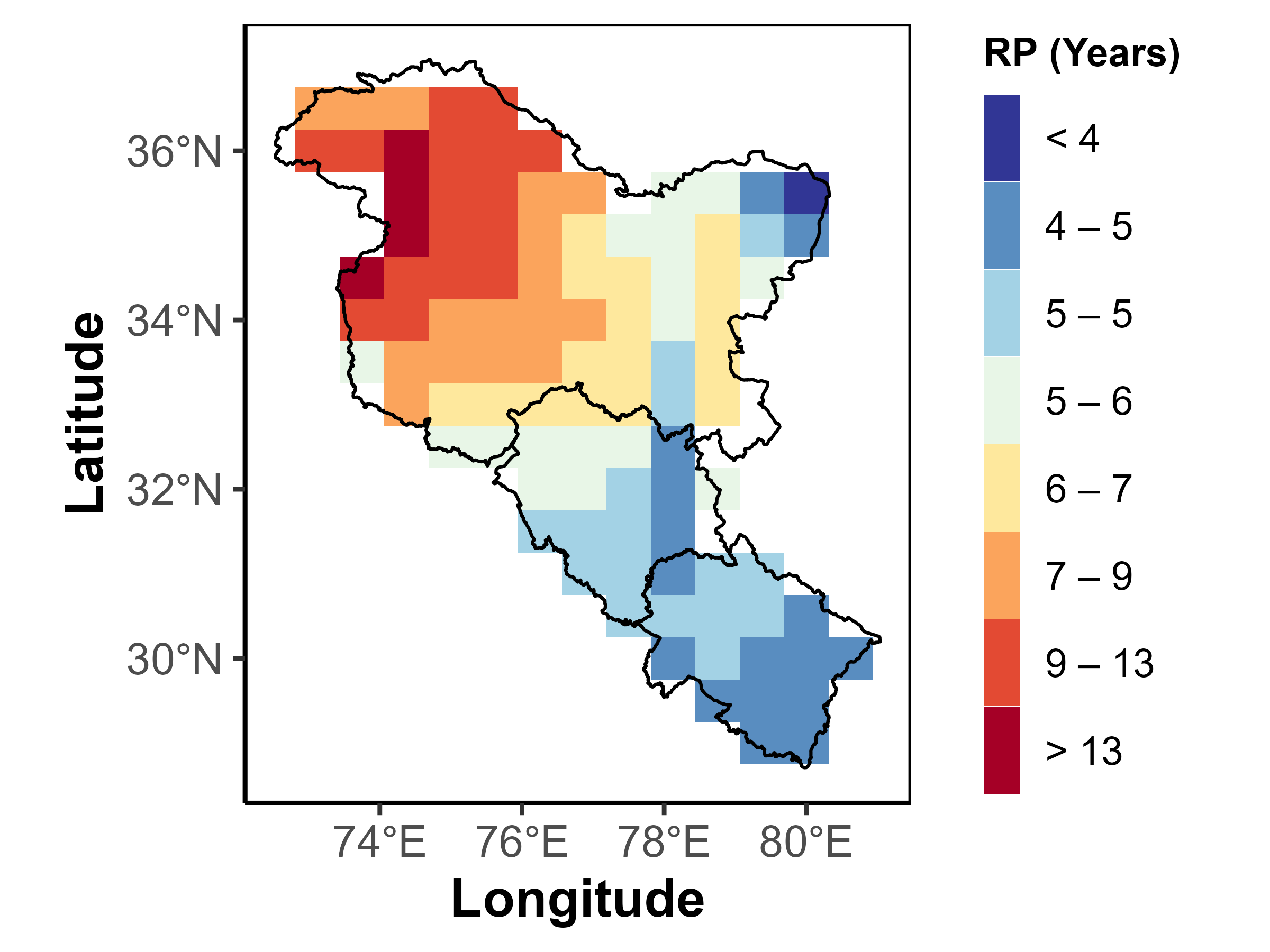}
		\caption*{(b)}
	\end{minipage}
	\caption{Univariate return periods calculated from monthly rainfall data: (a, b) summer and winter maximum rainfall reaching three times the mean rainfall, respectively.}
	\label{Figure 16}
\end{figure}

As shown in Figure \ref{Figure 16}, the univariate return periods are derived from monthly rainfall data. Panels (a) and (b) present the return periods of summer and winter maximum rainfall reaching three times the mean rainfall, respectively. From Figure \ref{Figure 16} (a) and (b), it is evident that across the entire NWH, the return periods associated with summer maximum rainfall exceeding three times the mean are consistently shorter than those observed in winter, indicating a higher likelihood of such extreme rainfall events during the summer season. Within the summer period, JK exhibits comparatively lower return periods than UK and HP, suggesting a greater frequency of extreme rainfall exceedances in this region. This spatial variability can be attributed to differences in baseline climatology: regions with relatively lower mean rainfall tend to have smaller thresholds for the “three times mean” criterion, making exceedances more probable and thereby resulting in shorter return periods, whereas regions with higher mean rainfall require substantially larger absolute amounts to reach the same threshold, leading to comparatively longer return periods.

\subsubsection{Joint Bivariate Return Periods}
In estimating univariate return periods, the interdependence between variables is not accounted for. To incorporate this dependence, joint bivariate return periods are considered, which can be evaluated under the following two cases:
\begin{enumerate}
	\item \textbf{OR case:} When at least one of the variables exceeds its threshold,
	\begin{equation}
		T_{\text{OR}} 
		= \frac{1}{P(X_1 \geq x_1 \text{ or } X_2 \geq x_2)}
		= \frac{1}{1 - C\left(F_{X_1}(x_1), F_{X_2}(x_2)\right)}
	\end{equation}
		\item \textbf{AND case:} When both variables simultaneously exceed their respective thresholds,
	\begin{equation}
		T_{\text{AND}} 
		= \frac{1}{P(X_1 \geq x_1 \text{ and } X_2 \geq x_2)}
		= \frac{1}{1 - F_{X_1}(x_1) - F_{X_2}(x_2) 
			+ C\left(F_{X_1}(x_1), F_{X_2}(x_2)\right)}
	\end{equation}
\end{enumerate}
where $C(\cdot)$ denotes the fitted copula function describing the dependence structure between temperature and rainfall. 
\begin{figure}[ht]
	\begin{minipage}[t]{0.45\textwidth}
		\centering
		\includegraphics[width=1\textwidth]{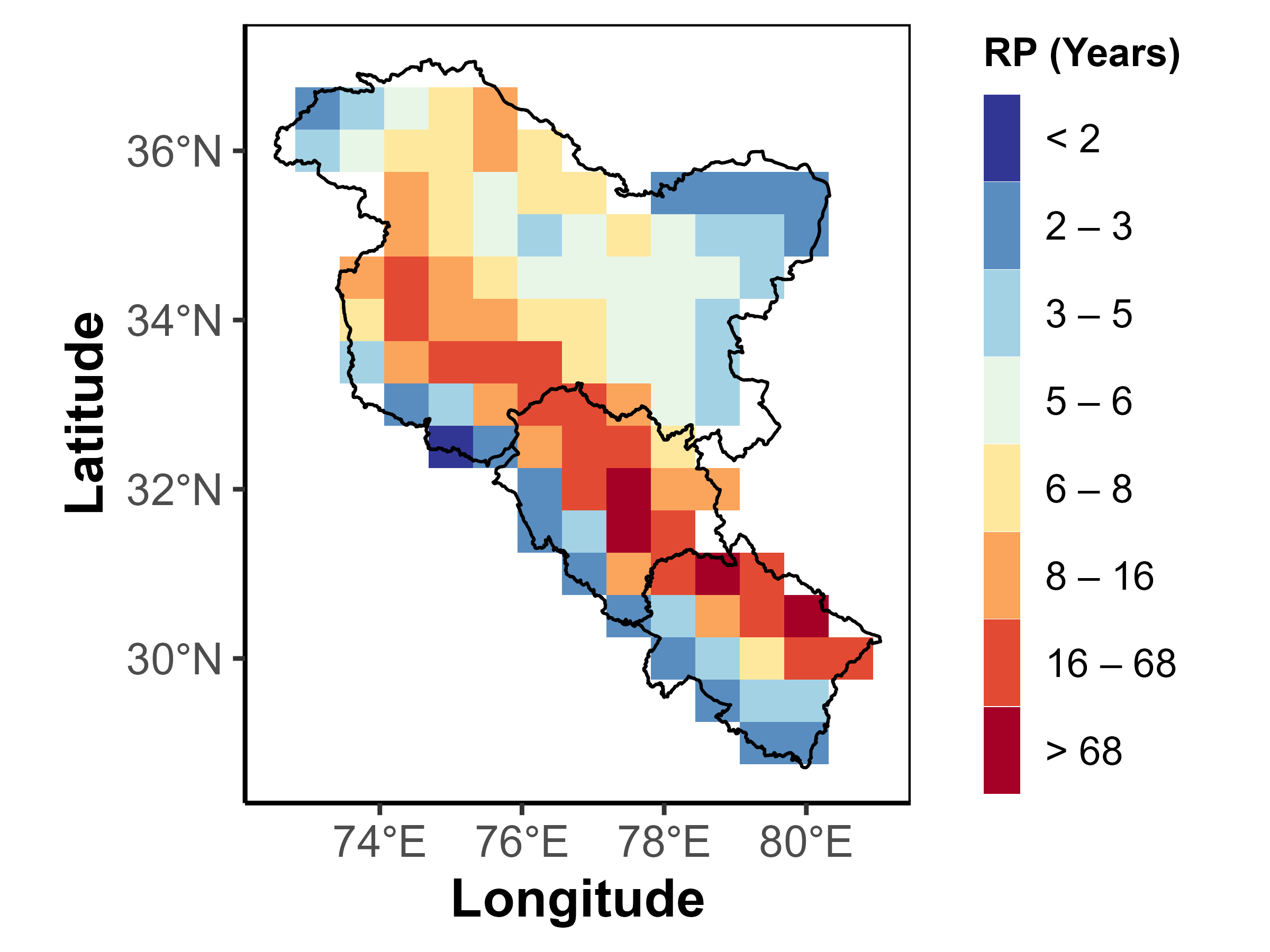}
		\caption*{(a)} 
	\end{minipage}\hfill
	\begin{minipage}[t]{0.45\textwidth}
		\centering
		\includegraphics[width=1\textwidth]{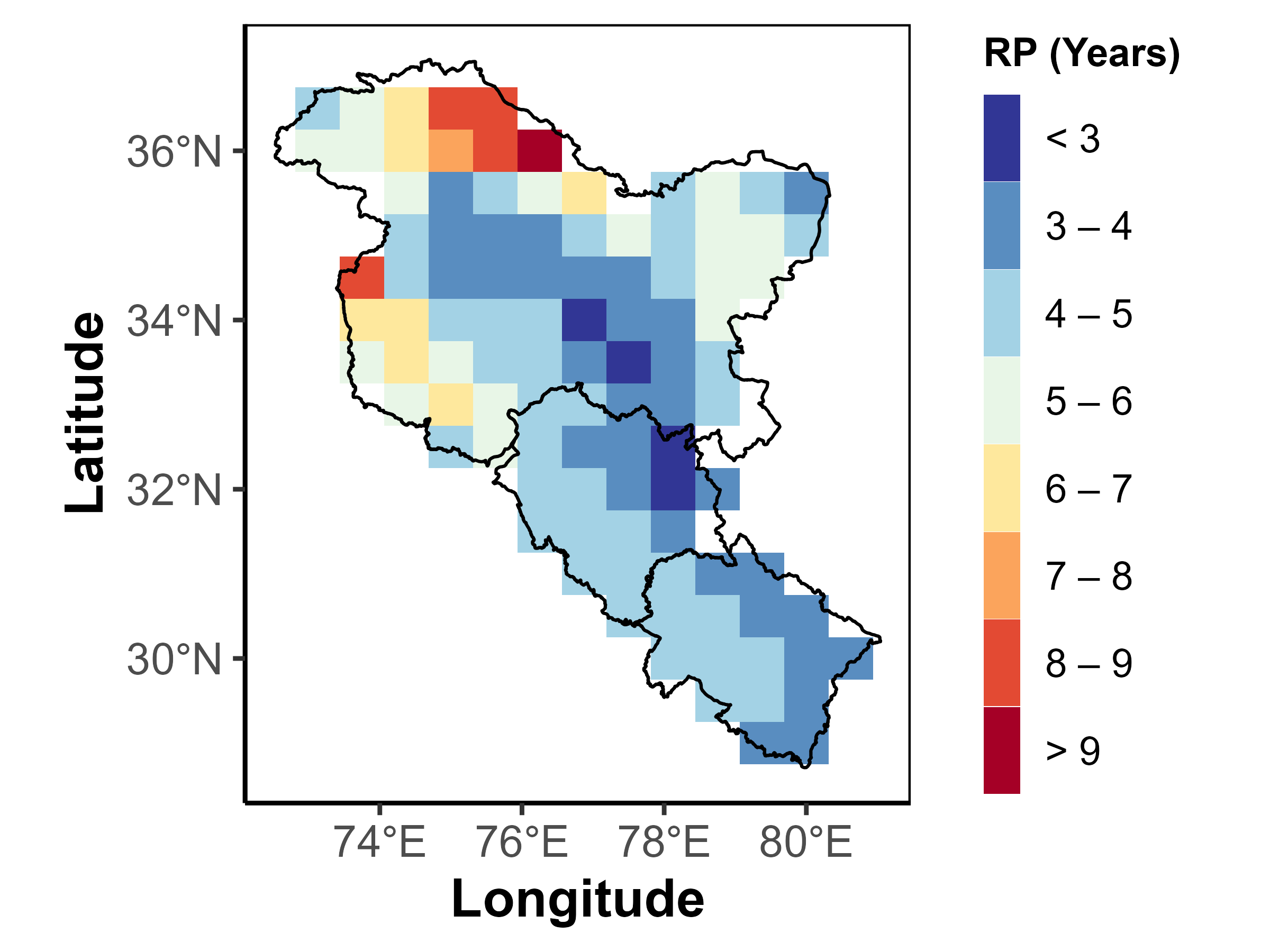}
		\caption*{(b)} 
	\end{minipage}	
	\caption{Bivariate joint return period: (a, b) summer and winter maximum temperatures exceeding the mean temperature by $4^\circ$C or maximum rainfall attaining three times the mean rainfall, respectively.}
	\label{Figure 14}
\end{figure}

Figure \ref{Figure 14} represents the return periods associated with either maximum temperature exceeding the mean by $4^\circ$C or maximum rainfall reaching three times the mean rainfall during the summer and winter seasons, respectively. Consistent with the univariate analysis, summer return periods are generally longer than those in winter, and this pattern persists in the joint OR framework. During summer as shown in Figure \ref{Figure 14} (a), LH and eastern parts of JK exhibit relatively shorter joint return periods—similar to the behavior observed for univariate temperature—whereas the MH region shows comparatively longer return periods. In contrast, during winter as represented in Figure \ref{Figure 14} (b), the entire NWH has return periods of less than 10 years, indicating that within a decade, at least one of the variables—temperature or rainfall—is likely to exceed the specified thresholds.

\begin{figure}[ht]
	\begin{minipage}[t]{0.45\textwidth}
		\centering
		\includegraphics[width=1\textwidth]{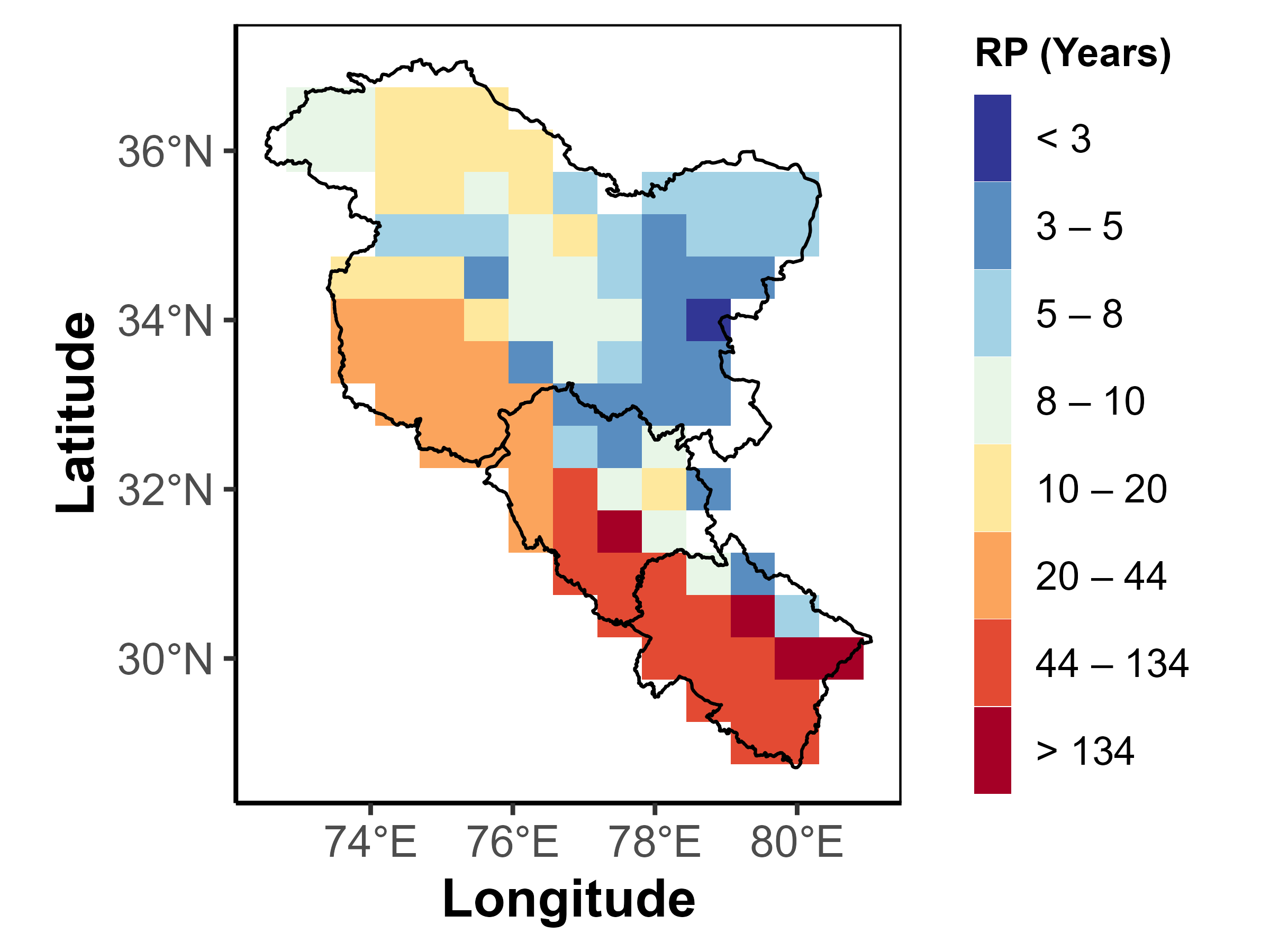}
		\caption*{(a)}
	\end{minipage}\hfill
	\begin{minipage}[t]{0.45\textwidth}
		\centering
		\includegraphics[width=1\textwidth]{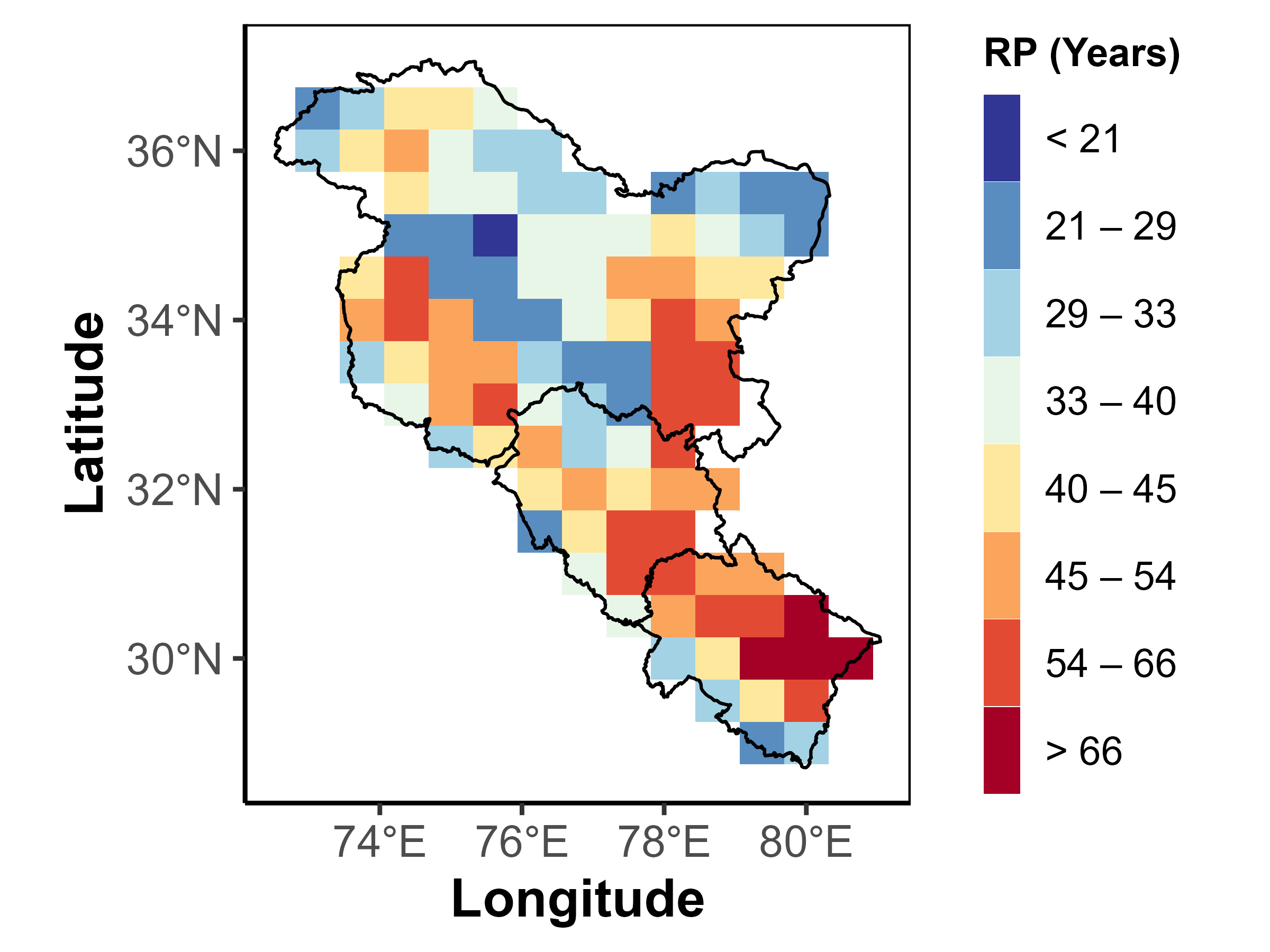}
		\caption*{(b)}
	\end{minipage}
	\caption{Bivariate joint return period: (a, b) summer and winter maximum temperatures exceeding the mean temperature by $2^\circ$C and maximum rainfall attaining two times the mean rainfall, respectively.}
	\label{Figure 17}
\end{figure}

Figure \ref{Figure 17} depicts the bivariate joint return periods associated with the simultaneous occurrence of maximum temperature exceeding the mean by 
$2^\circ$C and maximum rainfall reaching twice the mean during the summer and winter seasons. During summer (Figure \ref{Figure 17} (a)), LH and MH regions, covering most parts of UK and HP, exhibit relatively higher return periods, indicating that the concurrent exceedance of both temperature and rainfall thresholds is less frequent. In contrast, UH, particularly over JK, shows comparatively lower return periods, suggesting a higher likelihood of joint extremes. In winter (Figure \ref{Figure 17} (b)), LH, especially along the western boundary, displays lower return periods compared to other parts of UK and HP, implying more frequent joint exceedance. Overall, JK shows consistently lower return periods in both seasons, indicating that simultaneous exceedance of temperature and rainfall thresholds are more likely to occur in this region.

\subsubsection{Conditional Bivariate Return Periods}
After estimating the joint return periods, we want to quantify how the occurrence of one variable modifies the likelihood of extremes in the other. This dependence is captured through conditional return periods, which explicitly account for the effect of one variable given that the other has already reached or exceeded a specified threshold which can be defined as:
\begin{equation}
	T (X_1 \mid X_2) 
	= \frac{1}{1-\frac{F_{X_1}(x_1) - C\left(F_{X_1}(x_1), F_{X_2}(x_2)\right)}{1 - F_{X_2}(x_2)}}
\end{equation}
\begin{equation}
	T (X_2 \mid X_1) 
	= \frac{1}{1-\frac{F_{X_2}(x_2) - C\left(F_{X_1}(x_1), F_{X_2}(x_2)\right)}{1 - F_{X_1}(x_1)}}
\end{equation}

\begin{figure}[ht]
	\begin{minipage}[t]{0.45\textwidth}
		\centering
		\includegraphics[width=1\textwidth]{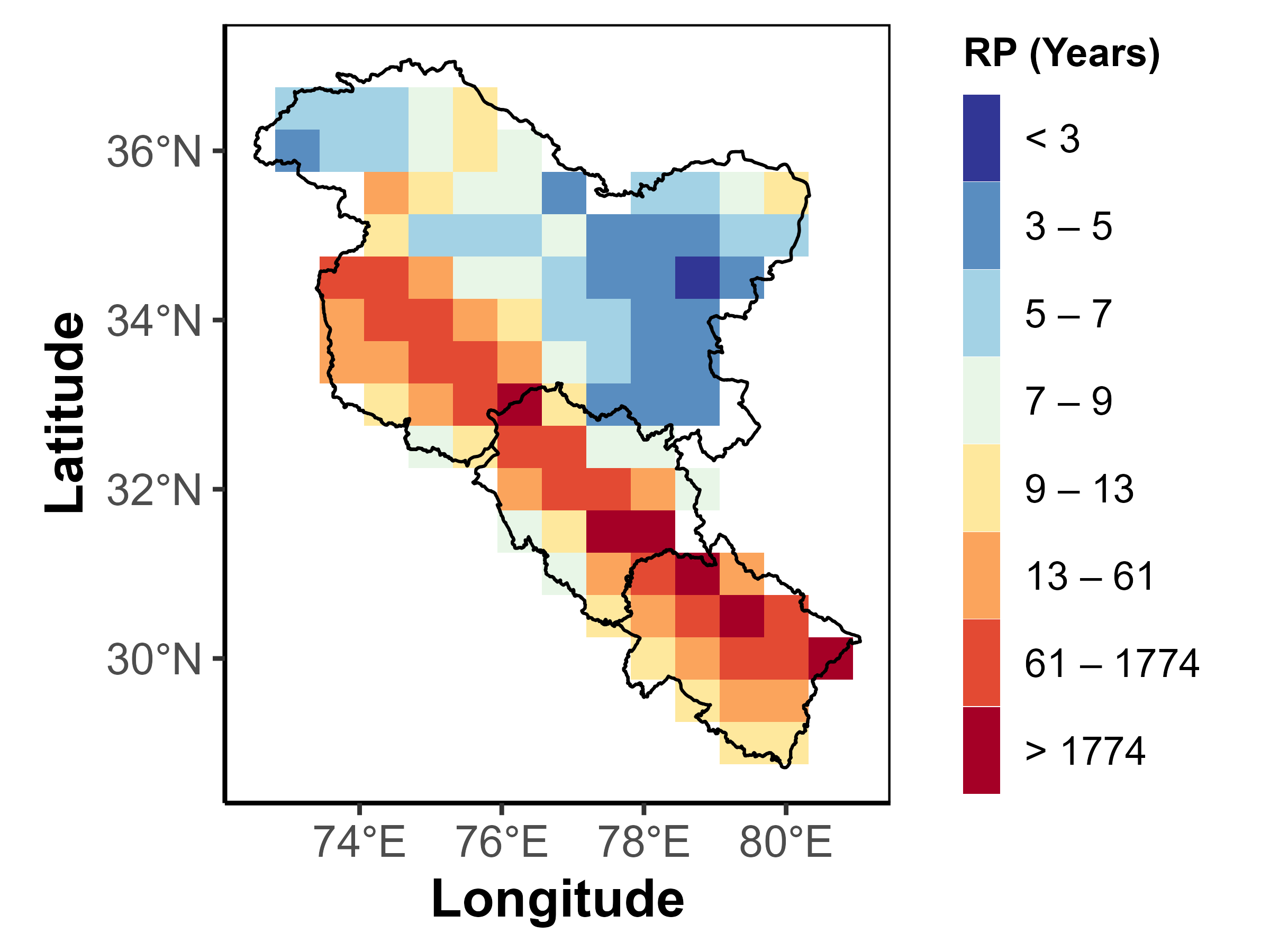}
		\caption*{(a)}
	\end{minipage}\hfill
	\begin{minipage}[t]{0.45\textwidth}
		\centering
		\includegraphics[width=1\textwidth]{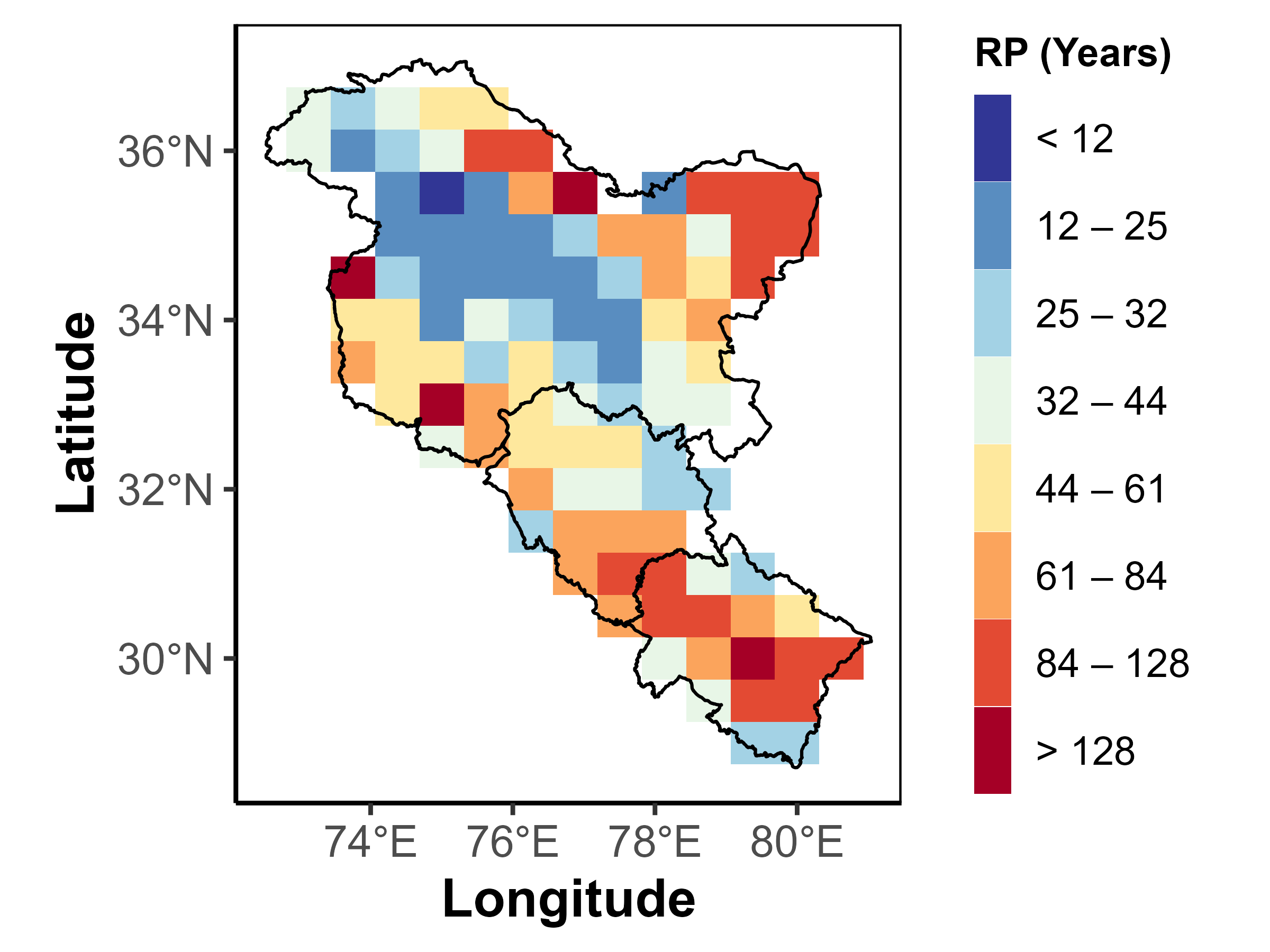}
		\caption*{(b)}
	\end{minipage}
	\caption{Conditional return periods: (a, b) summer and winter maximum temperatures surpassing the mean temperature by $4^\circ$C given maximum rainfall is equal to the mean rainfall, respectively.}
	\label{Figure 15}
\end{figure}

Figure \ref{Figure 15} illustrates the conditional return periods associated with maximum temperature exceeding the mean by $4^\circ$C, conditioned on rainfall being fixed at its mean, for summer and winter, respectively. In Figure \ref{Figure 15}(a), the summer patterns largely resemble the univariate case in both spatial distribution and magnitude, with the notable exception of LH, where higher conditional return periods are observed. This increase indicates the influence of negative dependence between temperature and rainfall, leading to reduced likelihood of temperature exceedance when rainfall is held constant. In winter (Figure \ref{Figure 15}(b)), the overall spatial structure remains comparable to the univariate case; however, return periods are generally elevated. This reflects the moderating effect of rainfall on temperature, consistent with negative dependence, resulting in less frequent temperature extremes under the given condition.
\begin{figure}[ht]
	\begin{minipage}[t]{0.45\textwidth}
		\centering
		\includegraphics[width=1\textwidth]{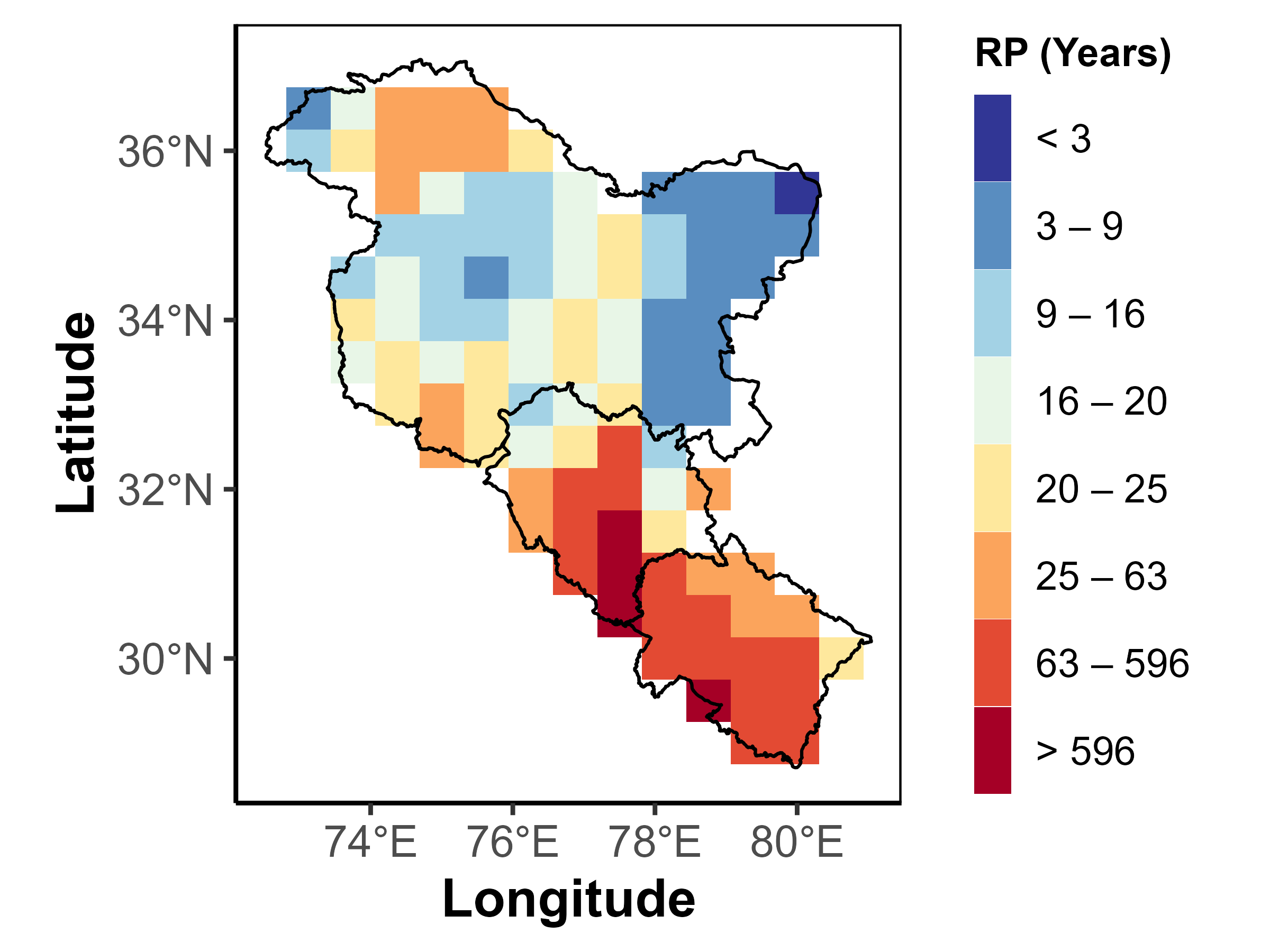}
		\caption*{(a)} 
	\end{minipage}\hfill
	\begin{minipage}[t]{0.45\textwidth}
		\centering
		\includegraphics[width=1\textwidth]{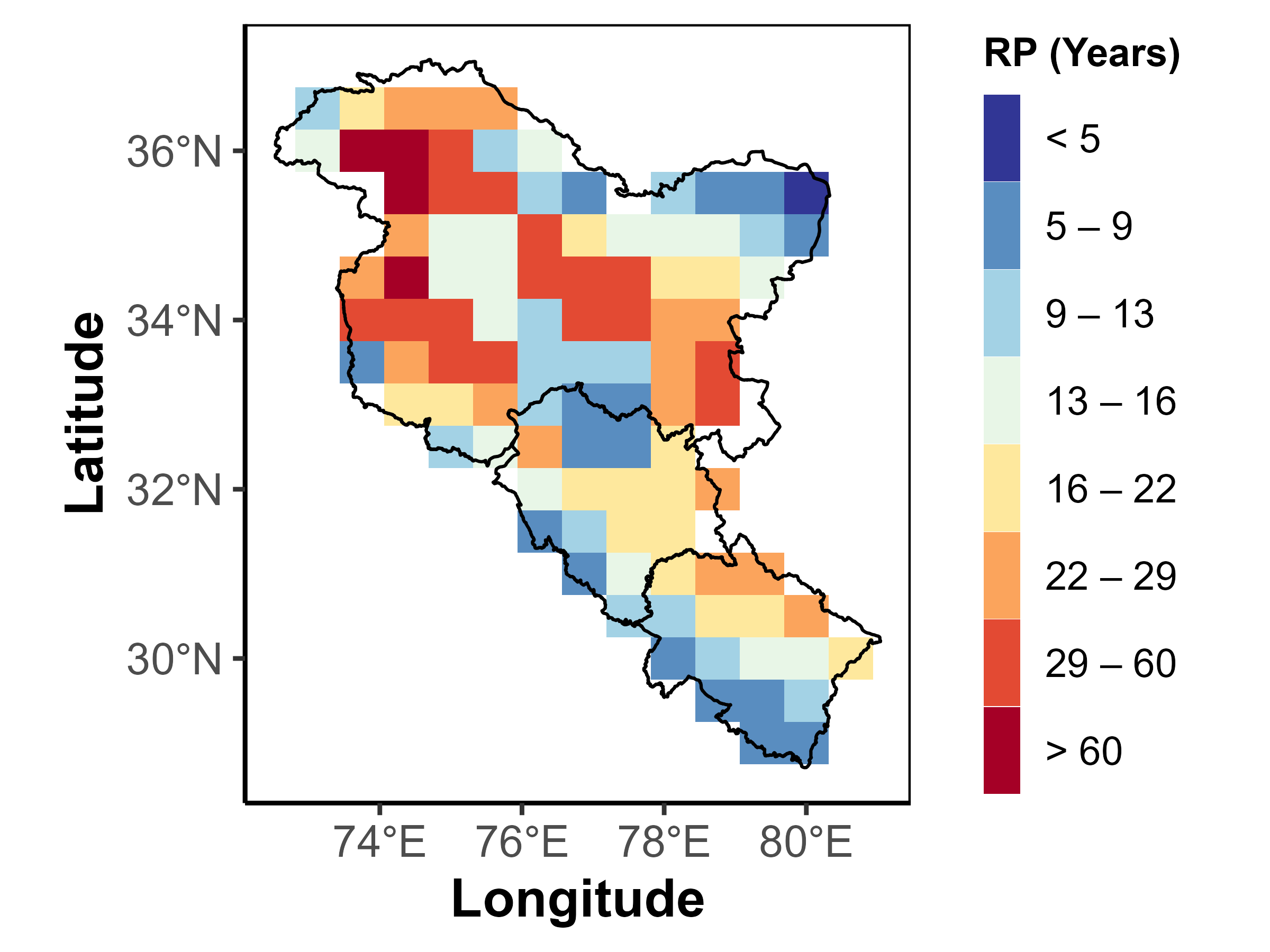}
		\caption*{(c)} 
	\end{minipage}
	\caption{Conditional return periods: (a, b) summer and winter maximum rainfall attaining three times the mean rainfall given summer and winter maximum temperatures is equal to the mean temperature respectively.}
	\label{Figure 18}
\end{figure}

Figure \ref{Figure 18} presents the conditional return periods for rainfall exceeding three times the mean, given temperature is fixed at its mean level, for both seasons. In summer (Figure \ref{Figure 18}(a)), the spatial pattern remains consistent with the univariate case, while the overall increase in return periods suggests a constraining influence of temperature on rainfall extremes. In winter (Figure \ref{Figure 18}(b)), distinct spatial variability is observed, with LH and eastern JK exhibiting relatively lower return periods, indicating a higher likelihood of extreme rainfall occurrences in these regions compared to MH and UH.

\section{Conclusion}
In this study, bivariate Clayton and Gumbel Kumaraswamy–Teissier distributions are developed with seven-parameter framework to flexibly model dependence structures between variables. To accommodate both positive and negative dependence, rotated copula forms (90°, 180°, and 270°) are incorporated, enabling a comprehensive characterization of tail dependence. Parameter estimation is carried out using MLE and IFM, and their finite-sample performance is assessed through a Monte Carlo simulation study, demonstrating accuracy and stability of the estimators. The practical applicability of the proposed models is demonstrated using monthly gridded rainfall and temperature data over NWH. The analysis indicates that positively correlated grids are well captured by the standard Clayton and Gumbel KTD, whereas negatively correlated structures are effectively modeled by the rotated $90^\circ$ copula. The adequacy of the proposed models across all grids is supported by CVM test results, along with improved goodness-of-fit measures based on $\mathcal{L}$, AIC, and BIC in comparison to several existing bivariate models. The selected models successfully capture the underlying tail dependence structures, including lower, upper, and cross-tail behaviour specific to each grid. Furthermore, to assess the occurrence of extreme events, uni return periods based on KTD, as well as joint and conditional return periods using the copula framework, are estimated. These results provide insight into the frequency and co-occurrence of extreme rainfall and temperature events. Overall, the derived tail dependence measures and return period estimates offer a robust framework for understanding and quantifying extreme climate events. The introduced modeling approach can be utilized for hydrological applications and for reliable assessment of flood risk and extreme temperature events over NWH.

\subsection*{Acknowledgments}
We express our gratitude to all the data providers for their contributions to the statistical data analysis, duly acknowledged through appropriate citations.

\appendix
\subsection*{Appendix A}\label{appendixA}
\begin{eqnarray*}
\frac{\partial \mathcal{L}(\pv_1)}{\partial \a_1} &=& \frac{\ss}{\a_1} + \sum_{\i=1}^{\ss} \log y_{1i} \left[1- \frac{ (\b_1-1)}{y_{1i}^{-\a_1}-1} \right] - \sum_{\i=1}^{\ss} \frac{b_1\ y_{1i}^{a_1} \ \log(y_{1i}) \ \left(1-y_{1i}^{a_1} \right)^{b_1-1}  }{ \left( 1 - \left(1- y_{1i} ^{a_1} \right) ^{\b_1} \right)}  \nonumber \\ && \left[ (\delta_1+1) -  
\frac{(2\delta_1+1) \left( 1 - \left(1- y_{1i} ^{a_1} \right) ^{\b_1} \right)^{-\delta_1}}{\left( 1 - \left(1- y_{1i} ^{a_1} \right) ^{\b_1} \right)^{-\delta_1} + \left( 1 - \left(1- y_{1i} ^{a_1} \right) ^{\b_1} \right)^{-\delta_1} - 1}\right] = 0 \\
\frac{\partial \mathcal{L}(\pv_1)}{\partial \a_2} &=& \frac{\ss}{\a_2} + \sum_{\i=1}^{\ss} \log y_{2i} \left[1- \frac{ (\b_2-1)}{y_{2i}^{-\a_2}-1} \right] - \sum_{\i=1}^{\ss} \frac{b_2\ y_{2i}^{a_2} \ \log(y_{2i}) \ \left(1-y_{2i}^{a_2} \right)^{b_2-1}  }{ \left( 1 - \left(1- y_{2i} ^{a_2} \right) ^{\b_2} \right)}  \nonumber \\ && \left[ (\delta_1+1) -  
\frac{(2\delta_1+1) \left( 1 - \left(1- y_{2i} ^{a_2} \right) ^{\b_2} \right)^{-\delta_1}}{\left( 1 - \left(1- y_{2i} ^{a_2} \right) ^{\b_2} \right)^{-\delta_1} + \left( 1 - \left(1- y_{2i} ^{a_2} \right) ^{\b_2} \right)^{-\delta_1} - 1}\right] = 0\\
\frac{\partial \mathcal{L}(\pv_1)}{\partial \b_1} &=& \frac{\ss}{\b_1} + \sum_{\i=1}^{\ss} \log\left(1- y_{1i} ^{\a_1} \right) + \sum_{\i=1}^{\ss} \frac{ \ \left(1-y_{1i}^{a_1} \right)^{b_1} \log\left(1-y_{1i}^{a_1} \right)   }{ \left( 1 - \left(1- y_{1i} ^{a_1} \right) ^{\b_1} \right)}  \nonumber \\ && \left[ (\delta_1+1) -  
\frac{(2\delta_1+1) \left( 1 - \left(1- y_{1i} ^{a_1} \right) ^{\b_1} \right)^{-\delta_1}}{\left( 1 - \left(1- y_{1i} ^{a_1} \right) ^{\b_1} \right)^{-\delta_1} + \left( 1 - \left(1- y_{1i} ^{a_1} \right) ^{\b_1} \right)^{-\delta_1} - 1}\right] = 0\\
\frac{\partial \mathcal{L}(\pv_1)}{\partial \b_2} &=& \frac{\ss}{\b_2} + \sum_{\i=1}^{\ss} \log\left(1- y_{2i} ^{\a_2} \right)  + \sum_{\i=1}^{\ss} \frac{ \ \left(1-y_{2i}^{a_2} \right)^{b_2} \log\left(1-y_{2i}^{a_2} \right)   }{ \left( 1 - \left(1- y_{2i} ^{a_2} \right) ^{\b_2} \right)}  \nonumber \\ && \left[ (\delta_1+1) -  
\frac{(2\delta_1+1) \left( 1 - \left(1- y_{2i} ^{a_2} \right) ^{\b_2} \right)^{-\delta_1}}{\left( 1 - \left(1- y_{2i} ^{a_2} \right) ^{\b_2} \right)^{-\delta_1} + \left( 1 - \left(1- y_{2i} ^{a_2} \right) ^{\b_2} \right)^{-\delta_1} - 1}\right] = 0\\
\frac{\partial \mathcal{L}(\pv_1)}{\partial \th_1} &=& \frac{\ss}{\th_1} + \sum_{\i=1}^{\ss} \frac{e^{\th_1 x_{1i}} \ x_{1i}}  {\left(e^{\th_1 x_{1i}}-1 \right)} - \sum_{\i=1}^{\ss} x_{1i} \left(e^{\th x_{1i}}-1 \right)+ \sum_{\i=1}^{\ss} \frac{x_{1i} e^{\phi(\x_{1i}; \th_1) } (e^{\th x_{1i}} -1)}{y_{1i}} \left[ \a_1-1 - \frac{\a_1 (\b_1-1)}{y_{1i}^{-\a_1}-1} \right]  \nonumber \\ && -\sum_{\i=1}^{\ss} \frac{a_1 \ b_1\ x_{1i} \ e^{\phi(\x_{1i}; \th_1) } \ (e^{\th_1 x_{1i}} -1) y_{1i}^{a_1-1} \ \left(1-y_{1i}^{a_1} \right)^{b_1-1}  }{ \left( 1 - \left(1- y_{1i} ^{a_1} \right) ^{\b_1} \right)}  \nonumber \\ && \left[ (\delta_1+1) -  \frac{(2\delta_1+1) \left( 1 - \left(1- y_{1i} ^{a_1} \right) ^{\b_1} \right)^{-\delta_1}}{\left( 1 - \left(1- y_{1i} ^{a_1} \right) ^{\b_1} \right)^{-\delta_1} + \left( 1 - \left(1- y_{1i} ^{a_1} \right) ^{\b_1} \right)^{-\delta_1} - 1}\right] = 0\\
\frac{\partial \mathcal{L}(\pv_1)}{\partial \th_2} &=& \frac{\ss}{\th_2} + \sum_{\i=1}^{\ss} \frac{e^{\th_2 x_{2i}} \ x_{2i}}  {\left(e^{\th_2 x_{2i}}-1 \right)} - \sum_{\i=1}^{\ss} x_{2i} \left(e^{\th x_{2i}}-1 \right)+ \sum_{\i=1}^{\ss} \frac{x_{2i} e^{\phi(\x_{2i}; \th_2) } (e^{\th_2 x_{2i}} -1)}{y_{2i}} \left[ \a_2-1 - \frac{\a_2 (\b_2-1)}{y_{2i}^{-\a_2}-1} \right]  \nonumber \\ && -\sum_{\i=1}^{\ss} \frac{a_2 \ b_2\ x_{2i} \ e^{\phi(\x_{2i}; \th_2) } \ (e^{\th_2 x_{2i}} -1) y_{2i}^{a_2-1} \ \left(1-y_{2i}^{a_2} \right)^{b_2-1}  }{ \left( 1 - \left(1- y_{2i} ^{a_2} \right) ^{\b_2} \right)}  \nonumber \\ && \left[ (\delta_1+1) -  \frac{(2\delta_1+1) \left( 1 - \left(1- y_{2i} ^{a_2} \right) ^{\b_2} \right)^{-\delta_1}}{\left( 1 - \left(1- y_{2i} ^{a_2} \right) ^{\b_2} \right)^{-\delta_1} + \left( 1 - \left(1- y_{2i} ^{a_2} \right) ^{\b_2} \right)^{-\delta_1} - 1}\right] = 0\\
\frac{\partial \mathcal{L}(\pv_1)}{\partial \delta_1} &=& \frac{n}{\delta_1+1} -\sum_{\i=1}^{\ss} \log \left( 1 - \left(1- y_{1i} ^{a_1} \right) ^{\b_1} \right) - \sum_{\i=1}^{\ss} \log \left( 1 - \left(1- y_{2i} ^{a_2} \right) ^{\b_2} \right) \nonumber \\ &&+ \frac{1}{\delta_1^2} \sum_{\i=1}^{\ss} \log \left( \left( 1 - \left(1- y_{1i} ^{a_1} \right) ^{\b_1} \right)^{-\delta_1} + \left( 1 - \left(1- y_{2i} ^{a_2} \right) ^{\b_2} \right)^{-\delta_1} - 1 \right) - \frac{(2\delta_1+1)}{\delta_1} \times \nonumber \\ &&  \sum_{\i=1}^{\ss} \frac{ \left( 1 - \left(1- y_{1i} ^{a_1} \right) ^{\b_1} \right)^{-\delta_1} \log\left( 1 - \left(1- y_{1i} ^{a_1} \right) ^{\b_1} \right) + \left( 1 - \left(1- y_{2i} ^{a_2} \right) ^{\b_2} \right)^{-\delta_1} \log\left( 1 - \left(1- y_{2i} ^{a_2} \right) ^{\b_2} \right)}{\left( 1 - \left(1- y_{1i} ^{a_1} \right) ^{\b_1} \right)^{-\delta_1}  + \left( 1 - \left(1- y_{2i} ^{a_2} \right) ^{\b_2} \right)^{-\delta_1} -1}\nonumber \\ && = 0
\end{eqnarray*}

\subsection*{Appendix B }\label{appendixB}

\begin{eqnarray*}
\frac{\partial \mathcal{L}(\pv_2)}{\partial \a_1} &=& \frac{\ss}{\a_1} + \sum_{\i=1}^{\ss} \log y_{1i} \left[1- \frac{ (\b_1-1)}{y_{1i}^{-\a_1}-1} \right] +  \sum_{\i=1}^{\ss} \frac{ b_1\ y_{1i}^{a_1} \ \log(y_{1i}) \ \left(1-y_{1i}^{a_1} \right)^{b_1-1}}{z_i \left\{-\log \left( 1 - \left(1- y_{1i} ^{a_1} \right) ^{\b_1} \right) \right\}^{1-{\delta_2}}}
\nonumber \\ && \left[ (2{\delta_2}-1) - \frac{1}{1+({\delta_2}-1) z_i^{-1/{\delta_2}}} + z_i^{1/{\delta_2}} \right] - \sum_{\i=1}^{\ss} \frac{b_1\ y_{1i}^{a_1} \ \log(y_{1i}) \ \left(1-y_{1i}^{a_1} \right)^{b_1-1}}{\left( 1 - \left(1- y_{1i} ^{a_1} \right) ^{\b_1} \right)}\nonumber \\ &&  \left[ \frac{1-{\delta_2}}{\log \left( 1 - \left(1- y_{1i} ^{a_1} \right) ^{\b_1} \right) } -1 \right] =0\\
\frac{\partial \mathcal{L}(\pv_2)}{\partial \a_2} &=& \frac{\ss}{\a_2} + \sum_{\i=1}^{\ss} \log y_{2i} \left[1- \frac{ (\b_2-1)}{y_{2i}^{-\a_2}-1} \right]+  \sum_{\i=1}^{\ss} \frac{ b_2\ y_{2i}^{a_2} \ \log(y_{2i}) \ \left(1-y_{2i}^{a_2} \right)^{b_2-1}}{z_i \left\{-\log \left( 1 - \left(1- y_{2i} ^{a_2} \right) ^{\b_2} \right) \right\}^{1-{\delta_2}}}
\nonumber \\ && \left[ (2{\delta_2}-1) - \frac{1}{1+({\delta_2}-1) z_i^{-1/{\delta_2}}} + z_i^{1/{\delta_2}} \right] - \sum_{\i=1}^{\ss} \frac{b_2\ y_{2i}^{a_2} \ \log(y_{2i}) \ \left(1-y_{2i}^{a_2} \right)^{b_2-1}}{\left( 1 - \left(1- y_{2i} ^{a_2} \right) ^{\b_2} \right)}\nonumber \\ &&  \left[ \frac{1-{\delta_2}}{\log \left( 1 - \left(1- y_{2i} ^{a_2} \right) ^{\b_2} \right) } -1 \right] =0\\
\frac{\partial \mathcal{L}(\pv_2)}{\partial \b_1} &=& \frac{\ss}{\b_1} + \sum_{\i=1}^{\ss} \log\left(1- y_{1i} ^{\a_1} \right) -  \sum_{\i=1}^{\ss} \frac{ \left(1-y_{1i}^{a_1} \right)^{b_1} \log\left(1-y_{1i}^{a_1} \right)  }{z_i  \left\{-\log \left( 1 - \left(1- y_{1i} ^{a_1} \right) ^{\b_1} \right) \right\}^{1-{\delta_2}}} \Bigg[ (2{\delta_2}-1) 
\nonumber \\ && - \frac{1}{1+({\delta_2}-1) z_i^{-1/{\delta_2}}} + z_i^{1/{\delta_2}} \Bigg] + \sum_{\i=1}^{\ss} \frac{\left(1-y_{1i}^{a_1} \right)^{b_1} \log\left(1-y_{1i}^{a_1} \right)  }{\left( 1 - \left(1- y_{1i} ^{a_1} \right) ^{\b_1} \right)}  \left[ \frac{1-{\delta_2}}{\log \left( 1 - \left(1- y_{1i} ^{a_1} \right) ^{\b_1} \right) } -1 \right] \nonumber \\ && =0\\
\frac{\partial \mathcal{L}(\pv_2)}{\partial \b_2} &=& \frac{\ss}{\b_2} + \sum_{\i=1}^{\ss} \log\left(1- y_{2i} ^{\a_2} \right)-  \sum_{\i=1}^{\ss} \frac{ \left(1-y_{2i}^{a_2} \right)^{b_2} \log\left(1-y_{2i}^{a_2} \right)  }{z_i  \left\{-\log \left( 1 - \left(1- y_{2i} ^{a_2} \right) ^{\b_2} \right) \right\}^{1-{\delta_2}}}\Bigg[ (2{\delta_2}-1) 
\nonumber \\ && - \frac{1}{1+({\delta_2}-1) z_i^{-1/{\delta_2}}} + z_i^{1/{\delta_2}} \Bigg] + \sum_{\i=1}^{\ss} \frac{\left(1-y_{2i}^{a_2} \right)^{b_2} \log\left(1-y_{2i}^{a_2} \right)  }{\left( 1 - \left(1- y_{2i} ^{a_2} \right) ^{\b_2} \right)}  \left[ \frac{1-{\delta_2}}{\log \left( 1 - \left(1- y_{2i} ^{a_2} \right) ^{\b_2} \right) } -1 \right] \nonumber \\ &&=0 \\
\frac{\partial \mathcal{L}(\pv_2)}{\partial \th_1} &=& \frac{\ss}{\th_1} + \sum_{\i=1}^{\ss} \frac{e^{\th_1 x_{1i}} \ x_{1i}}  {\left(e^{\th_1 x_{1i}}-1 \right)} - \sum_{\i=1}^{\ss} x_{1i} \left(e^{\th x_{1i}}-1 \right)+ \sum_{\i=1}^{\ss} \frac{x_{1i} e^{\phi(\x_{1i}; \th_1) } (e^{\th x_{1i}} -1)}{y_{1i}} \left[ \a_1-1 - \frac{\a_1 (\b_1-1)}{y_{1i}^{-\a_1}-1} \right]  \nonumber \\ && +  \sum_{\i=1}^{\ss} \frac{a_1 \ b_1\ x_{1i} \ e^{\phi(\x_{1i}; \th_1) } \ (e^{\th_1 x_{1i}} -1) y_{1i}^{a_1-1} \ \left(1-y_{1i}^{a_1} \right)^{b_1-1}}{z_i  \left\{-\log \left( 1 - \left(1- y_{1i} ^{a_1} \right) ^{\b_1} \right) \right\}^{1-{\delta_2}}} \left[ (2{\delta_2}-1) - \frac{1}{1+({\delta_2}-1) z_i^{-1/{\delta_2}}} + z_i^{1/{\delta_2}} \right]
\nonumber \\ &&  - \sum_{\i=1}^{\ss} \frac{a_1 \ b_1\ x_{1i} \ e^{\phi(\x_{1i}; \th_1) } \ (e^{\th_1 x_{1i}} -1) y_{1i}^{a_1-1} \ \left(1-y_{1i}^{a_1} \right)^{b_1-1}}{\left( 1 - \left(1- y_{1i} ^{a_1} \right) ^{\b_1} \right)}  \left[ \frac{1-{\delta_2}}{\log \left( 1 - \left(1- y_{1i} ^{a_1} \right) ^{\b_1} \right) } -1 \right] = 0\\
\frac{\partial \mathcal{L}(\pv_2)}{\partial \th_2} &=& \frac{\ss}{\th_2} + \sum_{\i=1}^{\ss} \frac{e^{\th_2 x_{2i}} \ x_{2i}}  {\left(e^{\th_2 x_{2i}}-1 \right)} - \sum_{\i=1}^{\ss} x_{2i} \left(e^{\th x_{2i}}-1 \right)+ \sum_{\i=1}^{\ss} \frac{x_{2i} e^{\phi(\x_{2i}; \th_2) } (e^{\th_2 x_{2i}} -1)}{y_{2i}} \left[ \a_2-1 - \frac{\a_2 (\b_2-1)}{y_{2i}^{-\a_2}-1} \right]  \nonumber \\ && +  \sum_{\i=1}^{\ss} \frac{a_2 \ b_2\ x_{2i} \ e^{\phi(\x_{2i}; \th_2) } \ (e^{\th_2 x_{2i}} -1) y_{2i}^{a_2-1} \ \left(1-y_{2i}^{a_2} \right)^{b_2-1} }{z_i  \left\{-\log \left( 1 - \left(1- y_{2i} ^{a_2} \right) ^{\b_2} \right) \right\}^{1-{\delta_2}}} \left[ (2{\delta_2}-1) - \frac{1}{1+({\delta_2}-1) z_i^{-1/{\delta_2}}} + z_i^{1/{\delta_2}} \right] \nonumber \\ &&  - \sum_{\i=1}^{\ss} \frac{a_2 \ b_2\ x_{2i} \ e^{\phi(\x_{2i}; \th_2) } \ (e^{\th_2 x_{2i}} -1) y_{2i}^{a_2-1} \ \left(1-y_{2i}^{a_2} \right)^{b_2-1} }{\left( 1 - \left(1- y_{2i} ^{a_2} \right) ^{\b_2} \right)}  \left[ \frac{1-{\delta_2}}{\log \left( 1 - \left(1- y_{2i} ^{a_2} \right) ^{\b_2} \right) } -1 \right] = 0\\
\frac{\partial \mathcal{L}(\pv_2)}{\partial {\delta_2}} &=& \sum_{\i=1}^{\ss} \log \left[ \left\{\log \left( 1 - \left(1- y_{1i} ^{a_1} \right) ^{\b_1} \right) \right\}  \left\{\log \left( 1 - \left(1- y_{2i} ^{a_2} \right) ^{\b_2} \right) \right\} \right] -\frac{1}{{\delta_2}^2} \sum_{\i=1}^{\ss} \log(z_i) \nonumber \\ &&- (2-1/{\delta_2}) \sum_{\i=1}^{\ss} \frac{1}{z} \biggl[ \left\{-\log \left( 1 - \left(1- y_{1i} ^{a_1} \right) ^{\b_1} \right) \right\}^{\delta_2} \log\left\{-\log \left( 1 - \left(1- y_{1i} ^{a_1} \right) ^{\b_1} \right) \right\} \nonumber \\ &&+ \left\{-\log \left( 1 - \left(1- y_{2i} ^{a_2} \right) ^{\b_2} \right) \right\}^{\delta_2} \log\left\{-\log \left( 1 - \left(1- y_{2i} ^{a_2} \right) ^{\b_2} \right) \right\} \biggr]  + \sum_{\i=1}^{\ss} \frac{1}{z_i^{1/{\delta_2}} + {\delta_2}-1}\nonumber \\ &&  \sum_{\i=1}^{\ss} \frac{z^{1/{\delta_2}} \log(z)}{{\delta_2}^2} \biggl[ \left\{-\log \left( 1 - \left(1- y_{1i} ^{a_1} \right) ^{\b_1} \right) \right\}^{\delta_2} \log\left\{-\log \left( 1 - \left(1- y_{1i} ^{a_1} \right) ^{\b_1} \right) \right\} \nonumber \\ &&+ \left\{-\log \left( 1 - \left(1- y_{2i} ^{a_2} \right) ^{\b_2} \right) \right\}^{\delta_2} \log\left\{-\log \left( 1 - \left(1- y_{2i} ^{a_2} \right) ^{\b_2} \right) \right\} \biggr] \left(1 - \frac{1}{z_i^{1/{\delta_2}} + {\delta_2}-1} \right)
\end{eqnarray*}
where $z_i = \left[ \left\{-\log \left( 1 - \left(1- y_{1i} ^{a_1} \right) ^{\b_1} \right) \right\}^{\delta_2} + \left\{-\log \left( 1 - \left(1- y_{2i} ^{a_2} \right) ^{\b_2} \right) \right\}^{\delta_2} \right]$

\bibliographystyle{apalike}
\bibliography{test}

\end{document}